\documentclass[runningheads,envcountsame]{llncs}
\usepackage{graphicx}
\usepackage{proof}
\usepackage{xcolor}
\usepackage{amssymb}
\usepackage{amsmath}
\usepackage{stmaryrd} 
\usepackage{rotating}
\usepackage{mathtools}
\usepackage{colonequals}
\usepackage{bibnames}

\newcommand{\mypar}[1]{\smallskip\noindent{\bf #1}}

\newdimen\mydisplayskip
\newenvironment{smallequation*}
{\par\nobreak\vskip\mydisplayskip\noindent\bgroup\small\csname equation*\endcsname}{\csname endequation*\endcsname\egroup}

\newenvironment{smallalign*}
{\par\nobreak\noindent\bgroup\small\csname align*\endcsname}{\csname endalign*\endcsname\egroup}

\newcommand\mydots{\hbox to 1em{.\hss.\hss.}}

\definecolor{brickred}{rgb}{0.8, 0.25, 0.33}
\newcommand{\phl}[1]{{\color{brickred}{#1}}} 
\newcommand{\thl}[1]{{{#1}}} 
\newcommand{\ghl}[1]{{\phl{#1}}} 

\newcommand{\nohl}[1]{{\color{black}{#1}}}

\newcommand{\substone}[2]{\{#1/#2\}}
\newcommand{\substtwo}[4]{\{#1/#2, #3/#4\}}

\newcommand{\m}[1]{\mathsf{#1}}

\newcommand{\Select}[1][]{\textcolor{blue}{\mathcal
    S}_{\textcolor{blue}{#1}}}
\newcommand{\Left}{\Select[l]}
\newcommand{\Right}{\Select[r]}
\newcommand{\Query}{{\color{blue}{\mathcal Q}}}
\newcommand{\Star}{{\color{blue}{\br{\ast}}}}
\newcommand{\Symbol}{{\color{blue}{\mathcal{X}}}}

\newcommand{\br}[1]{[#1]}
\newcommand{\dbr}[1]{[\![#1]\!]}

\newcommand{\til}[1]{\tilde{#1}}
\newcommand{\prf}[1]{\mathcal{#1}}

\newcommand{\DD}{\prf D}
\newcommand{\EE}{\prf E}
\newcommand{\FF}{\prf F}
\newcommand{\GG}{\prf G}

\newcommand{\bang}{\mathord{!}}

\newcommand{\query}{\mathord{?}}

\newcommand{\sendbase}[2]{#1. #2}
\newcommand{\closebase}[1]{#1}
\newcommand{\srvbase}[2]{\bang #1.#2}

\newcommand{\Casebase}[2]{{#1}.\m{case}(#2)}
\newcommand{\wait}[2]{\phl{#1().#2}}

\newcommand{\close}[1]{\phl{\closebase{{#1}[]}}}

\newcommand{\send}[4]{\phl{\sendbase{{#1}[#2\triangleright #3]}{#4}}}
\newcommand{\sendm}[5]{}
\newcommand{\recv}[3]{\phl{{#1}(#2).#3}}

\newcommand{\Case}[3]{\phl{\Casebase{#1}{#2,#3}}}

\newcommand{\inl}[2]{\phl{{#1}[\m{inl}].#2}}
\newcommand{\inlr}[2]{{#1}[\m{in}\#].#2}

\newcommand{\inr}[2]{\phl{{#1}[\m{inr}].#2}}

\newcommand{\srv}[3]{\phl{\srvbase{#1(#2)}{#3}}}

\newcommand{\client}[3]{\phl{\query #1[#2].#3}}

\newcommand{\pp}{{\ \;\boldsymbol{|}\ \;}}

\newcommand{\res}[1]{(\boldsymbol\nu #1)\,}
\newcommand{\fwd}[2]{\phl{#1\leftrightarrow #2}}

\newcommand{\pairQ}[2] { \phl{#1}:\thl{#2}}

\newcommand{\triCtx}[3] {\Query\dbr{#1}\phl{#2}:\thl{#3}}
\newcommand{\ltriCtx}[3] {\Left\dbr{#1}\phl{#2}:\thl{#3}}
\newcommand{\rtriCtx}[3] {\Right\dbr{#1}\phl{#2}:\thl{#3}}
\newcommand{\quadQ}[4] { [\phl{#1}:\thl{#2}]\phl{#3}:\thl{#4}}

\newcommand{\cllseq}[2]{\phl{#1} \vdash_{\mathsf{CLL}} {#2}}

\newcommand{\syncseqcut}[2]{\phl{#1} \vdash_{\mbox{\tiny
      cut}} {#2}}
\newcommand{\syncseq}[2]{\phl{#1} \vdash {#2} }
\newcommand{\syncseqbis}[2]{\phl{#1} \vdash_{{p_1}} {#2}}
\newcommand{\syncseqtris}[2]{\phl{#1}
  \vdash_{
    {p_2}} {#2}} \newcommand{\cohseq}[2]{\phl{#1} \vDash {#2}}

\newcommand{\seq}{\vdash} \newcommand{\gseq}{\vDash}
 
\newcommand{\pair}[2]{{#1\!:}{#2}} 

\newcommand{\dual}[1]{#1^{\bot}} 
\newcommand{\rep}[1]{\ulcorner {#1} \urcorner} 
\newcommand{\coepp}[1]{\nohl{\{\![ }{\phl{#1}} \nohl{]\!\}}} 
\newcommand{\coextract}[1]{\nohl{\{\!\! \{\!\!\{ }{\phl{#1}} \nohl{\}\!\!\}\!\!\}}} 
\newcommand{\tensor}{\otimes}

\newcommand{\parr}{\mathbin{\bindnasrepma}}

\newcommand{\with}{\mathbin{\binampersand}}

\newcommand{\one}{\mathbf{1}}

\newcommand{\swapG}
{\simeq^{\m g}}

\newcommand{\tcost}{\mathbf{cost}}
\newcommand{\tname}{\mathbf{name}}

\newcommand{\taddr}{\mathbf{addr}}

\newcommand{\fn}{\m{fn}}
\newcommand{\defeq}{\stackrel{\mbox{\scriptsize def}}{=}}

\newcommand{\reducesto}{\longrightarrow}
\newcommand{\Reducesto}{\Longrightarrow}

\newcommand{\ba}{\begin{array}}
\newcommand{\ea}{\end{array}}

\newenvironment{equations}{\[\ba{@{}r@{~}c@{~}l@{}}}{\ea\]\ignorespacesafterend}

\newcommand{\gto}{\to}

\newcommand{\gfromto}[2]{#1 \gto #2}

\newcommand{\parrtensorCS}[3]{{\gfromto{#1}{#2}}(#3).}
\newcommand{\parrtensor}[4]{\gfromto{#1}{#2}(#3).#4}
\newcommand{\botone}[2]{\closebase{\gfromto{#1}{#2}}}
\newcommand{\withplus}[4]{\Case{\gfromto{#1}{#2}}{#3}{#4}}

\newcommand{\bangquery}[3]{\bang \gfromto{#1}{#2}(#3)}

\newcommand{\cpaxiom}{\localaxiom}
\newcommand{\localaxiom}[4]{\nohl{{#1} \gto {#3}^{#4}}}
\newcommand{\globalaxiom}[4]{\nohl{\fwd{{#1}} {#3}}}

\newcommand{\cpres}[4]{{\res{{#1}^{}{#3}}}}

\newcommand{\Did}[1]{\textsc{#1}}

\begin{document}
\title{Synchronous Forwarders
}
%
%
\author{Marco Carbone\inst{1}
\and
Sonia Marin \inst{2}
\and
Carsten Sch\" urmann\inst{1}
}
\authorrunning{M. Carbone et al.}
%
\institute{IT-University of Copenhagen, Denmark \\
\email{\{carbonem,carsten\}@itu.dk}	
\and
University College London, UK \\
\email{s.marin@ucl.ac.uk}\\
}
\maketitle              
\begin{abstract}
  Session types are types for specifying protocols that processes must
  follow when communicating with each other. Session types are in a
  propositions-as-types correspondence with linear logic. Previous
  work has shown that a multiparty session type, a generalisation of
  session types to protocols of two or more parties, can be modelled
  as a proof of coherence, a generalisation of linear logic
  duality. And, protocols expressed as coherence can be simulated by
  arbiters, processes that act as a middleware by forwarding messages
  according to the given protocol.

  In this paper, we generalise the concept of arbiter to that of
  synchronous forwarder, that is a processes that implements the
  behaviour of an arbiter in several different ways. In a
  propositions-as-types fashion, synchronous forwarders form a logic
  equipped with cut elimination which is a special restriction of
  classical linear logic. Our main result shows that synchronous
  forwarders are a characterisation of coherence, i.e., coherence
  proofs can be transformed into synchronous forwarders and,
  viceversa, every synchronous forwarder corresponds to a coherence
  proofs.
\end{abstract}

\section{Introduction}\label{sec:intro}
A concurrent system is more than a sum of processes. It also includes
the fabric that determines how processes are tied together.
%
Session types, originally proposed by Honda et al.~\cite{HVK98}, are
type annotations that ascribe protocols to processes in a concurrent
system and, as a fabric, determine how they behave when communicating
with each other. Such type annotations are useful for various
reasons. First, they serve as communication blueprints for the entire
system and give programmers clear guidance on how to implement
communication patterns at each endpoint (process or service).  Second,
they make implementations of concurrent systems safer, since
well-typedness entails basic safety properties of programs such as
{\em lack of communication errors} (``if the protocol says I should
receive an integer, I will never receive a boolean''), {\em session
  fidelity} (``my programs follow the protocol specification
patterns''), and {\em in-session deadlock freedom} (``the system never
gets stuck by running a protocol'').
%
Intuitively, they make sure that the processes are {\em compatible}
and that they exchange messages in the prescribed way for the
concurrent system to work correctly.
For example, by preventing messages from being duplicated, as
superfluous messages would not be accounted for, and by preventing
messages from getting lost, otherwise a process might get stuck,
awaiting a message.

In the case of \emph{binary sessions types}, the version of session
types that deals only with protocols between two
parties, 
compatibility means for 
type annotations to be dual to one another: the send action of one
party must be matched by a corresponding receive action of the other
party, 
and vice versa.  Curiously, binary session types find their logical
foundations in linear logic, as identified by Caires and
Pfenning~\cite{CP10,CP16} and later also by Wadler~\cite{W12,W14}.
They have shown that session types correspond to linear logic
propositions, processes to proofs, reductions in the operational
semantics to cut reductions in linear logic proofs, and compatibility
to the logical notion of duality for linear formulas.  Duality, thus,
defines the fabric for how two processes communicate in an idealised
world, while at the same time abstracting away from all practical
details, such as message delay, message order, or message buffering.

The situation is not as direct for \emph{multiparty session
  types}~\cite{HYC08,HYC16}, 
type annotations for protocols with
two or more participants.
Carbone et al.~\cite{CMSY15,CLMSW16} extended Wadler's embedding of
binary session types into classical linear logic (CLL) to the
multiparty setting, by generalising duality to the notion of {\em
  coherence}.
They observed that the in-between fabric holds the very key to
understanding multiparty session types:
when forcing the type annotations to be {\em coherent}, one ensures
that sent messages will eventually be collected.
%
Coherence as a deductive system allows one to derive exactly these
compatible jugements, while proofs correspond precisely to multiparty
protocol specifications. 
A key result is that coherence proofs can be encoded as well-typed (as
proofs in CLL) processes, called {\em arbiters}, which means that the
fabric can actually be formally seen as a process-in-the-middle.
However, no precise logical characterisation of what constitutes
arbiters was given.
In this paper, we continue this line of research and define a
subsystem of processes, called \emph{synchronous forwarders}, that
provides one possible such characterisation that guarantees coherence.
%

As the name already suggests, a forwarder is a process that forwards
messages, choices, and services from one endpoint to another according
to the protocol specification.
Intuitively, similarly to an arbiter, a forwarder process mimics the
fabric by capturing the message flow.  However, when data-dependencies
allow, forwarders could, in theory, non-deterministically choose to
receive messages from different endpoints, and then forward such
messages at a later point. Or, they can also decide to buffer a
certain number of messages from a given receiver.  In any case,
they
retransmit messages only after receiving them,
without interpreting, modifying, or computing with them.

In this work, synchronous forwarders support buffers of size 1 -- only
one message can be stored for each endpoint at a given time. This
preserves the order of messages from the same sender, i.e., after
receiving a message from one party, the forwarder blocks the
connection to such party until the message has been delivered to its
destination.
%
Synchronous
forwarders could be used to explain communication patterns as they
occur in practice, such as message routing, proxy services, and
runtime monitors for message flows.

The meta-theoretic study of synchronous forwarders allows us to
conclude that there is a proof-as-processes correspondence between
synchronous forwarders and coherence proofs. We show that synchronous
forwarders can be safely composed through cut elimination, which
allows us to combine the fabric between two concurrent systems
(figuring arbitrary many processes). With respect to coherence, we
prove their {\em soundness}, namely that any coherence proof can be
emulated by a synchronous forwarder simulating the actions of the
fabric, and their {\em completeness}, meaning that every synchronous
forwarder guarantees coherence. In particular, arbiters are special
instances of forwarders.


\smallskip

\mypar{Outline and key contributions.}  The key contributions of this
paper include
\begin{itemize}
\item a logical characterisation of \emph{synchronous forwarders}
  (\S~\ref{sec:sync_system});
\item a reductive operational semantics based on cut-elimination
  (\S~\ref{sec:cut-elim});
\item a tight correspondence between coherence and synchronous
  forwarders (\S~\ref{sec:coherencesf}).
\end{itemize}
Additionally, \S~\ref{sec:prelim} recaps the definitions of coherence,
processes, and arbiters, while \S~\ref{sec:related} discusses related
and future work. Concluding remarks are in
\S~\ref{sec:conclusions}.
\section{Coherence, Processes, and Arbiters}\label{sec:prelim}
To prepare the ground for formally defining synchronous forwarders, we
use this section to introduce the basic ingredients: the formal notion
of {\em coherence}~\cite{CLMSW16}, the syntax of {\em
  processes}~\cite{W12,W14} used for synchronous forwarders, and the
encoding of coherence proofs into special processes known as {\em
  arbiters}~\cite{CLMSW16}.


\subsection{Coherence}
\emph{Coherence}~\cite{HYC16,CMSY15,CLMSW16}, which provides the
formal foundation for our results,
is a generalisation of the notion of {\em duality} from
CLL~\cite{G87}. 
Duality is used by binary session types when composing (two) processes
through a communication channel: the two ends of the channel are
compatible whenever their types are {\em dual} to each other, i.e.,
every output  is matched by an input and viceversa. 
%
Coherence can be understood as a generalisation of duality, that is, a
criterion that decides if two or more parties can agree on who sends
what to whom. Clearly, as there can be more than two parties, it is
impossible to base this criterion on duality alone.
Intuitively, we say that a set of processes are coherent, if each send
can be linked to an available receive from another party.


\begin{example}[The 2-Buyer Protocol]\label{ex:2buyer}
  As a running example throughout this paper, we use the classic
  \emph{2-buyer protocol}~\cite{HYC08,HYC16}, where two buyers
  intend to buy a book jointly from a seller.
  The first buyer
  sends the title of the book 
  to the seller,
  who, in turn, sends a quote to both buyers.
  Then, 
  the first buyer decides how much she wishes to contribute and
  informs the second buyer,
  who either pays the rest or cancels the transaction by informing the
  seller.

  The three participants 
  are connected through endpoints $\phl{b_1}$, $\phl{b_2}$, and
  $\phl s$ respectively. Each endpoint must be used according to its
  respective session type annotation, expressed as a CLL proposition:
  \begin{smallequation*}\label{eq:cotypes}
      \begin{array}{c}
        {\phl{b_1}} : \tname \tensor \tcost^\perp \parr \tcost \tensor
        \one
        \quad
        \phl{b_2} : \tcost^\perp \parr \tcost^\perp \parr ((\taddr \tensor \one) \oplus \one)
        \\[1mm]
        \phl{s} : \tname^\perp \parr \tcost \tensor \tcost
        \tensor ((\taddr^\perp \parr \bot) \with \bot)
      \end{array}
  \end{smallequation*}
  The typing above gives a precise description of how each endpoint
  has to act. For example,
  $\tname \tensor \tcost^\perp \parr \tcost \tensor \one$ says that
  buyer $\phl{b_1}$ must first send a value of type $\tname$ (the
  title of the book), then receive a value of type $\tcost$ (the price
  of the book), then send a value of type
  $\tcost$ (the amount of money she wishes to contribute), and
  finally terminate.

  Coherence will determine whether the three endpoints above are
  compatible by establishing for each output which endpoint should
  receive it. 
  For example, coherence will say that the first output of type
  $\tname$ at endpoint $\phl{b_1}$ must be received by the input of
  type $\tname^\perp$ at $\phl s$. Then, the output from $\phl s$
  of type $\tcost$ can be paired with the input of type $\tcost^\perp$
  at either $\phl{b_1}$ and $\phl{b_2}$. And so on, precisely describing
  what the execution of the 2-buyer protocol should be.  \qed
\end{example}

\mypar{Types}. Following the propositions-as-types approach, {\em
  types}, taken to be propositions (formulas) of CLL, are associated
to names, denoting the way an endpoint must be used at runtime. Their
formal syntax is given as:
\begin{smallequation*}
  \thl{A,B} \coloncolonequals  \quad
                            \thl{a}
                            \: \mid\: \thl{a^\perp}
                            \: \mid\: \thl{1}
                            \: \mid\: \thl{\bot}
                            \: \mid\: \thl{(A\tensor B)}
                            \: \mid\: \thl{(A\parr B)}
                            \: \mid\: \thl{(A\oplus B)}
                            \: \mid\: \thl{(A\with B)}
                            \: \mid\: \thl{!A}
                            \: \mid\: \thl{?A} 
\end{smallequation*}
\noindent We briefly comment on their interpretation. There is a predefined,
finite set of atoms $a$ and their duals $a^\perp$, e.g., $\tname$ and
$\tname^\perp$.
Types $\one$ and
$\perp$ denote an endpoint that must be closed by a last
synchronisation. A type $A \tensor
B$ is assigned to an endpoint that outputs a message of type
$A$ and then is used as $B$. Similarly, an endpoint of type $A\parr
B$, receives a message of type $A$ and then continues as
$B$. Types $A\oplus B$ and $A\with
B$ denote branching. The former is the type of an endpoint that may
select to go left or right and continues as $A$ or
$B$, respectively. The latter is the type of an endpoint that offers
two choices (left or right) and then, based on such choice, continues
as $A$ or $B$. Finally,
$!A$ types an endpoint offering a service of type $A$, while
$?A$ types an endpoint of a client invoking some service and behaving
as
$A$. Operators can be grouped in pairs of duals that reflect the
input-output duality. As a consequence, standard duality
$(\cdot)^\perp$ on types is inductively defined as:
\begin{smallequation*}
  \dual{(\dual{a})} = a
  \quad \dual\one=\perp
  \quad \dual{(A\tensor B)} = \dual A\parr\dual B
  \quad \dual {(A\oplus B)}=\dual A\with\dual B
  \quad \dual{(!A)} = ?\dual A
\end{smallequation*}
In the remainder, for any binary operator
$\oslash,\odot\in\{\tensor,\parr,\oplus,\with\}$, we interpret
$A\oslash B \odot C$ as $A \oslash (B \odot C)$. Moreover, both
$?$ and $!$ have higher priority than $\oslash$.


\mypar{Global Types}. Following the standard multiparty session types
approach~\cite{HYC08,HYC16}, we introduce the syntax of {\em global
types}~\cite{CLMSW16}, the language for expressing the protocols that
processes must follow when communicating. This is done by the
following grammar, starting with a set  $\phl{\mathcal N = \{x,y,z,\ldots\}}$ 
of names, and used for denoting communication endpoints.\footnote{
A list of endpoints $\phl{(x_1\ldots x_n)}$ can be abbreviated as $\phl{\til x}$.}
\def\bnfas{\mathrel{::=}} \def\bnfalt{\mid}
\begin{smallequation*}
  \begin{array}{rrl}
    \phl G, \phl H &\bnfas & \
                             \phl{\botone{\til x}{y}}
                             \ \bnfalt\ \phl{\parrtensor{\til x}{y}{G}{H}}
                             \ \bnfalt\ \phl{\withplus{x}{\til y}{G}{H}} 
                   \ \bnfalt  \phl{\bangquery{x}{\til y}{G}}
                                \ \bnfalt\ \phl{\globalaxiom{x}{A}{y}{\dual{A}}}
  \end{array}
\end{smallequation*}
A global type is an ``Alice-Bob''-notation for expressing the sequence
of interactions that a session (protocol) must follow. The term
$\phl{\botone{\til x}{y}}$ expresses a synchronisation between
endpoints $\phl{x_1\ldots x_n}$ and endpoint $\phl y$, which gathers
all the synchronisation messages from $\phl{\til x}$. Global type
$\phl{\parrtensor{\til x}{y}{G}{H}}$ is also a gathering operation
between $\phl{\til x}$ and $\phl y$, but the processes communicating
over endpoints $\phl{\til x}$ and $\phl{y}$ spawn together a new
session with global type $\phl G$ before continuing the current one
with protocol $\phl H$.  In the term
$\phl{\withplus{x}{\til y}{G}{H}}$, endpoint $\phl x$ broadcasts a
choice to endpoints $\phl{\til y}$: as a result, the session will
follow either protocol $\phl G$ or $\phl H$. The protocol expressed by
$\phl{\bangquery{x}{\til y}{G}}$ denotes the composition of a process
providing a service that must be invoked through $\phl{\til y}$. The
term $\phl{\globalaxiom{x}{A}{y}{\dual{A}}}$ is just a simple
forwarder; it connects two endpoints related via standard duality.

\mypar{Contexts and Judgements.} {\em Coherence}, denoted by $\gseq$,
is defined by the judgement $\phl G\gseq\Delta$.  The context $\Delta$
contains the types of the endpoints we wish to compose.  A coherence
proof is used to establish if the types in $\Delta$ are compatible,
i.e., when the corresponding processes are composed, they will
interact (according to protocol $\phl G$) without raising an error.
Formally, \emph{basic contexts} are sets $\Delta$ of propositions
labelled by endpoints:
$ \Delta \Coloneqq \cdot\ \mid\ \Delta, \pairQ xA $,
where $\pairQ xA$ states that endpoint $\phl x$ has type $A$. Each
endpoint occurs at most once.\footnote{ We use $\query\Delta$ as
  shorthand for any
	$\phl{x_1}:\thl{\query A_1}, \ldots, \phl{x_n}:\thl{\query A_n}$
	and 
	use $\Delta$ 
	for the corresponding set
	$\phl{x_1}:\thl{A_1}, \ldots, \phl{x_n}:\thl{A_n}$.  }

      Global types are the proof terms for coherence, which is defined
      in Fig.~\ref{fig:coherence}.
\begin{figure}[t]
  \begin{displaymath}\small
    \begin{array}{c}
      \infer[\textsc{Axiom}]
      {\phl{\globalaxiom{x}{A}{y}{\dual{A}}} \gseq \phl x:\thl A, \phl y:\thl{\dual{A}}}
      {}
      \qquad
    
      \infer[\one\bot]
      {\phl{\botone{\til x}{y}}
      \gseq \{\phl {x_i}:\thl{\one}\}_i, \phl{y}:\thl{\bot}}
      { }
      \\[2ex]
    
      \infer[\tensor\parr]
      {\phl{\parrtensor{\til x}{y}{G}{H}}
      \gseq \thl \Delta,  \{\phl{x_i} : \thl{A_i \tensor B_i}\}_i, \phl{y}:\thl{C \parr D}}
      { \phl {G} \gseq \thl \{\phl{x_i}:\thl{A_i}\}_i, \phl y:\thl C
      & \phl {H} \gseq \thl {\Delta}, \{\phl{x_i}:\thl{B_i}\}_i, \phl y:\thl D}
      \\[1ex]
    
      \infer[\oplus\with]
      {\phl{\withplus{x}{\til y}{G}{H}}
      \gseq \thl{\Delta}, \phl x:\thl{A \oplus B}, \{\phl{y_i}:\thl{C_i \with D_i}\}_i}
      { \phl G  \gseq  \Delta, \phl x:\thl A, \{\phl {y_i}:\thl{C_i}\}_i
      & \phl H \gseq  \Delta, \phl x:\thl B, \{\phl {y_i}:\thl{D_i}\}_i}
      \\[1ex]
    
      \infer[\query\bang]
      {\phl{\bangquery{x}{\til y}{G}} \gseq \phl x:\thl{\query A}, \{\phl{y_i}:\thl{\bang B_i}\}_i}
      {\phl G \gseq \phl x:\thl A, \{\phl{y_i}:\thl{B_i}\}_i}
    \end{array}
  \end{displaymath}
  \caption{Coherence} \label{fig:coherence}
\end{figure}
The rules reflect the concept of a generalisation of
duality~\cite{CLMSW16}: dual operators can be matched at different
endpoints, establishing how communications must be done between
them. 
The axiom rule is identical to that of CLL. Rule $\one\bot$ says that
many processes willing to close a session 
are compatible with a single process waiting for them; as a global
type we have the term $\ghl{\botone{\til x}{y}}$.
For multiplicatives, rule $\tensor\parr$ implements gathering: many
processes are outputting an endpoint of type $A_i$, and another
process is gathering such endpoints and establishing a new session;
the corresponding global type is $\ghl{\parrtensor{\til
    x}{y}{G}{H}}$. Rule $\oplus\with$ is for branching where the
global type $\ghl{\withplus{x}{\til y}{G}{H}}$ indicates that a
process with endpoint $\phl x$ decides which branch to take and
communicates that to $\phl{\tilde y}$. Finally, in the exponential rule $!?$
some client with endpoint $\phl x$ invokes many services at $\phl{\til y}$.

\begin{example} \label{ex:2bu} Returning to Example~\ref{ex:2buyer},
  we can formally attest that the three endpoints $\phl{b_1}$,
  $\phl{b_2}$, and $\phl{s}$ are coherent and can be
  composed.
  %
  Overall, a coherence proof can be summarized as a global
  type.
  In our example:
  \begin{smallequation*}
    \begin{array}{l}
      \phl{ \parrtensorCS{\phl{b_1}}{\phl{s}}{
      \fwd{n}{n'}}}\quad
     \phl{ \parrtensorCS{\phl{s}}{\phl{b_1}}{
      \fwd{x_1}{x_1'}}}\quad
     \phl{ \parrtensorCS{\phl{s}}{\phl{b_2}}{
      \fwd{x_2}{x_2'}}}\quad
      \phl{\parrtensorCS{\phl{b_1}}{\phl{b_2}}{\fwd{y}{y'}}} \\
      \qquad \phl{\gfromto{\phl{b_2}}{\phl{s}}.\phl{\m{case}}\quad
      (\ \phl{\parrtensorCS{\phl{b_2}}{\phl{s}}{\fwd{a}{a'}}
      \botone{(\phl{b_1},\phl{b_2})}{\phl{s}}},\qquad
      \botone{(\phl{b_1},\phl{b_2})}{\phl{s}})} 
    \end{array}
  \end{smallequation*}
  %
  The global type above corresponds to the following derivation for
  $\gseq$:
  \begin{displaymath}\small
    \infer[\tensor\parr]
    {
      \phl{\parrtensorCS{\phl{b_1}}{\phl{s}}{\fwd{b_1}{s}}G_1}
      \gseq
      {
        \begin{array}{l}
          \pairQ{b_1} {\tname \tensor \tcost^\perp \parr \tcost \tensor\one},\\
          \pairQ{b_2} {\tcost^\perp \parr \tcost^\perp \parr ((\taddr \tensor \one) \oplus \one)},\\
          \pairQ{s}   {\tname^\perp \parr \tcost \tensor \tcost
          \tensor ((\taddr^\perp \parr \bot) \with \bot)}
        \end{array}
      }
    }
    {
      \infer[\tensor\parr]
      {
        \phl {G_1 =
          \parrtensorCS{\phl{s}}{\phl{b_1}}{\fwd{s}{b_1}}G_2}
        \gseq
        {
          \begin{array}{l}
            \pairQ{b_1} {\tcost^\perp \parr \tcost \tensor\one},\\
            \pairQ{b_2} {\tcost^\perp \parr \tcost^\perp \parr ((\taddr \tensor \one) \oplus \one)},\\
            \pairQ{s}   {\tcost \tensor \tcost
            \tensor ((\taddr^\perp \parr \bot) \with \bot)}
          \end{array}
        }
      }
      {
        \infer[\tensor\parr]
        {
          \phl {G_2 =
            \parrtensorCS{\phl{s}}{\phl{b_2}}{\fwd{s}{b_2}}G_3}
          \gseq
          {
          \begin{array}{l}
            \pairQ{b_1\!} {\!\tcost \tensor\one},
            \pairQ{b_2\!} {\!\tcost^\perp \!\parr\! \tcost^\perp \!\parr\! ((\taddr \tensor \one) \oplus \one)},\\
            \pairQ{s\!}   {\!\tcost
            \tensor ((\taddr^\perp \parr \bot) \with \bot)}
          \end{array}
        }
      }
      {
        \infer[\tensor\parr]
        {
          \phl {G_3 =
            \parrtensorCS{\phl{b_1}}{\phl{b_2}}{\fwd{b_1}{b_2}}\ G_4}
          \gseq
          {
            \begin{array}{l}
              \pairQ{b_1} {\tcost \tensor\one},
              \pairQ{b_2} {\tcost^\perp \parr ((\taddr \tensor \one) \oplus \one)},\\
              \pairQ{s}   {((\taddr^\perp \parr \bot) \with \bot)}
            \end{array}
          }
        }
        {
          \infer[\oplus\with]
          {
            \phl {G_4\! =\!
              \gfromto{\phl{b_2}}{\phl{s}}.\phl{\m{case}} (G_5,\botone{(\phl{b_1},\phl{b_2})}{\phl{s}})}
            \gseq
            {
              \begin{array}{l}
                \pairQ{b_1\!} {\!\one},
                \pairQ{b_2\!} {\!(\taddr \tensor \one) \oplus \one},
                \pairQ{s\!}   {\!(\taddr^\perp \parr \bot) \with \bot}
              \end{array}
            }
          }
          {
            \infer[\tensor\parr]
            {
              \phl {
                G_5 = 
                \parrtensor{\phl{b_2}}{\phl{s}}{\fwd{b_2}{s}}
                G_6
              }
              \gseq
              {
                \begin{array}{l}
                  \pairQ{b_1\!} {\!\one},\\
                  \pairQ{b_2\!} {\!\taddr \tensor \one},\\
                  \pairQ{s\!}   {\!\taddr^\perp \parr \bot}
                \end{array}
              }
            }
            {
              \infer[\one\!\perp]
              {
                \phl {
                  G_6 = \botone{(\phl{b_1},\phl{b_2})}{\phl{s}}
                }
                \gseq
                {
                  \begin{array}{l}
                    \pairQ{b_1} {\one},
                    \pairQ{b_2} {\one},
                    \pairQ{s}   {\bot}
                  \end{array}
                }
              }
              {
              }
            }
            &
            \infer[\one\!\perp]
            {
              \phl {
                \botone{(\phl{b_1},\phl{b_2})}{\phl{s}}                
              }
              \gseq
              {
                \begin{array}{l}
                  \pairQ{b_1} {\one},
                  \pairQ{b_2} {\one},
                  \pairQ{s}   {\bot}
                \end{array}
              }
            }
            {
            }
          }
        }
      }
    }
  }
\end{displaymath} 
For clarity, we elided the left premisses in the applications of
$\tensor\parr$ since, in this derivation, they are axiomatic, e.g.,
$\phl {\fwd {b_1}{s}} \gseq \pairQ {b_1}{\tname}, \pairQ
{s}{\tname^\perp}$.
\qed
\end{example}

\subsection{Process Terms}\label{subsec:processes}
We use a language of {\em processes} to represent the forwarders that
decide to whom messages, choices, and services should be delivered,
expressing the communications enforced by coherence.  For that, we
introduce a standard process language which is a variant of the
$\pi$-calculus~\cite{MPW92} with specific communication primitives as
usually done for session calculi. Moreover, given that the theory of
this paper is based on the proposition-as-types correspondence with
CLL, we adopt a syntax akin to that of Wadler~\cite{W12,W14}: 
%
\begin{displaymath}
  \small
  \begin{array}{l@{\quad}l@{\quad}l@{\qquad}l@{\quad}l@{\quad}l@{\qquad}l}
    \phl P,\phl Q ::= & \fwd xy          & \text{(link)}
    &
      \phl{\res {xy} (P\!\!\pp\!\! Q)} & \text{(parallel)}
    \\
                      & \wait x P & \text{(wait)}
    &
      \close x & \text{(close)}
    \\
                      & \recv xyP        & \text{(input)}
                
    &
      \send xyPQ       & \text{(output)}
    \\
            & \Case xPQ        & \text{(choice)}
    & \inl xP          & \text{(left select)} \\
& & & \inr xP & \text{(right select)}
    \\
            & \client xyP & \text{(client request)}
    &
      \srv xyP & \text{(server accept)}
  
  \end{array}
\end{displaymath}
%
We briefly comment on the various production of the grammar above. A link
$\fwd xy$ is a binary forwarder, i.e., a process that forwards any
communication between endpoint $\phl x$ and endpoint $\phl y$. This
yields a sort of equality relation on names: it says that endpoints
$\phl x$ and $\phl y$ are equivalent, and communicating something over
$\phl x$ is like communicating it over $\phl y$.
The term $\phl{\res {xy}(P\pp Q)}$ is used for composing processes: it
is the parallel composition of $\phl P$ and $\phl Q$ that share a
private connection through endpoints $\phl x$ and $\phl y$.
%
Parallel composition formally define compositionality of processes and
will play a key role in deriving the semantics of synchronous
forwarders through logic (cf. \S~\ref{sec:cut-elim}).
Note that we use endpoints instead of channels~\cite{V12}. The
difference is subtle: the restriction $\phl{\res {xy}}$ connects the
two endpoints $\phl x$ and $\phl y$, instead of referring to the
channel between them.
The terms $\wait xP$ and $\close x$ handle synchronisation (no message
passing); $\wait xP$ can be seen as an empty input on $\phl x$, while
$\close x$ terminates the execution of the process.
The term $\send xyPQ$ denotes a process that creates a fresh name
$\phl y$, spawns a new process $\phl P$, and then continues as
$\phl Q$. The intuition behind this communication operation is that
$\phl P$ uses $\phl y$ as an interface for dealing with the
continuation of the dual primitive (denoted by term $\recv xyR$, for
some $\phl R$).
We observe that Wadler~\cite{W12,W14} uses the syntax $x[y].(P\pp Q)$,
but we believe that our version is more intuitive and gives a better
explanation of why we require two different processes to follow after
an output. However, our format is partially more restrictive, since
$\phl y$ is forced to be bound in $\phl P$ (which Wadler enforces with
typing). Also, note that output messages are always fresh, as for the
internal $\pi$-calculus~\cite{S96}, hence the output term $\send xyPQ$
is a compact version of the $\pi$-calculus term
$(\nu y)\, \overline x y.(P\pp Q)$.
%
Branching computations are handled by $\Case xPQ$, $\inl xP$ and
$\inr xP$. The former denotes a process offering two options (external
choice) from which some other process can make a selection with
$\inl xP$ or $\inr xP$ (internal choice).
Finally, the term $\srv xyP$ denotes a persistently available service
that can be invoked by $\client xzQ$ at $\phl x$ which will spawn a
new session to be handled by a copy of process $\phl P$.

As shown by Wadler~\cite{W12,W14}, among all of the many process
expressions one can write, CLL characterises the subset that is
well-behaved, i.e.\ they satisfy deadlock freedom and session
fidelity, and are therefore interesting for our work.

\subsection{Mapping Coherence into Arbiters} Carbone et
al.~\cite{CLMSW16} show that any coherence proof $\phl G\gseq\Delta$
can be transformed into a corresponding CLL proof of
$\cllseq P\Delta^\perp$, where $\phl P$ is called the {\em arbiter}.
Whenever there is an output on some endpoint $\phl x$, then the
arbiter can input such message and then forward it to the receiver
specified by coherence. For example, given a global type
$\phl{\parrtensor{x}{y}{G}{H}}$, we can build the arbiter
$\recv{x'}{u}{} \send{y'}{v}{P}{Q}$ (for some fresh $\phl u$ and
$\phl v$) which inputs from $\phl{x'}$ (binary endpoint connected to
$\phl x$) and outputs on $\phl{y'}$ (binary endpoint connected to
$\phl y$), where inductively $\phl{P}$ and $\phl{Q}$ are the arbiters
corresponding resp.~to $\phl{G}$ and $\phl{H}$.
The translation of coherence proofs (hereby expressed as global types)
into processes is reported in Fig.~\ref{fig:coepp}.
%
\begin{figure}[t]
	\begin{equations}
		\coepp{\globalaxiom{x}{A}{y}{\dual A}}
		&\ \defeq\ &
		\phl{\globalaxiom{x'}{\dual A}{y'}{A}}
		\\
		\coepp{\botone{\til x}{y}}
		&\ \defeq\ &
		\phl{\wait{x'_1}{\cdots \wait{x'_n}{\close{y'}}}}
		\\
		\coepp{\parrtensor{\til x}{y}{G}{H}}
		&\ \defeq\ &
		\recv{x'_1}{u_1}{\cdots
			\recv{x'_n}{u_n}{\send{y'}{v}{\coepp{G}\substtwo{\til u}{\til
						x'}{v}{y'}}{\coepp{H}}}}
		%
		\\
		\coepp{\withplus{x}{\til y}{G}{H}}
		&\ \defeq\ &
		\phl{\Case{x}{\inl{y_1}{\cdots \inl{y_n}{\coepp{G}}}}
			{\inr{y_1}{\cdots \inr{y_n}{\coepp{H}}}}}
		\\
		\coepp{\bangquery{x}{\til y}{G}}
		&\ \defeq\ &
		\phl{\srv{x}{u}{\client{y_1}{v_1}{
					\cdots \client{y_n}{v_n}{\coepp{G}\substtwo{u}{x}{\til
							v}{\til y}}}}}
		\\ [1ex]
		&&\hfill\text{where } \phl{u}, \phl{v}, \phl{\til u}, \phl{\til v}, \phl{x'}, \phl{y'}, \phl{\til x'} \text{and }\phl{\til y'}  \text{ are fresh} 
	\end{equations}
	\caption{Translation of Global Types into
          Processes~\cite{CLMSW16}}
	\label{fig:coepp}
\end{figure} 
\begin{example}[2-Buyer Protocol]\label{ex:2buf}
  Recall the 2-buyer protocol from Example~\ref{ex:2buyer}, where two
  buyers try to make a joint decision whether to buy a book from a
  seller.
  The interactions enforced by the global type/coherence proof seen in
  Example~\ref{ex:2bu} can be expressed by the arbiter process
  $\phl {P_1}$ below:
  \begin{smallequation*}
    \begin{array}{lllll}
      \phl{P_1} &=& \begin{array}[t]{@{}l}
                      \recv {b_1'}{n} {\send{s'}{n'}{\fwd{n}{n'}}{}}\  
                      \recv {s'}{x_1} {\send{b_1'}{x'_1}{\fwd{x_1}{x'_1}}{}}\
                      \recv {s'}{x_2} {\send{b_2'}{x'_2}{\fwd{x_2}{x'_2}}{}} \\
                      \quad\recv {b_1'}{y} {\send{b_2'}{y'}{\fwd{y}{y'}}{}} \
                      \Case {b_2'}{\ \inl{s'}{\recv {b_2'}{a}
                      {\send{s'}{a'}{\fwd{a}{a'}}Q_1\ }}}{\
                      \inr{s'}Q_1\ } 
                    \end{array} 
    \end{array}
  \end{smallequation*}
  Above, $\phl{b_1'}$, $\phl{b_2'}$, and $\phl {s'}$ are the endpoints
  connecting the forwarder $\phl{P_1}$ to the endpoints of the two
  buyers and the seller (resp. $\phl{b_1}$, $\phl{b_2}$, and
  $\phl s$). The process
  $\phl{Q_1} = \wait{b_1'}{\wait{b_2'}{\close{s'}}}$ closes all
  endpoints. \qed
\end{example}

We stress that an arbiter is nothing but a {\em forwarder}, transmitting
messages between the composed peers.  However, there are processes
that are not exactly arbiters (not in the image of the translation in
Fig.~\ref{fig:coepp}), but still forwarders and typable in CLL. For example,
consider the global type
$\phl{\parrtensor{x}{y}{G}{\parrtensor zy{G'}H}}$, the
corresponding arbiter has the form
$\recv{x'}{u}{} \send{y'}{v}{P}{} \recv {z'}{s}{}
\send{y'}{t}{Q}{R}$. However, another well-typed  process enforcing
the protocol could be the process
$\recv{x'}{u}{} \recv {z'}{s}{} \send{y'}{v}{P}{} \send{y'}{t}{Q}{R}$
or the process
$\recv {z'}{s}{} \recv{x'}{u}{} \send{y'}{v}{P}{}
\send{y'}{t}{Q}{R}$. In the next section, we define the class of
synchronous forwarders which give a much larger class of processes
that still correspond to coherence.

\section{Synchronous Forwarders}\label{sec:sync_system}
The goal of this section is to present a type system that captures precisely a
set of forwarder processes generalising arbiters. 
Following a proposition-as-types approach,
we aim at a restriction of CLL whose type contexts are provable by coherence.  
Intuitively, we want such restriction to guarantee that
messages cannot be opened/used, a received message is always
forwarded, forwarded messages have always been previously received,
and, the order of messages is preserved between any two endpoints.
The latter is a key for abiding to coherence which precisely enforces
in which order messages should be exchanged.
%

In this paper, we restrict our focus to a class of processes that is
also synchronous, i.e., a forwarder is ready to receive a message on
some endpoint $\phl x$ only after any previous message from $\phl x$
has been forwarded.  This corresponds to thinking of a synchronous
forwarder as a network queue of size one.
The technical device used to enforce this behaviour is a one-size buffer 
for each endpoint --- while a buffer is full, the forwarder is blocked on that particular
endpoint, and can only be unblocked by forwarding the message.
%



\mypar{Contexts.} To capture this one-size buffer mechanism, we need
to introduce new notation.
We continue to write $\pairQ{x}{A}$ for the typing of an unblocked
endpoint, which we also call an \emph{active} endpoint, but, when
typing a one-size buffer, we write $\quadQ{y}{B}{x}{A}$.
Here $\pairQ{x}{A}$ refers to the blocked
endpoint and $\pairQ{y}{B}$ to the message yet to be forwarded. 
The notion is also adapted to branching and exponentials.
	%
%
In summary, in this logic, contexts are formed as follows:
\[
  \begin{array}{rrllllllll}
    {\Gamma} & ::=   & \Delta \ & \quad\mid\quad \Gamma, \Star 
    & \quad\mid\quad   \ \Gamma, \quadQ{y}{B}{x}{A}
    & \quad\mid\quad  \Gamma, \ltriCtx{\Delta}{x}{A}\\
             && & & \quad\mid\quad \Gamma, \rtriCtx{\Delta}{x}{A}
    & \quad\mid\quad \Gamma, \triCtx{\Delta}{x}{A}
  \end{array}
\]
$\Delta$ is the context we defined in the previous section for
coherence, i.e., $\Delta ::= \cdot\ \mid\ \Delta, \phl x:\thl{A}$.
Intuitively, the extra ingredients in a $\Gamma$ context are used for
bookkeeping messages in-transit. And, we do that by {\em boxing}. For
example, the element $\quadQ yBxA$ expresses that some name $\phl y$
of type $B$ has been received from endpoint $\phl x$ and must be later
forwarded. Until that is done, endpoint $\phl x$ (that has type $A$)
is frozen.  Similarly, $\Star$, $\ltriCtx{\Delta}{x}{A}$ and
$\rtriCtx{\Delta}{x}{A}$, and $\triCtx{\Delta}{x}{A}$ indicate that a
request for closing a session, a branching request, or server
invocation, respectively, has been received and must be forwarded. In
the case of branching (additives) and servers (exponentials), the
context $\Delta$ contains the endpoints we must forward to.
%

In synchronous forwarders, we can have several occurrences of each
$\Left/\Right/\Query$ in a context, since boxes can occur multiple
times, containing potentially different typed endpoints.
However, this is not the case for $\Star$: it only acts as a flag,
i.e., multiple occurrences are automatically contracted to a single
one.

\mypar{Notation.} In the sequel, we silently use equivalences
$\Symbol\dbr{\cdot} \equiv \cdot$ for any
$\Symbol \in \{\Left, \Right, \Query\}$.
Also, we write $\Select[\#]$ for an unspecified $\Left$ or $\Right$.
Out of convenience, we write $\br{\Delta_1}\Delta_2$ for
$\Delta_1 =\{\pairQ{y_i}{B_i}\}$ and $\Delta_2=\{\pairQ{x_i}{A_i}\}$
as a slight abuse of notation for the set
$\{\br{\pairQ{y_i}{B_i}}\pairQ{x_i}{A_i}\}_i$ assuming it can be
inferred from the context how the $\phl {y_i}$'s and the
$\phl {x_i}$'s are paired. Moreover, we write $\oplus\Delta$ as a
shorthand for any set $\{\pairQ{x_i}{A_i\oplus B_i}\}$ and similarily
$\query\Delta$ for $\{\pairQ{x_i}{\query A_i}\}$.

\mypar{Judgement and Rules.} A judgement, denoted by $\syncseq P \Gamma$,
captures precisely those forwarding processes $\phl P$ that connect
the endpoints represented in $\Gamma$, buffer at most one message at
each endpoint, and preserve order.

In~Fig.~\ref{fig:sync}, we report the rules for typing processes: they
correspond to the CLL sequent calculus enhanced with process terms
(using endpoints~\cite{CLMSW16}), but with some extra restrictions for
characterising forwarders.
\begin{figure}[t]
\begin{displaymath}
  \begin{array}{c}
    \infer[\textsc{Ax}]{
    \syncseq 
    { \fwd xy}
    { \pairQ{x}{\dual{A}}, \pairQ{y}{A}}
    }{ }
     \qquad
     \infer[\bot]
     {\syncseq {\wait xP} {\Gamma, \pairQ{x}{\bot}}}
     {\syncseq {P} {\Gamma, \Star}}
	\qquad	
	\infer[\one]{
		\syncseq{\close x} {\pairQ{x}{\one}, \Star}}{ }
	\\[2ex]
    \infer[\tensor]{
    \syncseq{\send xyPQ} 
    {\Gamma, [\Delta_1]\Delta_2, \pairQ x{A \tensor B}}
    }
    { \syncseq{P}{\Delta_1, \pairQ yA}
    & \syncseq{Q}{\Gamma, \Delta_2, \pairQ xB}
     }
      \qquad
      \infer[\parr]{
      \syncseq{\recv xyP} {\Gamma, \pairQ x {A \parr B}}
      }{
      \syncseq P {\Gamma, \quadQ yAxB}
      }     
   	\\[2ex]
    \infer[\with]
    {\syncseq{\Case xPQ} {\Gamma,\oplus\Delta,\pairQ x{A \with B}}
    }{
    	\syncseq P {\Gamma, \ltriCtx {\oplus\Delta}x{A}}
    	&
    	\syncseq Q {\Gamma, \rtriCtx {\oplus\Delta}x{B}}
    }
     \\[2ex]
	 \infer[\oplus_l]
	{\syncseq {\inl xP} {\Gamma,  \ltriCtx{\Delta, \phl{x}:\thl{A\oplus B}}zC }}
       {\syncseq P {\Gamma,  \ltriCtx{\Delta}zC, \pairQ{x}{A}}}
    \qquad
    	 \infer[\oplus_r]
	{\syncseq {\inr xP} {\Gamma,  \rtriCtx {\Delta, \phl{x}:\thl{A\oplus B}}zC }}
       {\syncseq P {\Gamma,  \rtriCtx{\Delta}zC, \pairQ{x}{B}}}
    \\[2ex]
    \infer[!]{
    	\syncseq {\srv xyP}{\query \Delta, \phl x:\thl{\bang A} } 
    	}{
    	\syncseq {P}{\Query\dbr{\query\Delta}\phl y:\thl A}
    	}
   	\qquad
    \infer[?]{
    	\syncseq {\client xyP }{\Gamma, \Query\dbr{\Delta, \phl x:\thl{\query A}}\phl z:\thl{C}} 
    }{
    	\syncseq {P}{\Gamma, \Query\dbr{\Delta}\phl z:\thl{C}, \phl y:\thl A}
    }

  \end{array}
\end{displaymath}
\caption{Synchronous forwarder logic}
\label{fig:sync}
\end{figure}
%
The defining characteristic of $\syncseq{}{}$ is that it uses $\perp$,
$\parr$, $\with$, $\bang$ as a buffering mechanism for respectively
endpoint messages (units and multiplicatives), choices (additives),
and requests (exponentials) in order to render them temporarily
inaccessible to any other rule.
The only way to awaken them (render them accessible again) is to
forward 
to another endpoint, external to the forwarder using $\one$,
$\tensor$, $\oplus_\#$, and $\query$.

Rule \Did{Ax} is the axiom of CLL: it denotes a process that
interfaces two endpoints of dual type.
%
Rules $\one$ and $\perp$ type forwarding of empty messages. In
$\phl{\wait xP}$, the typing makes sure that after an empty message is
received (of type $\perp$), it is forwarded by adding to the typing
context for $\phl P$ the object $\Star$ which will make sure that
eventually rule $\one$ is used. In fact, rule $\one$ is applicable
only if there is $\Star$ in the context, i.e., at least one $\perp$
rule has been encountered before. Note that this corresponds to the
gathering operation for units enforced by coherence.
%
%
Rule $\parr$ types reception of $\phl y$ over $\phl x$ with type
$A\parr B$. The received name of type $A$ cannot be used but must be
forwarded, therefore it is wrapped as $\quadQ yAxB$ for the typing of
$\phl P$. Endpoint $\phl x$ is temporarily blocked, until $\phl y$ is
finally forwarded. This is done by rule $\tensor$, which collects the
received $\br{\Delta_1}$ and spawns a new forwarder $\phl P$ of type
$\Delta_1,\pairQ yA$, freeing $\pairQ xB$. As for units,
multiplicatives implement gathering, i.e., we forward many sends to a
single receiver. Note that, since contexts are sets, we explicitly
require $\fn(\Delta)$ and $\fn(\Gamma)$ to be different.

Rules $\oplus_\#$ (where $\#\in\{l,r\}$) and $\with$ type branching
processes. Unlike the case of units and multiplicatives which use a
gathering communication mechanism (many-to-one), additives (and later
exponentials) use broadcasting (one-to-many). This is a choice that
follows directly from the logic principles and the way each operator
are interpreted. Rule $\with$ types the process $\phl{\Case xPQ}$: in
branches $\phl P$ and $\phl Q$ the selected choice, left and right
respectively, must be forwarded to some other endpoint. This is
enforced by $\ltriCtx{\Delta}{x}{A}$ and
$\rtriCtx{\Delta}{x}{A}$. Besides containing the information on which
branch must be forwarded ($\Left$ or $\Right$), they also block
$\pairQ{x}{A}$ until the information signaling to pick the left (or
right) branch has been forwarded to all active endpoints in
$\Delta =\{\phl {x_i} : A_i \oplus B_i\}$.
Similarly, for exponentials, endpoint $\phl y$ is blocked until all
other endpoints in $\Delta = \{\phl {x_i} : \query B\}$ agree (in any
order) that $A$ may proceed.


We conclude this subsection with the straightforward result, that
embeds $\syncseq{}{}$ into $\cllseq{}{}$, where $\cllseq{}{}$ is the
typing sequent based on CLL~\cite{CLMSW16}.

\begin{proposition}\label{prop:embed}
  If $\syncseq P\Gamma$, then  $\cllseq{P}{\rep{\Gamma}}$;
  embedding $\rep{\cdot}$ being defined as:
  \begin{displaymath}
  	\begin{array}{l@{\qquad}l}
  		\rep{\Delta} = \Delta &
  		\rep{\Gamma, \Star}  = \rep{\Gamma} \\[1ex]
  		\rep{\Gamma, \quadQ{y\!}{\!B}{x\!}{\!A}}  = \rep{\Gamma},\pairQ{y\!}{\!B}, \pairQ{x\!}{\!A} &
  		\rep{\Gamma, \Symbol\dbr{\Delta}\pairQ{x\!}{\!A}}  = \rep{\Gamma},\Delta,  \pairQ{x\!}{\!A}
  	\end{array}
  \end{displaymath}
\end{proposition}

\section{Semantics via Cut Elimination}\label{sec:cut-elim}
We now turn to the formal semantics of forwarders. Our goal is
two-fold: provide a semantics for forwarders and have a way of
composing them safely. As a consequence, we obtain a methodology for
connecting two forwarders together, define how they can internally
communicate, and be sure that after composition the obtained process
is still a forwarder. In order to illustrate how this can be used, we
begin with of an extension of the 2-buyer protocol example.
\begin{example} \label{ex:2buff} We extend the 2-buyer example with a
  second concurrent system. Assume that the second buyer wishes to
  delegate the decision to buy to two colleagues and only if they both
  agree, the book will be bought.
  Here again, we can use a forwarder $\phl{P_2}$ to orchestrate the
  communication between the two colleagues and the second buyer:
  it forwards the price to the first colleague, and the first buyer's
  contribution (obtained through the second buyer) to the second
  colleague.
  Then, it relays decisions from the two colleagues to the second
  buyer. 

  We write this forwarder as process $\phl{P_2}$ below, where
  $\phl{b''_2}$, $\phl{c_1}$, and $\phl{c_2}$ are the endpoints
  connecting $\phl{P_2}$ to the second buyer and the two colleagues.
  Note that $\phl{P_1}$ is connected to the second buyer through
  endpoint $\phl{b_2'}$ instead (see Example~\ref{ex:2buf}).
  %
  \begin{eqnarray*}
    \phl{P_2} &\!=\!&
                      \begin{array}[t]{@{}l}
                        \recv {b_2''}{y_1} {\send{c_1}{y_1'}{\fwd{y_1}{y_1'}}{}}\
                        \recv {b_2''}{y_2} {\send{c_2}{y_2'}{\fwd{y_2}{y_2'}}{}} \\
                        \quad
                        \phl{{c_1}.\m{case}}
                        \phl{\bigg (}
                        \begin{array}{llll}
                          {\ \; \inl{c_2}{}{}
                          \Case{c_2}{\ \inl{b_2''}Q_2}{\
                          \inr{b_2''}Q_2\ }}{}\phl ,\\
                          {\ \;
                          \inr{c_2}{}{}\Case{c_2}{\ \inl{b_2''}Q_2}{\
                          \inr{b_2''}Q_2\ }\ \; }
                        \end{array}
                        \phl {\bigg )}
                      \end{array}
  \end{eqnarray*}
  Process $\phl{Q_2} =\wait{c_1}{\wait{c_2}{\close{b_2''}}}$ closes all
  endpoints.
  Intuitively, we wish to combine $\phl {P_1}$ and $\phl {P_2}$ into a
  new forwarder $\phl P$ that orchestrates the communications between
  the seller, the first buyer, and the colleagues, bypassing the
  second buyer. \qed
\end{example}

\mypar{Is a standard cut rule enough?} In a propositions-as-types
approach, the semantics of processes is obtained from the reductions
of proofs (which correspond to processes) given by the {\em cut elimination} process.
%
In CLL, a cut rule allows to compose two compatible proofs,
establishing their compatibility based on duality. Cut elimination is
then a procedure for eliminating any occurrence of the cut rule in a
proof (also known as normalisation) and it is usually done in small
steps called cut reductions. In our process terms, a cut rule
corresponds to connecting two endpoints from two parallel
processes; and, cut reductions correspond to reductions of processes,
thus yielding their semantics.

Ideally, since synchronous forwarders are embedded into CLL, we could
use the CLL cut rule to compose them. However, synchronous forwarder
have an extended typing context containing extra information on what
has to be forwarded. Therefore, the CLL cut rule is not sufficient. In
order to understand why, we start from defining the core rule
\textsc{Cut} for synchronous forwarders which, as in CLL, connects two
endpoints with dual types:
\begin{smallequation*}
  \infer[\textsc{Cut}]{
    \syncseq {\cpres{x}A{y}{A^\perp}(P \pp Q)} {\Gamma_1,\Gamma_2}
  }{
    \syncseq{P}{\Gamma_1, \pairQ{x}{A}}
    &
    \syncseq {Q} {\Gamma_2, \phl {y}: \thl{A^\perp}}
  }    
\end{smallequation*}
We can read the rule above as follows: we can compose the two
processes $\phl{P}$ and $\phl{Q}$ willing to communicate on endpoints
$\phl x$ and $\phl{y}$ because their types are dual. And, their
parallel composition yields a new process where both $\phl x$ and
$\phl y$ have disappeared, since now they have formed an internal
channel. Both processes implement input and output operations
over the two connected endpoints, but when the two processes
start communicating, their state will change into something that may
contain messages in transit awaiting to be forwarded. This means that
the cut rule may have to involve some of the special elements in a
context $\Gamma$.  In order to deal with such extra constraints on
contexts, we require special cut rules that deal with intermediary
computation (referred to as {\em runtime} rules) and associated to new
syntactic terms (runtime syntax), depending on the proposition
that is being cut (and the corresponding communication operation).
Such rules deal with blocked endpoints such as
the boxed judgments $\quadQ xAyB$, $\ltriCtx{\Delta}xA$,
$\rtriCtx{\Delta}xA$, and $\triCtx{\Delta}{x}A$.
%
%
As a consequence, we define six variations of the cut rule for
synchronous forwarders, which are reported in
Fig.~\ref{fig:sync_cutrules}.
\begin{figure}[t]
	\begin{displaymath}
		\begin{array}{c}
			\infer[\textsc{Cut}]{
				\syncseq {\cpres{x}A{y}{A^\perp}(P \pp Q)} {\Gamma_1,\Gamma_2}
			}{
				\syncseq{P}{\Gamma_1, \pairQ{x}{A}}
				& 
				\syncseq{Q}{\Gamma_2, \pairQ{y}{\dual{A}}}
			}
			\\[2ex]
			\infer[\textsc{Cut}_{\tensor\parr}]{
				\syncseq 
				{\cpres xB{\br{}y}{B^\perp}\big(Q\pp \cpres u-{\br{v}}{} (P\pp R)\big)}
				{[\Delta_1]\Delta_2,\Gamma_1, \Gamma_2}
			}{
				\syncseq {P} {\Delta_1,\pairQ uA}
				&
				\syncseq {Q} {\Delta_2, \Gamma_1, \pairQ xB}
				&
				\syncseq {R} {\Gamma_2,\quadQ {v}{A^\perp}{y}{B^\perp}}
			}
			\\[2ex]
			\infer[\textsc{Cut}_{\oplus\with_{1\#}}]{
				\syncseq {\cpres{x}{}{y}{}({P} \pp \Case{y}{Q}{R})} {\Gamma_1,  \Gamma_2, \Select[\#]\dbr{\Delta_1, \Delta_2}\phl z: \thl {C}}
			}{
				\syncseq {P} 
				{\Gamma_1, \Select[\#]\dbr{\Delta_1, \phl x: \thl{A \oplus B}}\phl z: \thl C}
				&
				\syncseq{Q} {\Gamma_2, \Left\dbr{\Delta_2}\pairQ {y}{A^\perp}}
				&
				\syncseq{R} {\Gamma_2, \Right\dbr{\Delta_2}\pairQ {y}{B^\perp}}
			}
			\\[2ex]
			\infer[\textsc{Cut}_{\oplus\with_{2\#}}]{
				\syncseq {\cpres{x}{}{y}{}({P} \pp Q)} {\Gamma_1,  \Gamma_2, \Select[\#]\dbr{\Delta_1, \Delta_2}\phl z: \thl {C}}
			}{
				\syncseq {P} {\Gamma_1, \Select[\#]\dbr{\Delta_1}\phl z: \thl C, \phl x: \thl{A}}
				&
				\syncseq{Q} {\Gamma_2, \Select[\#]\dbr{\Delta_2}\pairQ{y}{A^\perp}}
			}	
			\\[2ex]			
			\infer[\textsc{Cut}_{!?_1}]{
				\syncseq {\cpres  xAy{A^\perp}(P\pp \srv yvQ)} {\Gamma, \Query\dbr {\Delta_1,\query\Delta_2}\phl{z}:\thl{C}}
			}{
				\syncseq {P} {\Gamma, \Query\dbr{\Delta_1,\phl x:\thl{\query A}}\phl{z}:\thl{C}}
				&
				\syncseq {Q} {\Query\dbr{\Delta_2}\phl v: \thl{A^\perp}}
			}
			\\[2ex]
			\infer[\textsc{Cut}_{!?_2}]
			{
				\syncseq {\cpres  xAy{A^\perp}(P\pp Q)} {\Gamma_1,\Gamma_2, \Query\dbr{\Delta_1,\Delta_2}\phl{z}:\thl{C}}
			}
			{
				\syncseq {P} {\Gamma_1, \Query\dbr{\Delta_1}\phl{z}:\thl{C}, \phl x: \thl A}
				&
				\syncseq {Q} {\Gamma_2, \Query\dbr{\Delta_2}\phl{y}:\thl{A^\perp}}
			}
		\end{array} 
	\end{displaymath}
	\caption{Cut rules for synchronous forwarders}
	\label{fig:sync_cutrules} 
\end{figure}

We discuss these rules by considering the four possible cases associated
to each pair of dual modalities (units, multiplicatives, additives,
and exponentials). 
Later, we will show that the rule \Did{Cut} and all runtime cut rules
are admissible for synchronous forwarders, i.e., we can always create
a forwarder that has no parallel composition yet behaves as the
composition of the two original forwarders.
We use $\reducesto_\beta$ to denote a rewriting of a proof (and the associated process) that reduces the complexity of the type $A$ in $\Did{Cut}$.
\mypar{Units.} This case does not require any additional rule. We can
always remove a cut between an application of rule $\one$ and an
application of rule $\perp$:
\begin{smallequation*}
  \infer[\textsc{Cut}]{
    \syncseq {\cpres{x}A{y}{A^\perp}(\close x \pp \wait {y}{P})} {\Star,\Gamma}
  }{
    \infer[\one]{ 
      \syncseq{\close x} {\pairQ{x}{\one}, \Star}
    }{}
    &
    \infer[\bot]{
      \syncseq {\wait {y}{P}} {\pairQ{y}{\bot}, \Gamma}}{
      \syncseq {P} {\Star, \Gamma}}
  }
  \quad\reducesto_\beta\quad
  \deduce[]{
    \syncseq {P} {\Star,\Gamma}
  }{}
\end{smallequation*}
For processes, we obtain the reduction
$\phl{\res {xy}(\close x \pp \wait {y}{P})\reducesto_\beta P}$ which
shows how a final synchronisation closes the connection between
endpoints $\phl x$ and $\phl y$.


\mypar{Multiplicatives.} In the case of $\tensor$ and $\parr$, we get
the following principal case:
\begin{smallequation*}
  \infer[\textsc{Cut}]{
    \syncseq {\cpres{x}-{y}{-}(\send xuPQ \pp \recv {y}{v}{R})}
    {\br{\Delta_1}\Delta_2, \Gamma_1,  \Gamma_2}
  }{
    \infer[\tensor]{
      \syncseq{\send xuPQ} 
      {\br{\Delta_1}\Delta_2, \Gamma_1,  \phl x:\thl{A \tensor B}} 
    }{
      \syncseq P {\Delta_1, \pairQ uA} 
      & 
      \syncseq Q {\Delta_2, \Gamma_1, \pairQ xB }
    }
    &
    \infer[\parr]{
      \syncseq{\recv {y}{v}{R}} {\Gamma_2, \phl {y}: \thl{A^\perp \parr B^\perp}}
    }{
      \syncseq {R} {\Gamma_2, [\phl {v}: \thl {A^\perp}] \phl {y}: \thl{B^\perp}}
    }    
  }
\end{smallequation*}
What we would usually do in the proof of cut-elimination for CLL is 
to replace $\Did{Cut}$ by two cuts on $A$ and $B$, respectively.
This, however, is not possible in this case, as
$\phl {v}: \thl {A^\perp}$ is pushed into a buffer that is linked to
its sending endpoint $\phl {y}$, in the configuration
$[\phl {v}: \thl {A^\perp}] \phl {y}: \thl{B^\perp}$.  The two need to
remain linked until the message $\phl v$ is forwarded and channel
$\phl y$ becomes active again.

This requires us to introduce a new {runtime} cut rule which
handles $A$ and $B$ \emph{at the same time}, despite them splitting
into two distinct communications. Yet, they are intertwined and the
cut rule needs to capture the \emph{three premisses} together. 
This is achieved by rule $\textsc{Cut}_{\tensor\parr}$ in
Fig.~\ref{fig:sync_cutrules}.
Note that we have also adopted a runtime process syntax for the
new rules, yielding
$\phl{\cpres xB{\br{}y}{B^\perp}\big(Q\pp \cpres u-{\br{v}}{} (P\pp
  R)\big)}$ as a new term: the box $\phl {\br{}}$ in front of $\phl y$
signals that endpoint $\phl y$ is blocked, while $\phl{[v]}$ means
that $\phl v$ is a message in transit that must be forwarded.

The location of the corresponding sub-derivation ending in $A^\perp$
may be 
deep in the derivation of $\phl R$ and can only be retrieved by
applying several {\em commuting conversions}. Commuting conversions,
denoted by $\reducesto_\kappa$, are proof transformations that do not
change the size of a proof nor the size of a proposition (type), but
only perform rule permutations.
We report a full list of commuting conversions (as processes) at the
end of this section, in Fig.~\ref{fig:syncsemantics2}.

Returning to our case, once $\textsc{Cut}_{\tensor\parr}$ meets the send operation
($\tensor$) that frees $\phl y$, by forwarding
$\phl{v}:\thl{A^\perp}$ and spawning a new process $\phl S$, the
communication can continue with two basic $\textsc{Cut}$ rules:
\begin{smallequation*}
  \infer[\textsc{Cut}_{\tensor\parr}]
  {
    \syncseq 
    {\cpres xB{[]y}{B^\perp}(Q\pp \cpres uA{[v]}{A^\perp} (P \!\pp\! \send wz{S}{T}))}
    {[\Delta_1]\Delta_2,\Gamma_1, \Gamma_2, [\Delta_3]\Delta_4, \pairQ
      {z\!}{\!C\!\tensor\! D}}
  }{
    \syncseq {P} {\Delta_1,\pairQ uA}
    &\quad
    \syncseq {Q} {
    	\begin{array}{lll}
    		\Delta_2, \Gamma_1,\\
    		\pairQ xB
    	\end{array}
    }
    &\quad
    \infer[\tensor]{
      \syncseq {\send wz{S}{T}}
      {
        \begin{array}{lll}
          \Gamma_2, \quadQ {v\!}{\!A^\perp}{y}{B^\perp}, \\{}
          [\Delta_3]\Delta_4,  \pairQ {z\!}{\!C\!\tensor\! D}
        \end{array}
      }
    }{\syncseq{S}{
      	\begin{array}{lll}
        \Delta_3, \phl{v}:\thl{A^\perp}, \\
        \phl{w}:\thl{C}
	  	\end{array}
		}
        &
        \syncseq{T}{
        \begin{array}{lll}
          \Gamma_2, \phl{y}:\thl{B^\perp},\\
          \Delta_4, \phl z:\thl{D}
        \end{array}
    	}
    }
  }
\end{smallequation*}
\begin{smallequation*}\label{page:tensorcurcommute}
  \begin{array}{lll}
    \reducesto_\kappa\\
    \infer[\tensor]{
    \syncseq 
    {\send zw{\cpres uAv{A^\perp} (P \pp S)}{\cpres xBy{B^\perp}(Q \pp T)}}
    {
    \begin{array}{lll}
      [\Delta_1]\Delta_2,\Gamma_1, \Gamma_2, [\Delta_3]\Delta_4, 
      \phl z:\thl{C\tensor D}
    \end{array}
    }
    }{
    \infer[\textsc{Cut}]{
    \syncseq {\cpres uAv{A^\perp} (P\! \!\pp\!\! S)}
    {\Delta_1,\Delta_3, \phl w:\thl{C}}
    }{
    \syncseq {P} {\Delta_1,\pairQ uA}
    &
      \syncseq{S}{
      \begin{array}{lll}
          \Delta_3, \phl{v}:\thl{A^\perp}, \\
        \phl{w}:\thl{C}
      \end{array}
    }
    }
    &
      \infer[\textsc{Cut}]{
      \syncseq {\cpres xBy{B^\perp}(Q\!\! \pp\!\! T)}
      {\Delta_2,\Gamma_1, \Gamma_2, \Delta_4, \phl z:\thl{D}}
      }{
      \syncseq {Q} {
      \begin{array}{lll}
        \Delta_2, \Gamma_1,\\
        \pairQ xB
      \end{array}
    }
    &
      \syncseq{T}{
        \begin{array}{lll}
          \Gamma_2, \phl{y}:\thl{B^\perp},\\
          \Delta_4, \phl z:\thl{D}
        \end{array}
    }
    }
    }
  \end{array}
\end{smallequation*}

\mypar{Additives.}  What would a key reduction look like for the
additive connectives in this system?
That is, we need to consider the reductions available to a cut with
premisses $\syncseq {P}{\Gamma_1, \pairQ{x}{A \oplus B}}$ and
$\syncseq {Q}{\Gamma_2, \pairQ {y}{A^\perp \with B^\perp}}$.
%
By using commuting conversions, it is always possible to reach the
point where the right premiss (that of $\phl{Q}$) ends with the rule
for which the cut formula is {\em principal}, that is
$\phl{Q =} \Case{y}{Q_1}{Q_2}$,
$\Gamma_2 = \Gamma'_2, \oplus\Delta_2$, and the right branch ends as:
\begin{smallequation*}
  \infer[\with]{
    \syncseq{\Case{y}{Q_1}{Q_2}} {\Gamma'_2, \oplus\Delta_2, \pairQ {y}{A^\perp \with B^\perp}}
  }{
    \syncseq{Q_1} {\Gamma'_2, \Left\dbr{\oplus\Delta_2}\pairQ {y}{A^\perp}}
    &
    \syncseq{Q_2} {\Gamma'_2, \Right\dbr{\oplus\Delta_2}\pairQ {y}{B^\perp}}
  } 
\end{smallequation*}

On the left branch however, reaching the rule for which the
cut-formula is principal must be done in two steps as $A\oplus B$ can
only be principal after having been chosen for selection.
This means that we first need to consider the point at which
$\phl{P =} \Case{z}{P_1}{P_2}$,
$\Gamma_1 = \Gamma'_1, \oplus\Delta_1, \phl z: \thl {C\with D}$ and
the left branch of the cut is of the shape:
\begin{smallequation*}
  \infer[\with]
  {\syncseq {\Case{z}{P_1}{P_2}} {\Gamma'_1, \oplus\Delta_1, \pairQ{x}{A \oplus B}, \phl z: \thl {C\with D}}}
  {
    \syncseq {P_1} {\Gamma'_1, \Left[\![\oplus\Delta_1, \pairQ{x}{A \oplus B}]\!]\phl z: \thl C}
    &
    \syncseq {P_2} {\Gamma'_1, \Right[\![\oplus\Delta_1, \pairQ{x}{A \oplus B}]\!]\phl z: \thl D}
  }
\end{smallequation*}
That places $A\oplus B$ at least in a position to become, higher in
the proof, principal.
We are therefore in a special configuration of the cut between
additives:
\begin{smallequation*}
  \infer[
  ]{
    \syncseq {\cpres{x}A{y}{A^\perp}(\Case{z}{P_1}{P_2}\pp\Case{y}{Q_1}{Q_2})}
    {\Gamma'_1, \oplus\Delta_1, \phl z: \thl {C\with D},  \Gamma'_2, \oplus\Delta_2}
  }{
    \deduce[]{
      \syncseq {\Case{z}{P_1}{P_2}} {\Gamma'_1, \oplus\Delta_1,
        \pairQ{x\!}{\!A\! \oplus\! B}, \pairQ {z\!}{\!C\!\with\! D}}
    }{}
    \quad
    \deduce[]{
      \syncseq{\Case{y}{Q_1}{Q_2}} {\Gamma'_2, \oplus\Delta_2, \pairQ {y\!}{\!A^\perp \!\with\! B^\perp}}
    }{}    
  }
\end{smallequation*}

We introduce the \emph{runtime} cut $\textsc{Cut}_{\oplus\with_{1l}}$
that puts together the premisses as:
\begin{smallequation*}
  \infer[
  ]{
    \syncseq {R_1 = \cpres{x}{}{y}{}(P_1 \pp \Case{y}{Q_1}{Q_2})}
    {\Gamma'_1, \Left\dbr{\oplus\Delta_1, \oplus\Delta_2}\phl z: \thl {C},  \Gamma'_2}
  }{
    \syncseq {P_1} {\Gamma'_1, \Left[\![\oplus\Delta_1, \pairQ {x\!}{\!A\! \oplus\! B}]\!]\phl z: \thl C}
    &
    \syncseq{Q_1} {\Gamma'_2, \Left\dbr{\oplus\Delta_2}\pairQ {y\!}{\!A^\perp}}
    &
    \syncseq{Q_2} {\Gamma'_2, \Right\dbr{\oplus\Delta_2}\pairQ {y\!}{\!B^\perp}}
  }
\end{smallequation*}
Similarly, we introduce $\textsc{Cut}_{\oplus\with_{1r}}$ to get $\phl{R_2 = \cpres{x}{}{y}{}(P_2 \pp \Case{y}{Q_1}{Q_2})}$ on the right.
The two cuts can then be re-combined as follows:
\begin{smallequation*}
\infer[\with]{
	\syncseq {\Case{z}{R_1}{R_2}} {\Gamma'_1, \oplus\Delta_1, \phl z: \thl {C\with D}, \oplus\Delta_2,  \Gamma'_2}
}{
	\deduce[]{
		\syncseq {R_1} 
		{\Gamma'_1, \Left\dbr{\oplus\Delta_1, \oplus\Delta_2}\phl z: \thl {C},  \Gamma'_2}
	}{}
	&
	\deduce[]{
		\syncseq {R_2} 
		{\Gamma'_1,  \Right\dbr{\oplus\Delta_1, \oplus\Delta_2}\phl z: \thl {D},  \Gamma'_2}
	}{}
}
\end{smallequation*}
Indeed, this is a special case of a general equivalence (as in CLL)
obtained by commuting the rule applied to the side formula $C\with D$
below the cut:
\begin{smallequation*}
  \phl{
    \cpres{x}A{y}{A^\perp}(\Case{z}{P_1}{P_2} \pp Q)
    \quad\reducesto_\kappa\quad }
  \Case{z}{\cpres{x}{}{y}{}(P_1 \pp Q)}{\cpres{x}{}{y}{}(P_2 \pp Q)}
\end{smallequation*}
but in the case where 
$\phl{Q =} \Case{y}{Q_1}{Q_2}$.

Then, these $\textsc{Cut}_{\oplus\with_{1\#}}$ will be able to reduce
further when meeting the corresponding $\oplus_\#$-rule making the
cut-formula principal (inside its box), as for instance when $\phl{P_1 = }\inl{x}{P'}$:
\begin{smallequation*}
  \infer[\textsc{Cut}_{\oplus\with_{1l}}]{
    \syncseq {\cpres{x}{}{y}{}(\inl{x}{P'} \pp \Case{y}{Q_1}{Q_2})}
    {\Gamma'_1, \Left\dbr{\oplus\Delta_1, \oplus\Delta_2}\phl z: \thl {C},  \Gamma'_2}
  }{
    \infer[\oplus_1]{
      \syncseq {\inl{x}{P'}} {\Gamma'_1, \Left[\![\oplus\Delta_1, \phl x: \thl{A \oplus B}]\!]\phl z: \thl C}
    }{
      \syncseq {P'} {\Gamma'_1, \Left[\![\oplus\Delta_1]\!]\phl z: \thl C, \phl x: \thl{A}}
    }
&
\begin{array}{lll}
	\syncseq{Q_1} {\Gamma'_2, \Left\dbr{\oplus\Delta_2}\pairQ {y}{A^\perp}}
	\\[2mm]
	\syncseq{Q_2} {\Gamma'_2, \Right\dbr{\oplus\Delta_2}\pairQ {y}{B^\perp}}\\{}
\end{array}
    }
\end{smallequation*}
\begin{smallequation*}
\qquad\reducesto_\beta\qquad
\infer[\textsc{Cut}_{\oplus\with_{2l}}]{
	\syncseq {\cpres{x}{}{y}{}({P'} \pp {Q_1})} {\Gamma'_1, \Left\dbr{\oplus\Delta_1, \oplus\Delta_2}\phl z: \thl {C},  \Gamma'_2}
}{
	\syncseq {P'} {\Gamma'_1, \Left[\![\oplus\Delta_1]\!]\phl z: \thl C, \phl x: \thl{A}}
	&
	\syncseq{Q_1} {\Gamma'_2, \Left\dbr{\oplus\Delta_2}\pairQ {y}{A^\perp}}
}	
\end{smallequation*}
In terms of processes, the above models the reduction:
\begin{smallequation*}
  \phl{ \cpres{x}{}{y}{}(\inl{x}{P'} \pp \Case{y}{Q_1}{Q_2})
    \reducesto_\beta \cpres{x}{}{y}{}({P'} \pp {Q_1}) }
\end{smallequation*}
However, this reduction introduces again a new 
rule $\textsc{Cut}_{\oplus\with_{2\#}}$ since the cut formula on the
right branch is blocked until the selection phase is fully over, i.e.,
until $\oplus\Delta_2$ becomes empty and $\phl{y}$ becomes active
again.
This would happen in the following situation, where $\phl{Q_1=}\inl{w}{Q'}$, and the communication
then returns to a general $\textsc{Cut}$:

\begin{smallequation*}
  \infer[\textsc{Cut}_{\oplus\with_{2l}}]{
    \syncseq {\cpres{x}{}{y}{}({P'} \pp \inl{w}{Q'})} {\Gamma'_1, \Left\dbr{\oplus\Delta_1, \phl w: \thl{B \oplus D}}\phl z: \thl {C},  \Gamma'_2}
  }{
    \syncseq {P'} {\Gamma'_1, \Left[\![\oplus\Delta_1]\!]\phl z: \thl C, \phl x: \thl{A}}
    &
    \infer[\oplus_1]{
      \syncseq{\inl{w}{Q'}} {\Gamma'_2, \Left\dbr{\phl w: \thl{B \oplus D}}\pairQ {y}{A^\perp}}
    }{
      \syncseq{Q'} {\Gamma'_2, \pairQ {y}{A^\perp}, \phl w: \thl{B}}
    }
  }	
\end{smallequation*}
\begin{smallequation*}
  \quad\reducesto_\kappa\quad
  \infer[\oplus_1]{
    \syncseq {\inl{w}{\cpres{x}{}{y}{}({P'} \pp {Q'})}} {\Gamma'_1,
      \Left\dbr{\oplus\Delta_1, \phl w: \thl{B \oplus D}}\phl z: \thl
      {C},  \Gamma'_2}
  }{
    \infer[\textsc{Cut}]{
      \syncseq {\cpres{x}{}{y}{}({P'} \pp {Q'})} {\Gamma'_1,
        \Left\dbr{\oplus\Delta_1}\phl z: \thl {C},  \Gamma'_2, \phl w:
        \thl{B}}
    }{
      \syncseq {P'} {\Gamma'_1, \Left[\![\oplus\Delta_1]\!]\phl z: \thl C, \phl x: \thl{A}}
      &
      \syncseq{Q'} {\Gamma'_2, \pairQ {y}{A^\perp}, \phl w: \thl{B}}
    }	
  }
\end{smallequation*}
Note that here as well this equivalence is not specific to this case;
it happens whevener an $\oplus$-rule is permuted with a cut, but in
this particular configuration it allows
$\textsc{Cut}_{\oplus\with_{2l}}$ to turn into a $\textsc{Cut}$ again.

\mypar{Exponentials.} In the case of exponentials, we start from an
application of rule \Did{Cut} whose premisses are
$\syncseq {P} {\Gamma_1, \pairQ{x}{\query A}}$ and
$\syncseq {Q} {\Gamma_2, \pairQ {y}{\bang A} }$.
%
As for the additives, it is possible to reach a point where the right
branch ends with the rule for which $\bang A$ is principal. And also
in this case, we need to use a two-step strategy in the left branch.
First, we observe that a normal cut between $?A$ and $!A$ implies that the proof for $?A$ must contain a $\bang$-rule somewhere (which will box it and only then can a $\query$-rule be applied). 
Because of that, we must have a case for $\Did{Cut}$ which will
transform it into the runtime $\Did{Cut}_{!?_1}$:
\begin{smallequation*}
  \begin{array}{l}
    \infer[\textsc{Cut}]{
    \syncseq {\cpres{x}-{y}{-}(\srv rs{P_1}\pp \srv yz{Q_1})}
    {\query\Delta_1,  \query\Delta_2, \pairQ r{!B}}
    }{
	    \infer[!]
	    	{\syncseq {\srv rs{P_1}} {\query\Delta_1,\pairQ x{?A}, \pairQ r{!B}}  }
	    	{\syncseq {P_1} {\triCtx{\query\Delta_1,\pairQ x{?A}} s B}}
	    &
	     \infer[!]{
	     	\syncseq {\srv yz{Q_1}} {\query\Delta_2,\pairQ y{!A^\perp}}
	     }{
	     	\syncseq {Q_1} {\triCtx {\query\Delta_2} z{A^\perp}}
	     }
              }
              \\[2ex]
   \reducesto_\kappa\quad
    \infer[!]{
    \syncseq {\srv rs{\cpres x-y- (P_1\pp \srv yz{Q_1})} }
    {\query\Delta_1,\query\Delta_2,\pairQ r{!B}} 
    }{
    \infer[\textsc{Cut}_{!?_1}]
    {
    \syncseq {\cpres x-y- (P_1\pp \srv yz{Q_1})}
    {\triCtx{\query\Delta_1,\query\Delta_2} s B}
    }
    {
    {\syncseq {P_1} {\triCtx{\query\Delta_1,\pairQ x{?A}} s B}}
    &
      \syncseq {Q_1} {\triCtx {\query\Delta_2} z{A^\perp}}
      }
      }
  \end{array}
\end{smallequation*}
We can now observe a key reduction for $!$ and $?$ when the
corresponding $?$-rule applies to the left premiss of
$\Did{Cut}_{?!_1}$:
\begin{smallequation*}
  \begin{array}{l}
    \infer[\textsc{Cut}_{!?_1}]
    {
    \syncseq {\cpres{x}-{y}{-}(\client xw{P_2}\pp \srv yz{Q_1})}
    {\Gamma_1,\triCtx{\query\Delta_1, \query\Delta_2} s B}
    }
    {
    \infer[?]
    {\syncseq {\client xw{P_2}} {\Gamma_1,\triCtx{\query\Delta_1,\pairQ x{?A}} s B}  }
    {\syncseq {P_2} {\Gamma_1, \triCtx{\query\Delta_1} s B,\pairQ w{A}}}
    &
		\syncseq {Q_1} {\triCtx {\query\Delta_2} z{A^\perp}}
      }
    \\[2ex]
    \reducesto_\beta\quad
    \infer[\textsc{Cut}_{!?_2}]
    {
    \syncseq {\cpres{w}-{z}{-}({P_2}\pp {Q_1})}
    {\Gamma_1,\triCtx{\query\Delta_1, \query\Delta_2} s B}
    }
    {
    {\syncseq {P_2} {\Gamma_1, \triCtx{\query\Delta_1} s B,\pairQ w{A}}}
    &
      {
      \syncseq {Q_1} {\triCtx {\query\Delta_2} z{A^\perp}}
      }
      }
  \end{array}
\end{smallequation*}
This requires to introduce $\Did{Cut}_{?!_2}$ as the formula $A^\perp$ 
on the right premiss is blocked by the boxed $\Delta_2$.
Similarly to additives, we
can now push the cut up on the right premise until we can empty the
box in front of $A^\perp$. This is done by:
\begin{smallequation*}
  \infer[\textsc{Cut}_{?!_{2}}]{
    \syncseq {\cpres{w}{}{z}{}({P_2} \pp\client{u}{v}{Q_2})} {\Gamma_1,
      \Gamma_2, \Query\dbr{\query\Delta_1, \phl u: \thl{?C}}\phl s: \thl
      {B}}
  }{
    \syncseq {P_2} {\Gamma_1, \Query\dbr{\query\Delta_1}\phl s: \thl B, \phl w: \thl{A}}
    &
    \infer[\query]{
      \syncseq{\client{u}{v}{Q_2}} {\Gamma_2, \Query\dbr{\phl u:
          \thl{?C}}\pairQ {z}{A^\perp}}
    }{
      \syncseq{Q_2} {\Gamma_2, \pairQ {z}{A^\perp}, \phl v: \thl{C}}
    }
  }
\end{smallequation*}
\begin{smallequation*}
  \reducesto_\kappa\quad
  \infer[\query]{
    \syncseq {\client{u}{v}{\cpres{w}{}{z}{}({P_2} \pp {Q_2})}} {\Gamma_1,
      \Gamma_2, \Query\dbr{\query\Delta_1, \phl u: \thl{?C}}\phl z: \thl
      {B}}
  }{
    \infer[\textsc{Cut}]{
      \syncseq {\cpres{w}{}{z}{}({P_2} \pp {Q_2})} {\Gamma_1,  \Gamma_2,
        \Query\dbr{\query\Delta_1}\phl s: \thl {B}, \phl v: \thl{C}}
    }{
      \syncseq {P_2} {\Gamma_1, \Query\dbr{\query\Delta_1}\phl s: \thl B, \phl w: \thl{A}}
      &
      \syncseq{Q_2} {\Gamma_2, \pairQ {z}{A^\perp}, \phl v: \thl{C}}
    }	
  }
\end{smallequation*}

\mypar{Results and Semantics.}  
In Fig.~\ref{fig:syncsemantics1} and Fig.~\ref{fig:syncsemantics2}, we
report 
the reductions and structural transformation derived for synchronous
forwarders.
%
\begin{figure}[t]
  \begin{smallequation*}
    \begin{array}{l@{\ }c@{\ }l@{}l}
      \mbox{Structural equivalence} 
      \\ [1mm]
      \phl{\fwd xy}
      &\equiv&
               \phl{\fwd yx}
      \\
      \phl{\cpres y{\dual{A}}xA(Q \pp P)}
      &\equiv&
               \phl{\cpres xAy{\dual A} (P \pp Q)} 
      \\ 
      \phl{\cpres wBz{\dual B} (P \pp \cpres xAy{\dual A} (Q \pp R))}
      &\equiv&
               \phl{\cpres xAy{\dual A}(\cpres wBz{\dual B} (P \pp Q) \pp R)} 
      &\mbox{} \phl x,\phl z \in\fn(\phl Q)
      \\  
      \multicolumn{3}{l}{
      \phl{
      \cpres{y}{B}{\br{}z}{}(Q \!\!\pp\!\! \cpres xA{\br{w}}{} (\cpres{u}{C}{v}{}(P_1 \!\!\pp\!\! P_2)\!\!\pp\!\! R))	
      }}
  	\\
	\multicolumn{3}{l}{	  	
      \hspace*{2cm}\equiv\quad
      \phl{
      \cpres{u}{C}{v}{}(P_1 \!\!\pp\!\! \cpres{y}{B}{\br{}z}{}(Q  \!\!\pp\!\! \cpres xA{\br{w}}{} (P_2\!\!\pp\!\! R)))	
      }}
  	&
  	  \mbox{} \phl x,\phl v \in\fn(\phl{P_2})
      \\ 
      \multicolumn{3}{l}{
      \phl{
      \cpres{y}{B}{\br{}z}{}(\cpres{u}{C}{v}{}(Q_1 \!\!\pp\!\! Q_2)  \!\pp\! \cpres xA{\br{w}}{} (P\!\!\pp\!\! R))	
      }}
      \\
      \multicolumn{3}{l}{
      \hspace*{2cm}\equiv\quad
      \phl{
      \cpres{u}{C}{v}{}(Q_1 \!\!\pp\!\! \cpres{y}{B}{\br{}z}{}(Q_2  \!\!\pp\!\! \cpres xA{\br{w}}{} (P\!\!\pp\!\! R)))	
      }}
      &\mbox{}\phl y,\phl v \in\fn(\phl{Q_2})
      \\ 
      \multicolumn{3}{l}{
      \phl{
      \cpres{y}{B}{\br{}z}{}(Q  \!\!\pp\!\! \cpres xA{\br{w}}{} (P\!\!\pp\!\! \cpres{u}{C}{v}{}(R_1 \!\!\pp\!\! R_2)))	
      }}
      \\
      \multicolumn{3}{l}{
      \hspace*{2cm}\equiv\quad
      \phl{
      \cpres{u}{C}{v}{}(R_1 \!\!\pp\!\! \cpres{y}{B}{\br{}z}{}(Q  \!\!\pp\!\! \cpres xA{\br{w}}{} (P\!\!\pp\!\! R_2)))	
      }}
     &\mbox{}\phl z,\phl v \in\fn(\phl{R_2})
     \\ [2mm]
      \mbox{Key Reductions ($\beta$)}
      \\ [1mm]
       \phl{\cpres x{X}y{\dual X} ( \fwd xw \pp Q)}
       &\reducesto_\beta&
       \phl{Q \substone wy}
       \\
       \phl{\cpres x{\one}y\perp(\close x\pp \wait yP)}\quad
       &\reducesto_\beta&
       \phl{P}
       \\
      \phl{\cpres{x}{A\tensor B}{y}{\dual A\parr \dual B}(\send xuPQ  \pp \recv yvR)}
      &\reducesto_\beta&
      \phl{\cpres{x}{A}{\br{}y}{\dual A}(Q  \pp \cpres uA{\br{v}}{\dual A} (P\pp R))}
       \\
       \phl{\cpres x{A\oplus B}y{\dual A\with \dual B} (\inl xP\pp \Case yQR)}
       &\reducesto_\beta&
       \phl{\cpres xAy{\dual A} (P\pp Q)}
       \\
       \phl{\cpres x{A\oplus B}y{\dual A\with \dual B} (\inr xP\pp \Case yQR)}
       &\reducesto_\beta&
       \phl{\cpres x{B}y{\dual B} (P\pp R)}
       \\
       \phl{\cpres x{\query A}y{\bang A}(\client xuQ \pp \srv yvP)}
       &\reducesto_\beta&
       \phl{\cpres u{A}v{\dual A}(P \pp Q)}
    \end{array}
  \end{smallequation*}
  \caption{Semantics of Synchronous Forwarders: Equivalences and Reductions}
  \label{fig:syncsemantics1}
\end{figure}
The collection of rules defining essential reductions and commuting
conversions are sound and complete, which is summarized by the
following cut admissibility theorem.

\begin{figure}[t]
  \begin{smallequation*}
    \begin{array}{l@{\ }c@{\ }l@{}l}
      \\ [1mm] 
      \phl{\cpres{x}{}{y}{}(\wait{u}{P}\pp Q)}
      &\reducesto_\kappa&
                          \wait{u}{\cpres{x}{}{y}{}(P\pp Q)}
      \\
      \phl{\cpres{x}{}{y}{}(\send uv{P}{Q} \pp R)}
      &\reducesto_\kappa&
                          \send uv{P}{\cpres{x}{}{y}{}(Q \pp R)}
      \\  
       \phl{\cpres{x}{}{y}{}(\recv uv{P}\pp Q)}
       &\reducesto_\kappa&
       \recv uv{\cpres{x}{}{y}{}(P \pp Q)}
       \\
       \phl{\cpres{x}{}{y}{}(\inlr{u}{P}\pp Q)}
       &\reducesto_\kappa&
       \phl{\inlr{u}{\cpres{x}{}{y}{}(P\pp Q)}}
      \\
      \phl{\cpres{x}{}{y}{}(\Case{u}{P}{Q} \pp R)}
      & \reducesto_\kappa &
                            \Case{u}{\cpres{x}{}{y}{}(P\pp R)}{\cpres{x}{}{y}{}(Q\pp R)}
      \\
       \phl{\cpres{x}A{y}{A^\perp}(\srv xzP \pp \srv {u}{v}{Q})}
       &\reducesto_\kappa&
       \srv {u}{v}{\cpres{x}A{y}{A^\perp}(\srv xzP \pp {Q})}
       \\
       \phl{\cpres{x}A{y}{A^\perp}(\client uvP \pp Q)}
       &\reducesto_\kappa&
       \client {u}{v}{\cpres{x}A{y}{A^\perp}(P \pp {Q})}
      \\[1mm] 
      \phl{\cpres{y}{A}{\br{}z}{\dual A}(Q  \!\pp\! \cpres xA{\br{w}}{\dual A} (P\!\pp\! \wait{u}R))}
      &\reducesto_\kappa&
                          \multicolumn{2}{l}{
                          \phl{\wait{u}{\cpres{y}{A}{\br{}z}{\dual A}(Q  \!\pp\! \cpres xA{\br{w}}{\dual A} (P\!\pp\! R))}}
                          }
      \\
      \phl{
      \cpres yB{[]z}{B^\perp}(Q\!\pp\! \cpres xA{[w]}{A^\perp} (P \!\pp\! \send uv{R}{T}))
      }
      & \reducesto_\kappa&
      \send uv{\cpres xAw{A^\perp} (P \!\pp\! R)}{\cpres yBz{B^\perp}(Q \!\pp\! T)}
      &\mbox{}\phl w \in\fn(\phl R)
      \\ 
      \phl{
      \cpres yB{[]z}{B^\perp}(Q\!\!\pp\!\! \cpres xA{[w]}{A^\perp} (P \!\!\pp\!\! \send uv{R}{T}))
      }
      &\reducesto_\kappa&
                          \send uv{R}{\cpres yB{[]z}{B^\perp}(Q\!\!\pp\!\! \cpres xA{[w]}{A^\perp} (P \!\!\pp\!\! {T}))}
      &\mbox{}\phl w \in\fn(\phl T) 
      \\ 
      \phl{\cpres{y}{A}{\br{}z}{\dual A}(Q  \pp \cpres xA{\br{w}}{\dual A} (P\pp \recv uvR))}
      &\reducesto_\kappa&
      \multicolumn{2}{l}{
      \phl{\recv uv{\cpres{y}{A}{\br{}z}{\dual A}(Q  \pp \cpres xA{\br{w}}{\dual A} (P\pp R))}}
    }
      \\
      \phl{
	  \cpres{y}{A}{\br{}z}{\dual A}(Q  \!\pp\! \cpres xA{\br{w}}{\dual A} (P\!\pp\! \inlr{u}R))    
  		}
      &\reducesto_\kappa&
      \multicolumn{2}{l}{
      \phl{
     \inlr{u}{\cpres{y}{A}{\br{}z}{\dual A}(Q  \!\pp\! \cpres xA{\br{w}}{\dual A} (P\!\pp\! R))}    
    	}}
      \\
      \multicolumn{3}{l}{
      \phl{
      \cpres{y}{A}{\br{}z}{\dual A}(Q  \!\pp\! \cpres xA{\br{w}}{\dual A} (P\!\pp\! \Case{u}{R}{T}))
      } 
  	}
      \\
      \multicolumn{4}{l}{
      \hspace*{1cm}\reducesto_\kappa\quad
      \Case{u}{\cpres{y}{A}{\br{}z}{\dual A}(Q  \!\pp\! \cpres xA{\br{w}}{\dual A} (P\!\pp\! R))}
      {\cpres{y}{A}{\br{}z}{\dual A}(Q  \!\pp\! \cpres xA{\br{w}}{\dual A} (P\!\pp\! T))}
    	}
      \\
      \phl{\cpres{y}{A}{\br{}z}{\dual A}(Q  \pp \cpres xA{\br{w}}{\dual A} (P\pp \client {u}{v}R))}
      &\reducesto_\kappa&
      \multicolumn{2}{l}{
      \phl{\client {u}{v}{\cpres{y}{A}{\br{}z}{\dual A}(Q  \pp \cpres xA{\br{w}}{\dual A} (P\pp R))}}
    }
    \end{array}
\end{smallequation*}
   
  \caption{Semantics of Synchronous Forwarders: Commuting Conversions ($\kappa$)}
  \label{fig:syncsemantics2}
\end{figure}

\begin{theorem}[Cut Admissibility]
  The cut-rules $\Did{Cut}$, $\Did{Cut}_{\tensor\parr}$,
  $\Did{Cut}_{\with\oplus_{1\#}}$, $\Did{Cut}_{\with\oplus_{2\#}}$,
  $\Did{Cut}_{\bang\query_1}$, $\Did{Cut}_{\bang\query_2}$ are
  admissible rules in synchronous forwarder logic.
\end{theorem}

From the theorem above we can extend synchronous forwarders with the
six admissible cut rules, for which we write $\syncseqcut{P}{\Gamma}$.
By induction on the number of cut rules in the derivation of a
synchronous forwarder, we obtain immediately the main theorem of this
section. In the sequel, $\Longrightarrow$ corresponds to
$\reducesto_\kappa \circ \reducesto \circ \reducesto_\kappa$ and the
$*$ denotes its reflexive and transitive closure.
\begin{theorem}[Cut elimination]
  If
  $\syncseqcut{P}{\Gamma}$ 
  then there exists a cut-free $\phl Q$ such that
  $\phl P \Reducesto^* \phl Q$ and $\syncseq{Q}{\Gamma}$.
\end{theorem}

As mentioned in the introduction, we can deduce the following
corollaries
\begin{corollary}[Subject reduction]
  If $\syncseq P \Gamma$ and $\phl P \Reducesto \phl Q$, then
  $\syncseq Q  \Gamma$.
\end{corollary}
\begin{corollary}[Deadlock Freedom]
  If $\syncseq P \Gamma$, then there exists a restriction-free
  $\phl Q$ such that $\phl P \Reducesto^* \phl Q$ and
  $\syncseq Q \Gamma$.
\end{corollary}

\begin{example} \label{ex:comp} We revisit the two buyers example and
  its extension (Examples~\ref{ex:2buf} and~\ref{ex:2buff}). Using
  compositionality, we can combine forwarders $\phl {P_1}$ and
  $\phl{P_2}$ into $\phl P = \phl{\res{b_2'b_2''} (P_1 | P_2)}$.  By
  cut reductions, we can transform $\phl{P}$ into a new process
  $\phl{P'}$ that does not appeal to endpoints $\phl{b_2'}$ and
  $\phl{b_2''}$, such that:
  \begin{smallequation*}
    \begin{array}[t]{l}
      \quad \phl {P'}\syncseq {}{}\qquad
      \pairQ{b_1}{\tname^\perp \parr \tcost \tensor \tcost^\perp \parr\perp},
      \pairQ{c_1}{\tcost \parr ((\taddr^\perp \parr \perp) \with \perp)}, \\
      \pairQ{c_2}{\tcost \parr ((\taddr^\perp \parr \perp) \with \perp)},
      \pairQ{s}{\tname \tensor \tcost^\perp \parr \tcost^\perp \parr
      ((\taddr \tensor \one) \oplus \one)}
    \end{array}
  \end{smallequation*}
  And, by cut elimination, we know that the resulting process is a
  forwarder.  \qed
\end{example}



\section{From coherence to synchronous forwarders (and
  back)}\label{sec:coherencesf}
In this section, we establish 
a strong
connection between coherence and synchronous forwarders. First, we
show that any coherence proof can be transformed into a synchronous
forwarder (soundness). Then, we show the opposite, i.e., every
synchronous forwarder corresponds to a coherence proof (completeness).

\subsection{Soundness}
We begin with the soundness result, which builds on the already
mentioned result from Carbone et al.~\cite{CLMSW16}, where a
translation to arbiters in CLL is provided. The following theorem
shows a stronger result since coherence can be translated into
synchronous forwarders, which is a proper fragment CLL.

\begin{theorem}[Soundness]\label{thm:soundness} 
  If $\cohseq G \Delta$ then there exists a $\phl P$, such that
  $\syncseq{P}{\Delta^\perp}$
\end{theorem}
\begin{proof}
  By induction on the derivation of $\ghl G$, we show how to construct
  $\syncseq P\Delta^\perp$. The construction, denoted by
  $\coepp{\cdot}$, is identical to that of~\cite{CLMSW16} and is
  reported in
  Figure~\ref{fig:coepp}. 
  Below, we report the case of $\tensor\parr$. The other cases are similar.
  
  Suppose the derivation of $\phl G$ ends with 
  \begin{smallequation*}
    \infer[\tensor\parr]{
      \phl{G=} \ghl{\parrtensor{x_1 \dots x_n}{y}{G_1}{G_2}} \gseq
      \thl \Gamma,  \{\phl{x_i} : \thl{A_i \tensor B_i}\}_i,
      \phl{x}:\thl{C \parr D}
    }{
      \ghl {G_1} \gseq \thl \{\phl{x_i}:\thl{A_i}\}_i, \phl x:\thl C
      &
      \ghl {G_2} \gseq \thl {\Gamma}, \{\phl{x_i}:\thl{B_i}\}_i, \phl x:\thl D
    }
  \end{smallequation*}
  
  Then by i.h.\ on $\phl{G_1[y_i/x_i,y/x]}$, there exists
  $\syncseq{P_1}{\thl \{\phl{y_i}:\thl{A_i^\perp}\}_i, \phl y:\thl
    C^\perp}$; and by i.h. on $\phl{G_2}$ there exists
  $\syncseq{P_2}{\thl {\Gamma^\perp}, \{\phl{x_i}:\thl{B_i^\perp}\}_i,
    \phl x:\thl D^\perp}$.
  Therefore, we can construct:
  \begin{smallequation*}
    \infer[\parr]{
      \syncseq{\recv{x_1}{y_1}\ \dots \recv{x_n}{y_n}\ \send xy{P_1}{P_2}} 
      {\Gamma^\perp, \{\pairQ{x_i}{A_i^\perp \parr B_i^\perp}\}_i, \pairQ x{C^\perp \tensor D^\perp}}
    }{
      \deduce[\vdots]{}{
        \infer[\parr]{}{
          \infer[\tensor]{
            \syncseq{\send xy{P_1}{P_2}} {\Gamma^\perp, \{\quadQ
              {y_i}{A_i^\perp}{x_i}{B_i^\perp}\}_i, \pairQ x{C^\perp
                \tensor D^\perp}}
          }{
            \syncseq{P_1}{\thl \{\phl{y_i}:\thl{A_i^\perp}\}_i, \phl y:\thl C^\perp}
            &
            \syncseq{P_2}{\thl {\Gamma^\perp}, \{\phl{x_i}:\thl{B_i^\perp}\}_i, \phl x:\thl D^\perp}
          }
	}}}
  \end{smallequation*}
  which is indeed a synchronous forwarder derivation for 
  \begin{smallequation*}
    \coepp G\syncseq{}{(\Gamma, \{\pairQ{x_i}{A_i\tensor B_i}\}_i, \pairQ x{C \parr D})^\perp}
  \end{smallequation*}
as requested.\qed
\end{proof}
The proof of soundness shows not only that any coherence judgement
$\phl G\gseq\Delta$ can be proven as
$\coepp G\syncseq{}{\Delta^\bot}$, but also provides a constructive
algorithm for it.


%
%

\subsection{Completeness}\label{subsec:completeness}
We now move to proving the opposite of Theorem~\ref{thm:soundness},
i.e., given a synchronous forwarder $\syncseq P\Delta$, we can always
build a coherence proof $\ghl{G}\gseq\Delta^\perp$.
Intuitively, the procedure for transforming a synchronous forwarder
into a coherence proof consists of permuting rules within the derivation of
$\syncseq{P}{\Delta}$ in order to synthetise rules into coherence-like blocks.
Additives and exponentials are treated by permuting
$\oplus$s and $\query$s \emph{down} to the matching $\with$ and $\bang$ respectively; while, in contrast, multiplicatives and units get regrouped by pushing $\parr$s and
$\perp$s \emph{up} to their corresponding $\tensor$ and $\one$ respectively,
without changing the structure of the proof (apart from the rules
moved down/up). 
This is because of the different communication patterns modelled by
coherence: additives and exponentials use \emph{broadcasting} (one
endpoint broadcasts a message to many endpoints); on the other hand,
for multiplicatives and units, communication consists of
\emph{gathering} (many endpoints send a message to a single
collector).
%
%

\mypar{Additives and Exponentials.}  Our first step shows that any
application of $\oplus_\#$ (for $\#\in\{l,r\}$) and any application of
$?$ can be permuted down to the bottom of a proof if the proposition
it introduces is still in the context.
\begin{lemma}[$\oplus_\#$/$\query$ Invertibility]
	\label{lem:opluspermute}\label{lem:querypermute}
	\begin{enumerate}
		\item 	Let
                  $\syncseq P{\Gamma,
                    \Select[\#]\dbr{\Delta, \phl{x}:\thl{A_l\oplus
                        A_r}} \phl{z}:\thl{C}}$.
                  Then, there exists $\phl {P'}$ such that
		\begin{smallequation*}
			\infer[\oplus_\#]
			{\syncseq {\inl x{P'}} {\Gamma,  \Select[\#]\dbr{\Delta, \phl{x}:\thl{A_l\oplus
                        A_r}} \phl{z}:\thl{C} }}
			{\syncseq {P'} {\Gamma,  \Select[\#]\dbr{\Delta} \phl{z}:\thl{C}, \phl x: \thl{A_\#}}}
		\end{smallequation*}
		
		\item
			Let $\syncseq P{\Gamma, \triCtx{\Delta, \pairQ x{?A}}{z}{C}}$. Then,
			there exists $\phl {P'}$ such that
			\begin{smallequation*}
				\infer[?]{
					\syncseq {\client xy{P'} }{\Gamma, \Query[\![\Delta, \phl x:\thl{\query A}]\!]\phl z:\thl{C}} 
				}{
					\syncseq {P'}{\Gamma, \Query[\![\Delta]\!]\phl z:\thl{C}, \phl y:\thl A}
				}
			\end{smallequation*}
	\end{enumerate}
\end{lemma}
\begin{proof}
  We show the result for $\oplus_l$; the ones for $\oplus_r$ and
  $\query$ follow the same methodology,
  by induction on the size of the proof, 
  and a case analysis on the last rule applied in the typing derivation of
  $\phl{P}$.
  Note that besides permuting down the rule introducing $A_l\oplus
  A_r$, we do not change the proof structure. 
	
	If the last applied rule is $\oplus_l$ and it is on endpoint
        $\phl x$ then we are done (base case).
	Otherwise, rule $\oplus_l$ may concern either formulas in $\Delta$ or in some
        other box. In both cases, we proceed by a simple
        induction. Below, we look at the case where the formula is in
        another box.
	\begin{smallequation*}
	\infer[\oplus_l]
	{
		\syncseq{P = \inl wP_1} {\Gamma, \ltriCtx{\Delta,
				\phl{x}:\thl{A\oplus B}}{z}{C},
			\ltriCtx{\Delta', \pairQ w{D\oplus E}}yF
		}
	}
	{
		\syncseq{P_1}{\Gamma, \ltriCtx{\Delta,
				\phl{x}:\thl{A\oplus B}}{z}{C},
			\ltriCtx{\Delta'}yF, \pairQ w{D}
		}
	}
	\end{smallequation*}
	By induction hypothesis, there exists $\phl{P'_1}$ such that
	\begin{smallequation*}
		\infer[\oplus_l]
		{
			\syncseq{\inl x{P'_1}}{\Gamma, \ltriCtx{\Delta,
					\phl{x}:\thl{A\oplus B}}{z}{C},
				\ltriCtx{\Delta'}yF, \pairQ w{D}
			}
		}
		{
			\syncseq{P'_1}{\Gamma, \ltriCtx{\Delta}{z}{C},
				\pairQ xA,
				\ltriCtx{\Delta'}yF, \pairQ w{D}
			}
		}
	\end{smallequation*}
	We can then reorganise the derivation to obtain $\phl{P' = }{\inl w{P'_1}}$ with:
	\begin{smallequation*}
		\infer[\oplus_l]
		{
			\syncseq{\inl x{\inl w{P'_1}}} {\Gamma, \ltriCtx{\Delta,
					\phl{x}:\thl{A\oplus B}}{z}{C},
				\ltriCtx{\Delta', \pairQ w{D\oplus E}}yF
			}
		}
		{
			\infer[\oplus_l]
			{
				\syncseq{\inl w{P_1'}}{\Gamma, \ltriCtx{\Delta}{z}{C},
					\pairQ xA,
					\ltriCtx{\Delta',  \pairQ w{D\oplus E}}yF
				}
			}
			{
				\syncseq{P_1'}{\Gamma, \ltriCtx{\Delta}{z}{C},
					\pairQ xA,
					\ltriCtx{\Delta'}yF, \pairQ w{D}
				}
			}
		}
              \end{smallequation*}
              The rest of the case analysis is similar.\qed
\end{proof}

The first step for showing how to transform a synchronous forwarder
into a coherence proof is to use Lemma~\ref{lem:opluspermute} to build
a forwarder that has immediately above every $\with$ all the
associated $\oplus$s, and immediately above every $!$ all the
associated $?$s. For example, for the additives, it gives the
following structure:
\begin{smallequation*}
  \infer[]{ 
    \syncseq{\Case x{P_0}{Q_0}} {\Gamma,\{\pairQ {y_i}{C_i\oplus D_i}\}_{1 \le i\le n},\pairQ x{A \with B}}
  }{
    \infer[]{ 
      \syncseq {P_0 } {\Gamma, \ltriCtx {\{\pairQ {y_i}{C_i\oplus D_i}\}_{i}}x{A}} 
    }{
      \infer[]{ 
        \syncseq {P_1} {\Gamma, \ltriCtx {\{\pairQ {y_i}{C_i\oplus D_i}\}_{2\le i}}x{A}, \pairQ{y_1}{C_1}}
      }{
        \deduce[\vdots]{}{
          \infer[]{}{ 
            \syncseq {P_n} {\Gamma, \pairQ {x}{A}, \{\pairQ {y_i}{C_i}\}_{1 \le i\le n} }
          }
        }
      }
    }
    &
    \infer[]{ 
      \syncseq {Q_0 } {\Gamma, \rtriCtx {\{\pairQ {y_i}{C_i\oplus D_i}\}_i}x{B}} 
    }{
      \infer[]{ 
        \syncseq {Q_1} {\Gamma, \ltriCtx {\{\pairQ {y_i}{C_i\oplus D_i}\}_{2\le i}}x{B}, \pairQ{y_1}{D_1}}
      }{
        \deduce[\vdots]{}{
          \infer[]{}{ 
            \syncseq {Q_n} {\Gamma, \pairQ {x}{A}, \{\pairQ {y_i}{D_i}\}_{1 \le i\le n} }
          }
        }
      }
    }
  }
\end{smallequation*} 
To make this more precise, we extend the relation $\syncseq{}{}$ to a
new relation, dubbed $\syncseqbis{}{}$, with the following two
derivable rules (where index $p$ stands for partial):
\begin{smallequation*}
		\infer[\oplus\with^\perp_p]{
			\syncseqbis
			{\Case x{\inl {y_1}{}\mydots\inl {y_n}{P}}{\inr {y_1}{}\mydots\inr {y_n}{Q}}}
			{\Gamma, \oplus\Delta_1, \Delta_2\oplus\Delta_3, \pairQ x{A \with B}}
		}{ 
			\syncseqbis P
			{\Gamma, \ltriCtx{\oplus\Delta_1}xA,  \Delta_2}
			& 
			\syncseqbis Q
			{\Gamma, \rtriCtx{\oplus\Delta_1}xB, \Delta_3}
		}
\end{smallequation*}
where $\Delta_2 = \{\pairQ {y_i}{C_i}\}_i$, $\Delta_3 = \{\pairQ {y_i}{D_i}\}_i$ 
			and $\Delta_2\oplus\Delta_3 = \{\phl{y_i}:\thl{C_i\oplus D_i}\}_i$;
\begin{smallequation*}
		\infer[\query\bang^\perp_p]{
			\syncseqbis
			{\srv xy{\client {x_1}{y_1}{}\mydots\client{x_n}{y_n}{P}}}
			{\query\Delta_1, \query\Delta_2, \pairQ x{\bang A}}
		}{
			\syncseqbis P
			{\triCtx{\query\Delta_1}yA, \Delta_2}
		}
\end{smallequation*}
where  $\query\Delta_2 = \{\pairQ{x_i}{\query B_i}\}_i$ and $\Delta_2 = \{\pairQ{y_i}{B_i}\}_i$.

\smallskip

Rules $\oplus\with^\perp_p$ and $\query\bang^\perp_p$ simulate the
back-and-forth interaction of rules $\oplus/\with$ and $\query/\bang$,
respectively. In that respect, note that the process terms are
obtained exactly in the form of arbiters.
The following Lemma shows that, given a forwarder with no additive or
exponential message in transit,
additive and exponential rules become admissible in the presence of
the new rules above.
We say that $\with\oplus_{p}^\perp$ and $\query\bang^\perp_p$ are {\em
	full}, written $\with\oplus^\perp$ and $\query\bang^\perp$, whenever
$\oplus\Delta_1$, or respectively $\query\Delta_1$, is empty.
%
We write
$\DD::\syncseqbis P{\Gamma}$ when $\mathcal D$ is
a derivation of the judgement $\syncseqbis P{\Gamma}$.
\begin{lemma}[$\with/\bang$ Elimination]\label{lem:broadcastingadmissible}
	Let $\DD::\syncseqbis P\Gamma$ such that $\Gamma$ has no
	$\Left/\Right/\Query$ box. There exists $\phl {P'}$ with
	$\DD'::\syncseqbis {P'}\Gamma$ and $\DD'$
	is free from $\bang$, $\query$, $\with$, $\oplus_l$ and $\oplus_r$.
\end{lemma}
\begin{proof}
	The first step of the proof is to replace all the $\with$- and the $\bang$-rules, 
	with $\with\oplus_p^\perp$ and $\bang\query_p^\perp$ respectively, where $\Delta_2$ and $\Delta_3$ are empty.
	The proof then proceeds by induction on the height of $\DD$,
        by doing a case analysis of the last applied rule in $\DD$. We
        consider only the two cases where the last rule is
        $\with\oplus_p^\perp$ or $\bang\query_p^\perp$, the other
        cases are obtained straightforwardly by permutation.

	If $\DD$ ends with $\with\oplus_p^\perp$, 
	$\phl{P =}\Case z{\inl {y_1}{}\ldots\inl {y_n}{P_1}}{\inr {y_1}{}\ldots\inr {y_n}{P_2}}$ with:
		\begin{smallequation*}
			\infer[\with\oplus_p^\perp]
			{\syncseqbis{P} 
				{\Gamma, \oplus\Delta_1, \pairQ x{A\oplus B}, \Delta_2\oplus\Delta_3, \pairQ z{C\with D}} 
			}{
				\syncseqbis {P_1} {\Gamma, \ltriCtx{\oplus\Delta_1, \pairQ x{A\oplus B}}{z}{D},  \Delta_2} 
				&
				\syncseqbis {P_2} {\Gamma, \rtriCtx{\oplus\Delta_1, \pairQ x{A\oplus B}}{z}{C}, \Delta_3} 
			}
		\end{smallequation*}
		where $\Gamma$ is free of $\Left/\Right/\Query$ boxes.
		By Lemma~\ref{lem:opluspermute} (applicable here since a proof in
		$\syncseqbis{}{}$ can always be expanded into a proof in
		$\syncseq{}{}$), we know that there exist $\phl{P_1'}$ and
		$\phl{P_2'}$ such that:
                \begin{smallequation*}
                    \begin{array}{lll}      
                      \infer[\oplus_l]{
                      \syncseq
                      {\inl x{P_1'}}
                      {\Gamma, \ltriCtx{\oplus\Delta_1, \pairQ x{A\oplus B}}{z}{C}, \Delta_2}
                      }{
                      \syncseq {P_1'}
                      {\Gamma, \ltriCtx{\oplus\Delta_1}{z}{C}, \pairQ x{A}, \Delta_2}
                      }
                      \\[2mm]
                          \infer[\oplus_r]{
                          \syncseq {\inr x{P_2'}} 
                          {\Gamma, \rtriCtx{\oplus\Delta_1, \pairQ x{A\oplus B}}{z}{D}, \Delta_3}
                          }{
                          \syncseq {P_2'} 
                          {\Gamma, \rtriCtx{\oplus\Delta_1}{z}{D}, \pairQ x{B}, \Delta_3}
                          }
                    \end{array}
		\end{smallequation*}
		We can take $\phl{P' =}\Case z{\inl {y_1}{}\ldots\inl {y_n}{\inl x{P_1'}}}{\inr {y_1}{}\ldots\inr {y_n}{\inr x{P_2'}}}$ as
		\begin{displaymath}\small
			\infer[\oplus\with^\perp_p]{
				\syncseqbis {P'}
				{\Gamma, \oplus\Delta_1, \pairQ x{A\oplus B}, \Delta_2\oplus\Delta_3, \pairQ z{C \with D}}
			}{ 
				\syncseqbis{P_1'}
				{\Gamma, \ltriCtx{\oplus\Delta_1}zC,  \pairQ xA, \Delta_2}
				& 
				\syncseqbis{P_2'}
				{\Gamma, \rtriCtx{\oplus\Delta_1}zD, \pairQ xB, \Delta_3}
			}
		\end{displaymath}
		Recursively, 
		we repeat this process until all elements of $\oplus\Delta_1$
		are exhausted and $\oplus\with^\perp_p$ becomes full. Then, the
		result follows by induction on the premisses, as they become free of exponential/additive boxes (and so is $\Gamma$).

		If the last applied rule in $\DD$ is $\bang\query_p^\perp$,  then
		$\phl{P = }\srv uv{\client {x_1}{y_1}{}\ldots\client{x_n}{y_n}{P_1}}$ with:
		\begin{displaymath}\small
			\infer[\query\bang^\perp_p]{
				\syncseqbis
				{P} 
				{\query\Delta_1, \pairQ x{\query B}, \query\Delta_2, \pairQ u{\bang A}} 
			}{
				\syncseqbis {P_1}
				{\triCtx{\query\Delta_1,\pairQ x{\query B}}vA,\Delta_2} 
			}
		\end{displaymath}
		By Lemma~\ref{lem:querypermute} (applicable for the same reason), there exists $\phl {P_1'}$ such that:
		\begin{displaymath}\small
			\infer[\query]{
				\syncseq {\client xy{P_1'}}
				{\triCtx{\query\Delta_1,\pairQ x{\query B}}vA,\Delta_2}
			}{
				\syncseq{P_1'}
				{\triCtx{\query\Delta_1}vA, \pairQ yB,\Delta_2}
			}
		\end{displaymath}
		which allows us to define $\phl{P' =} \srv uv{\client {x_1}{y_1}{}\ldots\client{x_n}{y_n}{\client xy{P_1'}}}$ justified by 
		\begin{displaymath}\small
			\infer[\query\bang^\perp_p]{
				\syncseqbis{P'}
				{\query\Delta_1, \pairQ x{\query B}, \query\Delta_2, \pairQ u{\bang A} }
			}{
				\syncseqbis {P_1'}
				{\triCtx{\query\Delta_1}vA, \pairQ yB,\Delta_2}
			}
		\end{displaymath}
		Recursively, we repeat this procedure until all of
                $\query\Delta_1$ is exhausted, obtaining a full
                rule. The rest follows by induction. \qed
\end{proof}

\mypar{Multiplicatives and Units.}  We are now ready to undertake the
last step of our transformation, i.e., dealing with multiplicatives
and units. In such cases, given that the nature of the communication
is gathering, we need to push all $\parr$s up to their corresponding
$\tensor$ and all $\perp$s up to their corresponding $\one$. For this
purpose, we change $\syncseqbis{}{}$ to a new jugement, dubbed
$\syncseqtris{}{}$, where we remove rules $\with$, $\oplus_l$,
$\oplus_r$, $?$ and $!$, replace them by $\with\oplus^\perp$ and
$\query\bang^\perp$, and add the rules:
\begin{smallequation*}
	\infer[\tensor\parr^\perp_p]{
		\syncseqtris {\recv {x_1}{u_1}{}\ldots\recv{x_n}{u_n}{\send yvPQ}}
		{\Gamma, [\Delta_1]\Delta_2, \{\phl{x_i} : \thl{A_i \parr B_i}\}_i, \pairQ y{C \tensor D}}
	}{
		\syncseqtris {P} {\Delta_1, \{\pairQ{u_i}{A_i}\}_i,  \pairQ vC}
		&
		\syncseqtris {Q} 
		{\thl {\Gamma}, \Delta_2, \{\phl{x_i}:\thl{B_i}\}_i,  \phl y:\thl D}
	}
\end{smallequation*}
and
\begin{smallequation*}
	\infer[\one\bot^\perp_p]
	{
		\syncseqtris {\wait {x_1}{}\ldots\wait{x_n}{\close y}}
		{
			\{\phl {x_i}:\thl{\perp}\}_i, \phl{y}:\thl{\one}, \Star
		}
	}
	{
	}
\end{smallequation*}
Rules $\tensor\parr^\perp_p$ and $\one\bot^\perp_p$ simulate the
interplay of rules $\tensor/\parr$ and $\one/\bot$ respectively.
Ultimately, we will use the {\em full} versions of these, written
$\tensor\parr^\perp$ and $\one\bot^\perp$, that designate whenever
$\br{\Delta_1}\Delta_2$ is empty in the former, and whenever $\Star$
is absent in the latter.
Similarly to what we did for additives and exponentials, we show now
that we can replace the remaning rules from $\syncseq{}{}$ by compound
ones. 

\begin{lemma}[$\parr$/$\perp$ Admissibility]\label{lem:parradmissibility}\label{lem:perpadmissibility}
	\begin{enumerate}
	\item \label{i}
	Let $\DD::\syncseqtris P{\Gamma, \quadQ yAxB}$
	such that $\mathcal D$ is $\parr$-free. Then, there exist a forwarder $\phl Q$ and a
	$\parr$-free proof $\EE$ such that
	$\EE::\syncseqtris {Q}{\Gamma, \pairQ x{A\parr
			B}}$.
		
	\item\label{ii}
	Let $\DD::\syncseqtris P{\Gamma, \Star}$ such
	that $\DD$ is $\perp$-free. Then, there exists a $\perp$-free
	proof $\EE$ and $\phl Q$ such that
	$\EE::\syncseqtris {Q}{\Gamma, \pairQ x\perp}$, for
	some $\phl x$.
	\end{enumerate}
\end{lemma}
\begin{proof}
	We proceed by induction on the size of $\DD$ and a case analysis on the
	last applied rule. 
	We only report on case in the proof of item~\ref{i} when the last applied rule in $\DD$ is $\tensor\parr^\perp_p$,
	that is,
	$\phl{P =} \recv {x_1}{u_1}{}\ldots\recv{x_n}{u_n}{\send zv{P_1}{P_2}}$ and we have two possible subcases. (The other ones are simpler.) 

		If $\quadQ yAxB$ is not touched by $\tensor\parr^\perp_p$, we have:
		\begin{smallequation*}
			\infer[\tensor\parr^\perp_p]{
				\syncseqtris {P}
				{\Gamma, \quadQ yAxB, [\Delta_1]\Delta_2, \{\phl{x_i} : \thl{A_i \parr B_i}\}_i, \pairQ z{C \tensor D}}
			}{
				\syncseqtris {P_1} {\Delta_1, \{\pairQ{u_i}{A_i}\}_i,  \pairQ vC}
				& 
				\syncseqtris {P_2} {\thl {\Gamma}, \quadQ yAxB, \Delta_2, \{\phl{x_i}:\thl{B_i}\}_i, \phl z:\thl D}
			}
		\end{smallequation*}
		By induction hypothesis, since our proof is $\parr$-free, 
		there exists $\phl{Q_2}$ such that 
		
		\noindent $\syncseqtris {Q_2} {\thl {\Gamma}, \pairQ{x}{A\parr B}, \Delta_2, \{\phl{x_i}:\thl{B_i}\}_i, \phl z:\thl D}$.
		By applying $\tensor\parr^\perp_p$ again, we obtain $\phl{Q =} \recv {x_1}{u_1}{}\ldots\recv{x_n}{u_n}{\send zv{P_1}{Q_2}}$ and:
		\begin{smallequation*}
			\infer[\tensor\parr^\perp_p]{
				\syncseqtris {Q}
				{\Gamma, \pairQ x{A\parr B}, [\Delta_1]\Delta_2, \{\phl{x_i} : \thl{A_i \parr B_i}\}_i, \pairQ z{C \tensor D}}
			}{
				\syncseqtris {P_1} {\Delta_1, \{\pairQ{u_i}{A_i}\}_i,  \pairQ vC}
				& 
				\syncseqtris {Q_2} {\thl {\Gamma}, \pairQ x{A\parr B}, \Delta_2, \{\phl{x_i}:\thl{B_i}\}_i, \phl z:\thl D}
			}
		\end{smallequation*}
		
		If $\quadQ yAxB$ is indeed modified by
		$\tensor\parr^\perp_p$, then we have a base case:
		\begin{smallequation*}
			\infer[\tensor\parr^\perp_p]{
				\syncseqbis {P}
				{\Gamma, \quadQ yAxB, [\Delta_1]\Delta_2, \{\phl{x_i} : \thl{A_i \parr B_i}\}_i, \pairQ z{C \tensor D}}
			}{
				\syncseqbis {P_1} {\phl y: \thl A, \Delta_1, \{\pairQ{u_i}{A_i}\}_i,  \pairQ vC}
				& 
				\syncseqbis {P_2} {\thl {\Gamma}, \phl x: \thl B, \Delta_2, \{\phl{x_i}:\thl{B_i}\}_i, \phl z:\thl D}
			}
		\end{smallequation*}
		We can take $\phl{Q = \recv xyP}$ and simply change the rule $\tensor\parr^\perp_p$ obtaining: 
		\begin{smallequation*}
			\infer[\tensor\parr^\perp_p]{
				\syncseqbis {Q}
				{\Gamma, \pairQ x{A\parr B}, [\Delta_1]\Delta_2, \{\phl{x_i} : \thl{A_i \parr B_i}\}_i, \pairQ y{C \tensor D}}
			}{
				\syncseqbis {P_1} {\phl y: \thl A, \Delta_1, \{\pairQ{x_i}{A_i}\}_i,  \pairQ yC}
				& 
				\syncseqbis {P_2} {\thl {\Gamma}, \phl x: \thl B, \Delta_2, \{\phl{x_i}:\thl{B_i}\}_i, \phl y:\thl D}
			}
		\end{smallequation*}
	
	The proof for item~\ref{ii} is similar.
	Note that in the base case, we can have
\begin{smallequation*}
		\DD = \infer[\one\bot^\perp_p]{ \syncseqbis {\phl{P =}\wait {x_1}{}\ldots\wait{x_n}{\close y}} {\{\phl
			{x_i}:\thl{\perp}\}_i, \phl{y}:\thl{\one}, \Star} }{} 
\end{smallequation*}
	%
	%
	which gives us $\EE =
	\infer[\one\bot^\perp]{
		\syncseqbis {\phl{Q =} \wait{x}{\wait {x_1}{}\ldots\wait{x_n}{\close y}}}
		{\{\phl {x_i}:\thl{\perp}\}_i,  \phl{y}:\thl{\one},  \pairQ x\perp}
	}
	{}
	$.\qed
\end{proof}

As a corollary, we can always eliminate all $\tensor$ with their corresponding $\parr$s and all the $\one$ with their corresponding $\bot$s.
\begin{lemma}[$\tensor$/$\one$ Elimination]
	\label{lem:parrelimination}\label{lem:perpelimination}
	Let $\DD::\syncseqtris P{\Gamma}$. Then, there
        exist a forwarder $\phl Q$ with 
        $\mathcal E::\syncseqtris {Q}{\Gamma}$, and $\EE$ is free from $\parr$, $\tensor$,
        $\perp$, and $\one$.
\end{lemma}
\begin{proof}
	It follows from replacing any instance of $\tensor$ and $\one$ with $\tensor\parr^\perp_p$ (with empty $\Delta_1$ and $\Delta_2$) 
	and $\one\bot^\perp_p$, respectively;
	then, applying the previous Lemma repeatedly  to 
	the top-most instances of $\parr$ or $\bot$ first.
\end{proof}
\begin{theorem}[Completeness]
	If $\syncseq{P}{\Delta}$, then there exists a global type $\phl G$,
	s.t. $\ghl G \gseq {\Delta^\perp}$.
\end{theorem}
\begin{proof}
  It follows from the previous results, noting that in order to get
  coherence, all rules from $\syncseqtris{}{}$ must be full, which is
  the case since our context is a basic $\Delta$.  We observe that in
  the case of multiplicatives, we need to perform a name substitution
  in order to obtain a valid coherence proof. Below, let
  $\coextract{P}$ be the function that transforms a process term
  $\phl P$ corresponding to a proof in $\syncseqtris{}{}$ with only
  full rules to a global type. Then,
  \begin{smallequation*}
    \infer[\tensor\parr^\perp]{
      \syncseqtris {\recv {x_1}{u_1}{}\ldots\recv{x_n}{u_n}{\send yvPQ}}
      {\Delta, \{\phl{x_i} : \thl{A_i \parr B_i}\}_i, \pairQ y{C \tensor D}}
    }{
      \syncseqtris {P} { \{\pairQ{u_i}{A_i}\}_i,  \pairQ vC}
      &
      \syncseqtris {Q} 
      {\thl {\Delta},  \{\phl{x_i}:\thl{B_i}\}_i,  \phl y:\thl D}
    }
  \end{smallequation*}
  is transformed into
  \begin{smallequation*}
    \infer[\tensor\parr] {\phl{\parrtensor{\til x}{y}{\coextract
          {P'}}{\coextract{Q}}} \gseq \thl \Delta, \{\phl{x_i} :
      \thl{A_i \tensor B_i}\}_i, \phl{y}:\thl{C \parr D}} { \phl
      {\coextract{ P'}} \gseq \thl \{\phl{x_i}:\thl{A_i}\}_i, \phl
      y:\thl C & \phl {\coextract{Q}} \gseq \thl {\Delta},
      \{\phl{x_i}:\thl{B_i}\}_i, \phl y:\thl D}
  \end{smallequation*}
  such that $\phl{P'} = \phl{P\substtwo {y}{v}{\til x}{\til u}}$. \
  \qed
\end{proof}

\begin{example}
  Let us consider a variation of process $\phl{P_1}$ from
  Example~\ref{ex:2buf}: 
    \begin{smallequation*}
    \begin{array}{lllll}
      \begin{array}[t]{@{}l}
        \text{\phl{(as $P_1$)\ldots}}\quad
        \recv {b_1'}{y}{}\
        \recv {s'}{x_2}{}\
        {\send{b_2'}{x'_2}{\fwd{x_2}{x'_2}}{}} \
        {\send{b_2'}{y'}{\fwd{y}{y'}}{}}\ \phl{\ldots\text{(as $P_1$)}} 
      \end{array} 
    \end{array}
  \end{smallequation*}
  The process above still enforces the same protocol, we can transform
  it into $\phl{P_1}$ by permuting some actions, and then into the
  coherence proof in Example~\ref{ex:2bu}. \qed
\end{example}

%
%

\section{Related Work}\label{sec:related}
%
Our work takes~\cite{CLMSW16} as a starting point.
Guided by CLL, we set out to explore if coherence can be broken down
into more elementary logical rules which led to the discovery of
synchronous forwarders.
Caires and Perez~\cite{CP16} also study multiparty session types in
the context of intuitionistic linear logic by translating global types
to processes, called \emph{mediums}.
Their work does not start from a logical account of global types
(their global types are just syntactic terms). But, as in this paper
and previous work~\cite{CLMSW16}, they do generate arbiters as linear
logic proofs. In this work, we also achieve the converse: from a forwarder
process, we provide a procedure for generating a global type
(coherence proof).

%

Sangiorgi~\cite{S96}, probably the first to treat forwarders for the
$\pi$-calculus, uses binary forwarders, i.e., processes that only
forward between two channels, which is equivalent to our $\fwd xy$. We
attribute our result to the line of work that originated in 2010 by
Caires and Pfenning~\cite{CP10},
where forwarders \`a la Sangiorgi were introduced as processes to be
typed by the axiom rule in linear logic. Van den Heuvel and
Perez~\cite{Heuvel2020SessionTS} have recently developed a version of
linear logic that encompasses both classical and intuitionistic logic,
presenting a unified view on binary forwarders in both logics.

Gardner et al.~\cite{GLW07} study the expressivity of the linear
forwarder calculus, by encoding the asynchronous $\pi$-calculus (since
it can encode distributed choice).  The linear forwarder calculus is a
variant of the (asynchronous) $\pi$-calculus that has binary
forwarders and a restriction on the input $x(y).P$ such that $y$
cannot be used for communicating (but only forwarded). Such a
restriction is similar to the intuition behind synchronous forwarders,
with the key difference that it would not work for some of our
session-based primitives. 

Barbanera and Dezani~\cite{BD19} study multiparty session types as
gateways which are basically forwarders that work as a medium among
many interacting parties, forwarding communications between two
multiparty sessions. Such mechanism reminds us of our forwarder
composition: indeed, in their related work discussion they do mention
that their gateways could be modeled by 
a ``connection-cut''.

Recent work~\cite{JGP16,GJP18} proposes an extension of linear logic
that models {\em identity providers}, a sort of monitoring mechanisms
that are basically forwarders between two channels in the sense of
Sangiorgi, but asynchronous, i.e., they allow unbounded buffering of
messages before forwarding.

Our forwarder mechanism may be confused with that of
locality~\cite{MS04}, which is discussed from a logical point of view
by Caires et al.~\cite{CPT16}. Locality only requires that received
channels cannot be used for inputs (that can only be done at the
location where the channel was created). 
In our case instead, we do not allow received channels to be used at
all until a new forwarder is created.


The transformations between coherence and synchronous forwarders are
related to those of projection and extraction for choreographies. A
choreography is basically a global description of a the sequence of
interactions (communications) that must happen in a distributed system
(like a global type). Carbone et al.~\cite{CMS14} give a
characterisation of this in intuitionistic linear logic, by using
hypersequents to represent both choreographies and the processes
corresponding to those choreographies: through proof transformations
they show how to go from choreographies to processes and vice
versa. Although we
also 
transform choreographies (coherence) intro processes, our forwarder is
a single point of control while they deal with a distributed
implementation.

\section{Discussion and Future Work}\label{sec:discussion}

\mypar{Coherence Compositionality.}  The results of this paper
give us compositionality (cut) and cut elimination also for coherence
proofs. In fact, we can always transform two coherence proofs into
synchronous forwarders, compose and normalise them, and finally
translate them back to coherence.


\mypar{Process Language.} Our process language is based on that
of~\cite{W14} with some omissions. For the sake of presentation, we
have left out polymorphic communications. We believe that these
communication primitives, together with polymorphic types
$\exists X.A$ and $\forall X.A$, can be added to synchronous
forwarders. Moreover, our process language does not support recursion
for coherence nor for synchronous
forwarders.  
We leave these points as future work.

\mypar{Classical vs Intuitionistic Linear Logic.} In this paper, we
have chosen to base our theory on CLL for two main reasons.  Coherence
is indeed defined by Carbone et al.~\cite{CLMSW16} in terms of CLL and
therefore our results can immediately be related to theirs without
further investigations. 
%
We would like to remark that an earlier version of synchronous
forwarders was based on intuitionistic linear logic, but moving to CLL
required many fewer rules and greatly improved the presentation of our
results.  Nevertheless, our results can be easily reproduced in
intuitionistic linear logic, including a straightforward adaptation of
coherence.

\mypar{Exponentials, Weakening and Contraction.} In our work, we
do not allow synchronous forwarders to harness the full power of
exponentials, because we disallow the use of weakening and
contraction. In some sense exponentials are used linearly. Weakening
allows us to extend the context by fresh channels provided that they
are $?$-quantified.  Weakening is also useful when composing a process
offering some service of type $!A$ with some process that does not
wish to use such service. Contraction, on the other hand, models
server duplication, i.e., creating a copy of a server for every
possible client.  Guided by the given definition of coherence, we are
sure that neither weakening nor contraction reflect our intuition of
synchronous forwarders.  In fact, we would like to remark that adding
these rules to synchronous forwarders would invalidate the
completeness result. This is because coherence is apparently
incompatible with weakening and contraction of assumptions in the
context.  We leave a further investigation of how and if to add
weakening and contraction to future work.

\mypar{Unlimited-Size Buffers.} Synchronous forwarders guarantee that
the order of messages between two endpoints is preserved. This is
achieved by preventing the sending endpoint from sending further
messages until the previous message has been forwarded. 
As future work, we wish to consider buffers of any size, i.e., a
sender can keep on sending messages that can be stored in the
forwarder and then forwarded at a later time. Our idea is to
generalise, e.g., a boxing $\quadQ yAxB$ to allow for more messages to
be stored as in $\quadQ {y_1}{A_1\ldots\pairQ {y_n}{A_n}}xB$, and at
the same time allowing $\pairQ xB$ to be used. At first this may seem
a simple extension of synchronous forwarders. However, it has major
implications in the proof of cut elimination that we must better
understand. Note that the system of forwarders by Gommerstadt et
al.~\cite{JGP16,GJP18} does have unlimited-size buffers, but it is
restricted to binary forwarders.  In this case, the proof of cut
elimination is standard. Unfortunately, adding just a third endpoint
having unbounded queues breaks the standard structural proof
($\tensor$ and $\Did{Cut}$ do not commute as in
page~\pageref{page:tensorcurcommute}).

\mypar{Complete Interleaving of Synchronous Forwarders.} In
synchronous forwarders, the rule for $\with$ requires that a non-empty
set of formulas $\Delta = \{A_i \oplus B_i\}$ is selected from the
context and boxed, effectively forcing the processes to interact with
the choice.  This design decision was necessary to achieve cut
elimination as well as our completeness theorem for mapping
synchronous forwarders to coherence. But it comes at a price: it
restricts the number of proofs, and consequently, there are processes
that are still forwarders, implement a 1-size buffer, their CLL type
is coherent, but are not typable in $\syncseq{}{}$. This is because
the order of communications interferes with the $\with$ rule. For
example, the process $\recv zy{}\Case x {\ \inl zP} {\ \inr zQ}$ (for
some adequate $\phl P$ and $\phl Q$) is a synchronous
forwarder. However, the process
$\Case x {\ \recv zy{}\inl zP} {\ \recv zy{}\inr zQ}$ is typable in
CLL but is not a synchronous forwarder. Note that the input on
$\phl z$ is totally unrelated to the branching and, moreover, the
typing context of the process is coherent. Unfortunately, attempts to
include such cases have broken cut elimination or completeness. We
leave a further investigation 
to future work.

\mypar{Variants of Coherence.} Our results show that synchronous
forwarders are also coherent. As a follow-up, we would like to
investigate in future work, whether generalised variants of forwarders
also induce interesting generalised notions of coherence, and, as a
consequence, generalisations of global types.




%

\section{Conclusions}\label{sec:conclusions}
To our knowledge, this work is the first to give characterisation of
coherence in terms of forwarders which generalise the concept of
arbiter.
%
%
%
We have developed a proof system based on linear logic that models a
class of forwarders, called synchronous forwarders, that preserve
message order. Well typed-forwarders are shown to be compositional. We
show that synchronous forwarders provide a sound and complete
characterisation of coherence and therefore provide a logic of global
types and a protocol language for describing distributed protocols.
%
%


\newpage
\bibliographystyle{plain}
\bibliography{biblio}

\begin{thebibliography}{10}

\bibitem{BD19}
Franco Barbanera and Mariangiola Dezani{-}Ciancaglini.
\newblock Open multiparty sessions.
\newblock In Massimo Bartoletti, Ludovic Henrio, Anastasia Mavridou, and
  Alceste Scalas, editors, {\em Proceedings 12th Interaction and Concurrency
  Experience, {ICE} 2019, Copenhagen, Denmark, 20-21 June 2019}, volume 304 of
  {\em {EPTCS}}, pages 77--96, 2019.

\bibitem{CP16}
Lu{\'{\i}}s Caires and Jorge~A. P{\'{e}}rez.
\newblock Multiparty session types within a canonical binary theory, and
  beyond.
\newblock In Elvira Albert and Ivan Lanese, editors, {\em Formal Techniques for
  Distributed Objects, Components, and Systems - 36th {IFIP} {WG} 6.1
  International Conference, {FORTE} 2016, Held as Part of the 11th
  International Federated Conference on Distributed Computing Techniques,
  DisCoTec 2016, Heraklion, Crete, Greece, June 6-9, 2016, Proceedings}, volume
  9688 of {\em Lecture Notes in Computer Science}, pages 74--95. Springer,
  2016.

\bibitem{CP10}
Lu\'{\i}s Caires and Frank Pfenning.
\newblock Session types as intuitionistic linear propositions.
\newblock In {\em CONCUR}, pages 222--236, 2010.

\bibitem{CPT16}
Lu{\'{\i}}s Caires, Frank Pfenning, and Bernardo Toninho.
\newblock Linear logic propositions as session types.
\newblock {\em Math. Struct. Comput. Sci.}, 26(3):367--423, 2016.

\bibitem{CLMSW16}
Marco Carbone, Sam Lindley, Fabrizio Montesi, Carsten Sch{\"{u}}rmann, and
  Philip Wadler.
\newblock Coherence generalises duality: {A} logical explanation of multiparty
  session types.
\newblock In Jos{\'{e}}e Desharnais and Radha Jagadeesan, editors, {\em 27th
  International Conference on Concurrency Theory, {CONCUR} 2016, August 23-26,
  2016, Qu{\'{e}}bec City, Canada}, volume~59 of {\em LIPIcs}, pages
  33:1--33:15, Germany, 2016. Schloss Dagstuhl - Leibniz-Zentrum f{\"{u}}r
  Informatik.

\bibitem{CMS14}
Marco Carbone, Fabrizio Montesi, and Carsten Sch{\"{u}}rmann.
\newblock Choreographies, logically.
\newblock In {\em {CONCUR}}, pages 47--62, 2014.

\bibitem{CMSY15}
Marco Carbone, Fabrizio Montesi, Carsten Sch{\"{u}}rmann, and Nobuko Yoshida.
\newblock Multiparty session types as coherence proofs.
\newblock In {\em CONCUR}, pages 412--426, 2015.

\bibitem{GLW07}
Philippa Gardner, Cosimo Laneve, and Lucian Wischik.
\newblock Linear forwarders.
\newblock {\em Inf. Comput.}, 205(10):1526--1550, 2007.

\bibitem{G87}
Jean{-}Yves Girard.
\newblock Linear logic.
\newblock {\em Theor. Comput. Sci.}, 50:1--102, 1987.

\bibitem{GJP18}
Hannah Gommerstadt, Limin Jia, and Frank Pfenning.
\newblock Session-typed concurrent contracts.
\newblock In Amal Ahmed, editor, {\em Programming Languages and Systems - 27th
  European Symposium on Programming, {ESOP} 2018, Held as Part of the European
  Joint Conferences on Theory and Practice of Software, {ETAPS} 2018,
  Thessaloniki, Greece, April 14-20, 2018, Proceedings}, volume 10801 of {\em
  Lecture Notes in Computer Science}, pages 771--798. Springer, 2018.

\bibitem{HVK98}
Kohei Honda, Vasco Vasconcelos, and Makoto Kubo.
\newblock Language primitives and type disciplines for structured
  communication-based programming.
\newblock In {\em ESOP}, pages 22--138, 1998.

\bibitem{HYC08}
Kohei Honda, Nobuko Yoshida, and Marco Carbone.
\newblock Multiparty asynchronous session types.
\newblock In George~C. Necula and Philip Wadler, editors, {\em Proceedings of
  the 35th {ACM} {SIGPLAN-SIGACT} Symposium on Principles of Programming
  Languages, {POPL} 2008, San Francisco, California, USA, January 7-12, 2008},
  pages 273--284. {ACM}, 2008.

\bibitem{HYC16}
Kohei Honda, Nobuko Yoshida, and Marco Carbone.
\newblock Multiparty asynchronous session types.
\newblock {\em JACM}, 63(1):9, 2016.
\newblock Also: POPL, 2008, pages 273--284.

\bibitem{JGP16}
Limin Jia, Hannah Gommerstadt, and Frank Pfenning.
\newblock Monitors and blame assignment for higher-order session types.
\newblock In Rastislav Bod{\'{\i}}k and Rupak Majumdar, editors, {\em
  Proceedings of the 43rd Annual {ACM} {SIGPLAN-SIGACT} Symposium on Principles
  of Programming Languages, {POPL} 2016, St. Petersburg, FL, USA, January 20 -
  22, 2016}, pages 582--594. {ACM}, 2016.

\bibitem{MS04}
Massimo Merro and Davide Sangiorgi.
\newblock On asynchrony in name-passing calculi.
\newblock {\em Mathematical Structures in Computer Science}, 14(5):715–767,
  2004.

\bibitem{MPW92}
Robin Milner, Joachim Parrow, and David Walker.
\newblock A calculus of mobile processes, {I and II}.
\newblock {\em Information and Computation}, 100(1):1--40,41--77, September
  1992.

\bibitem{S96}
Davide Sangiorgi.
\newblock pi-calculus, internal mobility, and agent-passing calculi.
\newblock {\em Theor. Comput. Sci.}, 167(1{\&}2):235--274, 1996.

\bibitem{Heuvel2020SessionTS}
Bas van~den Heuvel and Jorge~A. P{\'{e}}rez.
\newblock Session type systems based on linear logic: Classical versus
  intuitionistic.
\newblock In Stephanie Balzer and Luca Padovani, editors, {\em Proceedings of
  the 12th International Workshop on Programming Language Approaches to
  Concurrency- and Communication-cEntric Software, PLACES@ETAPS 2020, Dublin,
  Ireland, 26th April 2020}, volume 314 of {\em {EPTCS}}, pages 1--11, 2020.

\bibitem{V12}
Vasco~T. Vasconcelos.
\newblock Fundamentals of session types.
\newblock {\em Inf. Comput.}, 217:52--70, 2012.

\bibitem{W12}
Philip Wadler.
\newblock Propositions as sessions.
\newblock In {\em ICFP}, pages 273--286, 2012.

\bibitem{W14}
Philip Wadler.
\newblock Propositions as sessions.
\newblock {\em Journal of Functional Programming}, 24(2--3):384--418, 2014.
\newblock Also: ICFP, pages 273--286, 2012.

\end{thebibliography}

\newpage
\appendix
\section{CLL~\cite{W12} with simple exponentials}\label{app:cll}
  \begin{displaymath}
    \small
    \begin{array}{c}
      \infer[\textsc{Axiom}]
      {\phl{\cpaxiom{x}{A^\perp}{y}{A}} \seq \phl x:\thl {A^\bot}, \phl y:\thl{A}}
      { }
		\qquad
		\infer[\bot]
		{\phl {\wait xP} \seq \thl \Gamma, \phl x:\thl\bot}
		{\phl P \seq \thl\Gamma}
		\qquad
		\infer[\one]
		{\phl{\close x} \seq \phl x: \thl\one}
		{ }
      \\[1ex]
      \infer[\tensor]
      {
        \phl{\send xyPQ} 
        \seq 
        \thl{\Gamma}, \thl{\Delta}, \phl x:\thl{A \tensor B}
      }
      {
        \phl P 
        \seq 
        \thl{\Gamma}, \phl y:\thl A
        & \phl Q 
        \seq 
        \thl \Delta, \phl x: \thl B
      }
      \qquad\qquad\qquad
      \infer[\parr]
      {\phl{\recv xyP} 
        \seq \thl{\Gamma}, \phl x: \thl {A \parr B}
      }       
      {\phl P \seq \thl{\Gamma}, \phl y:\thl A, \phl x:\thl B}
      \\[1ex]
      \infer[\oplus_l]
      {\phl{\inl xP} \seq \thl\Gamma, \phl x:\thl{A \oplus B}}
      {\phl P \seq \thl\Gamma, \phl x:\thl A}
      \qquad
      \infer[\oplus_r]
      {\phl{\inr xP} \seq \thl\Gamma, \phl x:\thl{A \oplus B}}
      {\phl P \seq \thl\Gamma, \phl x:\thl B}
      \qquad
      \infer[\with]
      {\phl{\Case xPQ} \seq \thl \Gamma, \phl x:\thl{A \with B}}
      {\phl P \seq \thl\Gamma, \phl x:\thl A & \phl Q \seq
        \thl\Gamma, \phl
        x:\thl B}
      \\[1ex]
      \infer[?]
      {\phl{\client xyP} \seq \thl \Gamma, \phl x:\thl {\query A}}
      {\phl P \seq \thl \Gamma, \phl y:\thl A}
      \qquad\qquad
      \infer[!]
      {\phl{\srv xyP} \seq \,\thl{\query \Gamma}, \phl x:\thl{\bang A}}
      {\phl P \seq \,\thl{\query \Gamma}, \phl y:\thl A}
         \\[1ex]
         \qquad\qquad
         \infer[\textsc{Cut}]
         {\phl{\cpres{x}A{y}{A^\perp}(P \pp Q)} \seq \thl\Gamma,\thl\Delta}
         {\phl P \seq \thl\Gamma, \phl x:\thl A & \phl Q \seq \thl\Delta,
         	\phl y:\thl{\dual{A}}}
      \end{array}
  \end{displaymath}

\newpage
\section{Cut Elimination Proof for Synchronous
  Forwarders}\label{app:synccutelim}
\newcommand{\ax}[0]{\textsc{Ax}}
\newcommand{\Lbr}[1]{\Left\dbr{#1}}
\newcommand{\Rbr}[1]{\Right\dbr{#1}}
\newcommand{\Qbr}[1]{\Query\dbr{#1}}
\newcommand{\red}[1]{\textcolor{red}{#1}}

%
%
%


\begin{theorem}[Cut-admissibility] \ \par
  \begin{enumerate}  
  \item If $\DD:: \Gamma_1, A$
    and $\EE:: \Gamma_2,\dual{A}$
    then $ \Gamma_1,\Gamma_2$.\label{cuta}
  \item If $\DD:: \Delta_1, A$
    and $\EE:: \Delta_2, \Gamma_2, B$
    and $\FF:: \Gamma_3,[A^\perp]B^\perp$
    then $ [\Delta_1]\Delta_2,\Gamma_2, \Gamma_3$.\label{cutb}

  \item If $\DD:: \Gamma_1, \Lbr{\Delta_1, A \oplus B} C$
     and $ \EE :: \Gamma_2, \Lbr{\Delta_2} A^\perp$
     and $ \FF :: \Gamma_2, \Rbr{\Delta_2} B^\perp$
    then $ \Gamma_1,  \Gamma_2, \Lbr{\Delta_1, \Delta_2} C$.\label{cs1}
  \item If $\DD ::\Gamma_1, \Lbr{\Delta_1} C, A$
    and $\EE :: \Gamma_2, \Lbr{\Delta_2} A^\perp$
    then $ \Gamma_1,  \Gamma_2, \Lbr{\Delta_1, \Delta_2}C$.\label{cs11}
  \item If $\DD:: \Gamma_1, \Rbr{\Delta_1, A \oplus B} C$
     and $ \EE :: \Gamma_2, \Lbr{\Delta_2} A^\perp$
     and $ \FF :: \Gamma_2, \Rbr{\Delta_2} B^\perp$
    then $ \Gamma_1,  \Gamma_2, \Rbr{\Delta_1,  \Delta_2} C$.\label{cs2}
  \item If $\DD :: \Gamma_1, \Rbr{\Delta_1} C, B$
    and $\EE :: \Gamma_2, \Rbr{\Delta_2}B^\perp$
    then $ \Gamma_1,  \Gamma_2, \Rbr{\Delta_1, \Delta_2}C $.\label{cs22}


  \item If  $\DD :: \Gamma_1, \Qbr{\Delta_1,\query A}C$
    and $\EE :: \Qbr{\Delta_2} A^\perp$
    then $\Gamma_1, \Qbr {\Delta_1,\Delta_2}C$.\label{cs3}
  \item If $\DD :: \Gamma_1, \Qbr{\Delta_1}C, A$
    and $\EE :: \Gamma_2, \Qbr{\Delta_2} A^\perp$
    then $\Gamma_1,\Gamma_2, \Qbr{\Delta_1,\Delta_2} C$.\label{cs33}
  \end{enumerate}
\end{theorem}

\begin{proof}   by induction over the cut formula and the left and right derivation.
  \begin{enumerate}
  \item  By induction on $\DD$
    \begin{description}
    \item [Impossible Cases:] 
    \item [Axiom Case:]
      \[
        \DD =
        \infer[\ax]
        {A^\perp, A}
        {}
      \]
      \begin{tabbing}
        $A^\perp, \Gamma_2$ \` by $\EE$ 
      \end{tabbing}
    \item [Key Case:]
      \[
        \DD = \infer[\one]
        {\one, [\ast]^n}
        {}
      \]
      and
      \[
        \EE =
        \infer[\bot]
        {\Gamma_2, \bot}
        {\EE_1 :: \Gamma_2, [\ast]}
      \]
      \begin{tabbing}
        $\Gamma_2, [\ast]$ \` by $\EE_1$
      \end{tabbing}
    \item [Key Case:] 
      \[\DD =
        \infer[\tensor]
        {\Gamma_1, [\Delta_1]\Delta_2, A \tensor B}
        { \DD_1 :: \Delta_1, A
          \quad 
          \DD_2 :: \Gamma_1, \Delta_2, B}
      \] 
      and 
      \[
        \EE = 
        \infer[\parr]
        {\Gamma_2, A^\perp \parr B^\perp}
        {\EE_1 :: \Gamma_2, [A^\perp] B^\perp} 
      \]
        \begin{tabbing}
          $[\Delta_1] \Delta_2, \Gamma_1, \Gamma_2$ \` by i.h.~\ref{cutb} on $\DD_1$, $\DD_2$, and $\EE_1$
      \end{tabbing}
    \item [Key Case:] \[\DD =     \infer[\with]
        {\Gamma_1,A \oplus B, \oplus\Delta_1,C \with D}
        {\DD_1 :: \Gamma_1, \Lbr{A \oplus B,\oplus\Delta_1}C
          \quad
          \DD_2 :: \Gamma_1, \Rbr{A \oplus B,\oplus\Delta_1} D}
      \] and \[ \EE =     \infer[\with]
        {\Gamma_2,\oplus\Delta_2,A^\perp \with B^\perp}
        {\EE_1 :: \Gamma_2, \Lbr{\oplus\Delta_2}A^\perp
          \quad
          \EE_2 :: \Gamma_2, \Rbr{\oplus\Delta_2} B^\perp} \]
      \begin{tabbing}
        $ \GG_1 :: \Gamma_1,  \Gamma_2, \Lbr{\oplus \Delta_1, \oplus \Delta_2} C$ \` by i.h.~\ref{cs1} on $\DD_1$ and $\EE_1$ and $\EE_2$ \\
        $ \GG_2 :: \Gamma_1,  \Gamma_2, \Rbr{\oplus \Delta_1, \oplus \Delta_2} D$ \` by i.h.~\ref{cs2} on $\DD_2$ and $\EE_1$ and $\EE_2$ \\
        $ \Gamma_1,  \Gamma_2, \oplus \Delta_1, \oplus \Delta_2, C \with D$ \` by $\with$ on $\GG_1$ and $\GG_2$
      \end{tabbing}
    \item [Key Case:]
      \[\DD =   \infer[\bang]
      {\query \Delta_1, \query A, \bang C} 
      {\DD_1
        :: \Qbr{\query\Delta_1, \query A} C}
    \]
    and
    \[\EE =
      \infer[\bang]
      {\query {\Delta_2}, \bang A^\perp} 
      { \EE_1 :: \Qbr{\query\Delta_2} A^\perp}
    \]

    \begin{tabbing}
      $\GG :: \Qbr {\query\Delta_1,\query\Delta_2}C$  \` by i.h.~\ref{cs3} on $\DD_1$ and $\EE_1$\\
      $?\Delta_1,\query\Delta_2, \bang C$  \` by $\bang$ on $\GG$\\
    \end{tabbing}

  \item [Left-Commutative Case:]
    \[\DD  =
      \infer[\bot]
      {\Gamma_1,A, \bot}
      {\DD_1 :: \Gamma_1,A, [\ast]}
    \]
    \begin{tabbing}
      $\GG ::\Gamma_1, \Gamma_2, [\ast]$ \` by i.h.~\ref{cuta} on $\DD_1$ and $\EE$ \\
      $\Gamma_1, \Gamma_2, \bot$ \` by $\bot$ on $\GG$
    \end{tabbing}
 \item [Left-Commutative Case:]
   \[\DD  =
     \infer[\tensor]
     {\Gamma_1, A, [\Delta_1]\Delta_2, D \tensor E}
     { \DD_1 :: \Delta_1, D
       \quad
       \DD_2 :: \Gamma_1, A, \Delta_2, E
     }
    \]
    \begin{tabbing}
      $\GG ::\Gamma_1, \Gamma_2, \Delta_2, E$ \` by i.h.~\ref{cuta} on $\DD_2$ and $\EE$ \\
      $\Gamma_1, \Gamma_2, [\Delta_1]\Delta_2,D \tensor E,$ \` by $\tensor$ on $\DD_1$ and $\GG$ \\
    \end{tabbing}

 \item [Left-Commutative Case:]
   \[\DD  =
     \infer[\parr]
     {\Gamma_1, A, D \parr E}
     {\DD_1::\Gamma_1, A, [D] E}     
    \]
    \begin{tabbing}
      $\GG ::\Gamma_1, \Gamma_2, [D]E $ \` by i.h.~\ref{cuta} on $\DD_1$ and $\EE$ \\
      $\Gamma_1, \Gamma_2, D \parr E$ \` by $\parr$ on $\GG$
    \end{tabbing}

 \item [Left-Commutative Case:]
    \[\DD  =
      \infer[\with]
      {\Gamma_1,A,\oplus\Delta,D \with E}
      {\DD_1 :: \Gamma_1,A, \Lbr{\oplus\Delta}D
        \quad
        \DD_2 :: \Gamma_1,A, \Rbr{\oplus\Delta} E}
    \]
    \begin{tabbing}
      $\GG_1 ::\Gamma_1, \Gamma_2, \Lbr{\oplus\Delta}D $ \` by i.h.~\ref{cuta} on $\DD_1$ and $\EE$ \\
      $\GG_2 ::\Gamma_1, \Gamma_2,  \Rbr{\oplus\Delta} E$ \` by i.h.~\ref{cuta} on $\DD_2$ and $\EE$ \\
      $\Gamma_1, \Gamma_2, \oplus\Delta,D \with E$ \` by $\with$ on $\GG_1$ and $\GG_2$
    \end{tabbing}

%

 \item [Left-Commutative Case:]
   \[\DD  =
     \infer[\oplus_1]
     {\Gamma_1, A,  \Lbr{\Delta, D\oplus E} C}
    {\DD_1::\Gamma_1, A,  \Lbr{\Delta} C, D}    
  \]
    \begin{tabbing}
      $\GG ::\Gamma_1, \Gamma_2,  \Lbr{\Delta} C, D$ \` by i.h.~\ref{cuta} on $\DD_1$ and $\EE$ \\
      $\Gamma_1, \Gamma_2,\Lbr{\Delta, D\oplus E} C$ \` by $\oplus_1$ on $\GG$
    \end{tabbing}

 \item [Left-Commutative Case:]
   \[\DD  =  \infer[\oplus_2]
     {\Gamma_1, A, \Rbr{\Delta, D\oplus E}C}
     {\DD_1 :: \Gamma_1, A,\Rbr{\Delta} C, E}
   \]
   \begin{tabbing}
      $\GG ::\Gamma_1, \Gamma_2,  \Rbr{\Delta} C, E$ \` by i.h.~\ref{cuta} on $\DD_1$ and $\EE$ \\
      $\Gamma_1, \Gamma_2,\Rbr{\Delta, D\oplus E} C$ \` by $\oplus_2$ on $\GG$
   \end{tabbing}

  \item [Left-Commutative Case:]
    \[\DD  =     \infer[\query]
      {\Gamma_1,A, \Qbr{\Delta, \query D}C}
      {\DD_1 :: \Gamma_1,A, \Qbr{\Delta}C, D}
    \]
    \begin{tabbing}
      $\GG ::\Gamma_1, \Gamma_2,  \Qbr{\Delta}C, D$ \` by i.h.~\ref{cuta} on $\DD_1$ and $\EE$ \\
      $\Gamma_1, \Gamma_2,\Qbr{\Delta, \query D}C$ \` by $\query$ on $\GG$
    \end{tabbing}

  \end{description}
\bigskip
  \item By induction on $\FF$. 
    \begin{description}
    \item [Impossible Cases:] $\ax$, $\one$, $\bang$. 
    \item [Key Case:] $$\FF = \infer[\tensor]
      {\Gamma_3, C \tensor D, [A^\perp]B^\perp}
      {\FF_1 :: A^\perp, C \quad \FF_2 :: \Gamma_3 , B^\perp, D}$$
      \begin{tabbing}
        $\GG_1 :: \Delta_1, C$ \` by i.h.~\ref{cuta} on $\DD$ and $\FF_1$ \\
        $\GG_2 :: \Delta_2, \Gamma_2, \Gamma_3, D$ \` by i.h.~\ref{cuta} on $\EE$ and $\FF_2$ \\
        $[\Delta_1] \Delta_2, \Gamma_2, \Gamma_3, C \tensor D$ \` by  $\tensor$ on $\GG_1$ and $\GG_2$
      \end{tabbing}
    \item [Commutative Case:] $$\FF = \infer[\tensor]
      {\Gamma_3, [\Delta'_1]\Delta'_2, C \tensor D, [A^\perp]B^\perp}
      {\FF_1 :: \Delta'_1, C \quad \FF_2 :: \Gamma_3 , [A^\perp]B^\perp, \Delta'_2, D}$$
      \begin{tabbing}
        $\GG :: [\Delta_1] \Delta_2, \Gamma_2, \Gamma_3, \Delta_2', D$ \` by i.h.~\ref{cutb} on $\DD$, $\EE$, and $\FF_2$ \\
        $[\Delta'_1] \Delta'_2, C \tensor D,  [\Delta_1] \Delta_2, \Gamma_2, \Gamma_3$ \` by  $\tensor$ on $\FF_1$ and $\GG$
      \end{tabbing}
    \item [Commutative Case:]  $$\FF = \infer[\bot]
      {\Gamma_3,  [A^\perp]B^\perp, \bot}
      {\FF_1:: \Gamma_3,  [A^\perp]B^\perp, [\ast]}$$
      \begin{tabbing}
        $\GG :: [\Delta_1] \Delta_2, \Gamma_2, \Gamma_3, [\ast] $ \` by i.h.~\ref{cutb} on $\DD$, $\EE$, and $\FF_1$ \\
        $[\Delta_1] \Delta_2, \Gamma_2, \Gamma_3, \bot $ \` by $\bot$ on $\GG$
      \end{tabbing}
    \item [Commutative Case:] $$\FF =  \infer[\parr]
      {\Gamma_3, [A^\perp]B^\perp, C \parr D}
      {\FF_1 :: \Gamma_3, [A^\perp]B^\perp, [C] D}$$
      \begin{tabbing}
        $\GG :: [\Delta_1] \Delta_2, \Gamma_2, \Gamma_3, [C] D  $ \` by i.h.~\ref{cutb} on $\DD$, $\EE$, and $\FF_1$ \\
        $[\Delta_1] \Delta_2, \Gamma_2, \Gamma_3, C \parr D $ \` by $\parr$ on $\GG$
      \end{tabbing}
    \item [Commutative Case:] $$\FF =  \infer[\with]
      {\Gamma_3, [A^\perp]B^\perp,\oplus\Delta,C \with D}
      {\FF_1 :: \Gamma_3,  [A^\perp]B^\perp,\Lbr{\oplus\Delta}C
        \quad
        \FF_2 :: \Gamma_3,  [A^\perp]B^\perp,\Rbr{\oplus\Delta} D}$$
      \begin{tabbing}
        $\GG_1 :: [\Delta_1] \Delta_2, \Gamma_2, \Gamma_3, \Lbr{\oplus\Delta}C $ \` by i.h.~\ref{cutb} on $\DD$, $\EE$, and $\FF_1$ \\
        $\GG_2 :: [\Delta_1] \Delta_2, \Gamma_2, \Gamma_3,\Rbr{\oplus\Delta} D$ \` by i.h.~\ref{cutb} on $\DD$, $\EE$, and $\FF_2$ \\
        $[\Delta_1] \Delta_2, \Gamma_2, \Gamma_3, \oplus \Delta, C \with D$ \` by $\with$ on $\GG_1$ and $\GG_2$
      \end{tabbing}
%
    \item [Commutative Case:] $$\FF =  \infer[\oplus_1]
      {\Gamma_3,  [A^\perp]B^\perp, \Lbr{\Delta, D\oplus E} C}
      {\FF_1::\Gamma_3,  [A^\perp]B^\perp, \Lbr{\Delta} C, D}$$
      \begin{tabbing}
        $\GG :: [\Delta_1] \Delta_2, \Gamma_2, \Gamma_3, \Lbr{\Delta} C,D $ \` by i.h.~\ref{cutb} on $\DD$, $\EE$, and $\FF_1$ \\
        $ [\Delta_1] \Delta_2, \Gamma_2, \Gamma_3, \Lbr{\Delta, D \oplus E} C $ \` by $\oplus_1$ on  $\GG$ \\
      \end{tabbing}
      
    \item [Commutative Case:] $$\FF =  \infer[\oplus_2]
      {\Gamma_3,  [A^\perp]B^\perp, \Lbr{\Delta, D\oplus E} C}
      {\FF_1::\Gamma_3,  [A^\perp]B^\perp, \Lbr{\Delta} C, E}$$
      \begin{tabbing}
        $\GG :: [\Delta_1] \Delta_2, \Gamma_2, \Gamma_3, \Lbr{\Delta} C,E $ \` by i.h.~\ref{cutb} on $\DD$, $\EE$, and $\FF_1$ \\
        $ [\Delta_1] \Delta_2, \Gamma_2, \Gamma_3, \Lbr{\Delta, D \oplus E} C $ \` by $\oplus_2$ on  $\GG$ \\
      \end{tabbing}
      
    \item [Commutative Case:] $$\FF =   \infer[\query]
      {\Gamma_3, [A^\perp]B^\perp, \Qbr{\Delta, \query D}C}
      {\FF_1 :: \Gamma_3, [A^\perp]B^\perp, \Qbr{\Delta}C, D}$$
      \begin{tabbing}
        $\GG :: [\Delta_1] \Delta_2, \Gamma_2, \Gamma_3, \Qbr{\Delta}C, D $ \` by i.h.~\ref{cutb} on $\DD$, $\EE$, and $\FF_1$ \\
        $ [\Delta_1] \Delta_2, \Gamma_2, \Gamma_3, \Qbr{\Delta, \query D}C $ \` by $\query$ on  $\GG$ \\
      \end{tabbing}
      
      
    \end{description}
  \item By induction on $\DD$
    \begin{description}
    \item [Impossible Cases:]  $\ax$, $\one$, $\bang$
    \item [Key Case:]
      \[\DD  =
        \infer[\oplus_1]
        {\Gamma_1,  \Lbr{\Delta_1, A\oplus B} C}
        {\DD_1 :: \Gamma_1,  \Lbr{\Delta_1} C, A}\]
      \begin{tabbing}
        $\Gamma_1, \Gamma_2, \Lbr{\Delta_1, \Delta_2}C$ \` by i.h. \ref{cs11} on $\DD_1$ and $\EE$
      \end{tabbing}
    \item [Left-Commutative Case:]
      \[\DD  =
        \infer[\bot]
        {\Gamma_1, \Lbr{\Delta_1, A \oplus B} C, \bot}
        {\DD_1 :: \Gamma_1, \Lbr{\Delta_1, A \oplus B} C, [\ast]} \]
      \begin{tabbing}
        $\GG :: \Gamma_1, \Gamma_2,[\ast], \Lbr{\Delta_1,  \Delta_2} C $ \` by i.h.~\ref{cs1} on $\DD_1$, $\EE$, and $\FF$\\
        $\Gamma_1, \Gamma_2, \bot, \Lbr{\Delta_1,  \Delta_2} C $ \` by $\bot$ on $\GG$
      \end{tabbing}
    \item [Left-Commutative Case:]
      \[\DD  =  \infer[\tensor]
        {\Gamma_1, \Lbr{\Delta_1, A \oplus B} C, [\Delta_3]\Delta_4, D \tensor E}
        {\DD_1::\Delta_3, D \quad \DD_2::\Gamma_1, \Lbr{\Delta_1, A \oplus B} C, \Delta_4, E} \]
      \begin{tabbing}
        $\GG :: \Gamma_1, \Gamma_2, \Delta_4, E, \Lbr{\Delta_1,  \Delta_2} C $ \` by i.h.~\ref{cs1} on $\DD_2$, $\EE$, and $\FF$\\
        $ \Gamma_1, \Gamma_2, [\Delta_3]\Delta_4, D \tensor E, \Lbr{\Delta_1,  \Delta_2} C $ \` by $\tensor$ on $\DD_1$ and $\GG$
      \end{tabbing}
    \item [Left-Commutative Case:]
      \[\DD  =
       \infer[\parr]
       {\Gamma_1, \Lbr{\Delta_1, A \oplus B} C, D \parr E}
       {\DD_1::\Gamma_1, \Lbr{\Delta_1, A \oplus B} C, [D] E} \]
     \begin{tabbing}
       $\GG :: \Gamma_1, \Gamma_2, [D]E, \Lbr{\Delta_1,  \Delta_2} C $ \` by i.h.~\ref{cs1} on $\DD_1$, $\EE$, and $\FF$\\
       $ \Gamma_1, \Gamma_2, D \parr E, \Lbr{\Delta_1,  \Delta_2} C $ \` by $\parr$ on $\GG$ 
      \end{tabbing}

    \item [Left-Commutative Case:]
      \[\DD  =
        \infer[\with]
        {\Gamma_1,\Lbr{\Delta_1, A \oplus B} C,\oplus\Delta,D \with E}
        {\DD_1 :: \Gamma_1, \Lbr{\Delta_1, A \oplus B} C,\Lbr{\oplus\Delta}D
          \quad
          \DD_2 ::\Gamma_1, \Lbr{\Delta_1, A \oplus B} C,\Rbr{\oplus\Delta} E}
      \]
      \begin{tabbing}
      $\GG_1 :: \Gamma_1, \Gamma_2, \Lbr{\oplus\Delta}D , \Lbr{\Delta_1,  \Delta_2} C $ \` by i.h.~\ref{cs1} on $\DD_1$, $\EE$, and $\FF$\\
      $\GG_2 :: \Gamma_1, \Gamma_2, \Rbr{\oplus\Delta} E, \Lbr{\Delta_1,  \Delta_2} C $ \` by i.h.~\ref{cs1} on $\DD_2$, $\EE$, and $\FF$\\
       $ \Gamma_1, \Gamma_2, D \with E, \Lbr{\Delta_1,  \Delta_2} C $ \` by $\with$ on $\GG_1$ and $\GG_2$
     \end{tabbing}

%

    \item [Left-Commutative Case:]
      \[\DD  =
        \infer[\oplus_1]
        {\Gamma_1,  \Lbr{\Delta_1, A \oplus B} C, \Rbr{\Delta_3, D\oplus E}F}
        {\DD_1 :: \Gamma_1,  \Lbr{\Delta_1, A \oplus B} C, \Rbr{\Delta_3} F, D}
      \]
      \begin{tabbing}
        $\GG :: \Gamma_1, \Gamma_2, \Rbr{\Delta_3} F, D, \Lbr{\Delta_1,  \Delta_2} C $ \` by i.h.~\ref{cs1} on $\DD_1$, $\EE$, and $\FF$\\
        $\Gamma_1, \Gamma_2, \Rbr{\Delta_3, D\oplus E} F, \Lbr{\Delta_1,  \Delta_2} C $ \` by $\oplus_1$ on $\GG$      
      \end{tabbing}
    \item [Left-Commutative Case:]
      \[\DD  =
        \infer[\oplus_2]
        {\Gamma_1,  \Lbr{\Delta_1, A \oplus B} C, \Rbr{\Delta_3, D\oplus E}F}
        {\DD_1 :: \Gamma_1,  \Lbr{\Delta_1, A \oplus B} C, \Rbr{\Delta_3} F, E}
      \]
      \begin{tabbing}
        $\GG :: \Gamma_1, \Gamma_2, \Rbr{\Delta_3} F, E, \Lbr{\Delta_1,  \Delta_2} C $ \` by i.h.~\ref{cs1} on $\DD_1$, $\EE$, and $\FF$\\
        $\Gamma_1, \Gamma_2, \Rbr{\Delta_3, D\oplus E} F, \Lbr{\Delta_1,  \Delta_2} C $ \` by $\oplus_2$ on $\GG$      
      \end{tabbing}
    \item [Left-Commutative Case:]
      \[\DD  = \infer[\query]
        {\Gamma_1,  \Lbr{\Delta_1, A \oplus B} C, \Qbr{\Delta, \query D}E}
        {\DD_1:: \Gamma_1,  \Lbr{\Delta_1, A \oplus B} C, \Qbr{\Delta}E, D}
      \]
      \begin{tabbing}
        $\GG :: \Gamma_1, \Gamma_2, \Qbr{\Delta}E, D, \Lbr{\Delta_1,  \Delta_2} C $ \` by i.h.~\ref{cs1} on $\DD_1$, $\EE$, and $\FF$\\
        $\Gamma_1, \Gamma_2, \Qbr{\Delta, \query D}E, \Lbr{\Delta_1,  \Delta_2} C $ \` by $\query$ on $\GG$      
      \end{tabbing}


    \end{description}
  \item By induction on $\EE$ 
    \begin{description}
    \item [Impossible Cases:]  $\ax$, $\one$, $\bang$
    \item [Key Case:]
      \[\EE =
        \infer[\oplus_1]
        {\Gamma_2, \Lbr{B \oplus D}A^\perp}
        {\EE_1::\Gamma_2, B, A^\perp}
      \]
      \begin{tabbing}
        $\GG :: \Gamma_1, \Gamma_2, \Lbr{\Delta_1}C, B$ \` by i.h. \ref{cuta} on $\DD$ and $\EE_1$\\
        $\Gamma_1, \Gamma_2, \Lbr{\Delta_1,B \oplus D}C$ \` by $\oplus_1$ on $\GG$.
      \end{tabbing}
    \item [Right-Commutative Case:]
      \[\EE  =
        \infer[\bot]
        {\Gamma_2, \Lbr{\Delta_2}A^\perp, \bot}
        {\EE_1 :: \Gamma_2, \Lbr{\Delta_2}A^\perp, [\ast]} \]
      \begin{tabbing}
        $\GG :: \Gamma_1, \Gamma_2,[\ast], \Lbr{\Delta_1,  \Delta_2} C $ \` by i.h.~\ref{cs11} on $\DD$ and $\EE_1$\\
        $\Gamma_1, \Gamma_2, \bot, \Lbr{\Delta_1,  \Delta_2} C $ \` by $\bot$ on $\GG$
      \end{tabbing}
    \item [Right-Commutative Case:]
      \[\EE  =  \infer[\tensor]
        {\Gamma_2, \Lbr{\Delta_2}A^\perp, [\Delta_3]\Delta_4, D \tensor E}
        {\EE_1::\Delta_3, D \quad \EE_2::\Gamma_2, \Lbr{\Delta_2}A^\perp, \Delta_4, E} \]
      \begin{tabbing}
        $\GG :: \Gamma_1, \Gamma_2, \Delta_4, E, \Lbr{\Delta_1,  \Delta_2} C $ \` by i.h.~\ref{cs11} on $\DD$ and $\EE_2$\\
        $ \Gamma_1, \Gamma_2, [\Delta_3]\Delta_4, D \tensor E, \Lbr{\Delta_1,  \Delta_2} C $ \` by $\tensor$ on $\EE_1$ and $\GG$
      \end{tabbing}
    \item [Right-Commutative Case:]
      \[\EE  =
       \infer[\parr]
       {\Gamma_2, \Lbr{\Delta_2}A^\perp, D \parr E}
       {\EE_1::\Gamma_2, \Lbr{\Delta_2}A^\perp, [D] E} \]
     \begin{tabbing}
       $\GG :: \Gamma_1, \Gamma_2, [D]E, \Lbr{\Delta_1,  \Delta_2} C $ \` by i.h.~\ref{cs11} on $\DD$ and $\EE_1$\\
       $ \Gamma_1, \Gamma_2, D \parr E, \Lbr{\Delta_1,  \Delta_2} C $ \` by $\parr$ on $\GG$ 
      \end{tabbing}
    \item [Right-Commutative Case:]
      \[\EE  =
        \infer[\with]
        {\Gamma_2,\Lbr{\Delta_2}A^\perp,\oplus\Delta,D \with E}
        {\EE_1 :: \Gamma_2, \Lbr{\Delta_2}A^\perp,\Lbr{\oplus\Delta}D
          \quad
          \EE_2 ::\Gamma_2, \Lbr{\Delta_2}A^\perp,\Rbr{\oplus\Delta} E}
      \]
      \begin{tabbing}
      $\GG_1 :: \Gamma_1, \Gamma_2, \Lbr{\oplus\Delta}D , \Lbr{\Delta_1,  \Delta_2} C $ \` by i.h.~\ref{cs11} on $\DD$ and $\EE_1$\\
      $\GG_2 :: \Gamma_1, \Gamma_2, \Rbr{\oplus\Delta} E, \Lbr{\Delta_1,  \Delta_2} C $ \` by i.h.~\ref{cs11} on $\DD$ and $\EE_2$\\
       $ \Gamma_1, \Gamma_2, D \with E, \Lbr{\Delta_1,  \Delta_2} C $ \` by $\with$ on $\GG_1$ and $\GG_2$
     \end{tabbing}

%

    \item [Right-Commutative Case:]
      \[\EE  =
        \infer[\oplus_1]
        {\Gamma_2,  \Lbr{\Delta_2}A^\perp, \Rbr{\Delta_3, D\oplus E}F}
        {\EE_1 :: \Gamma_2,  \Lbr{\Delta_2}A^\perp, \Rbr{\Delta_3} F, D}
      \]
      \begin{tabbing}
        $\GG :: \Gamma_1, \Gamma_2, \Rbr{\Delta_3} F, D, \Lbr{\Delta_1,  \Delta_2} C $ \` by i.h.~\ref{cs11} on $\DD$ and $\EE_1$\\
        $\Gamma_1, \Gamma_2, \Rbr{\Delta_3, D\oplus E} F, \Lbr{\Delta_1,  \Delta_2} C $ \` by $\oplus_1$ on $\GG$      
      \end{tabbing}

    \item [Right-Commutative Case:]
      \[\EE  =
        \infer[\oplus_2]
        {\Gamma_2,  \Lbr{\Delta_2}A^\perp, \Rbr{\Delta_3, D\oplus E}F}
        {\EE_1 :: \Gamma_2,  \Lbr{\Delta_2}A^\perp, \Rbr{\Delta_3} F, E}
      \]
      \begin{tabbing}
        $\GG :: \Gamma_1, \Gamma_2, \Rbr{\Delta_3} F, E, \Lbr{\Delta_1,  \Delta_2} C $ \` by i.h.~\ref{cs11} on $\DD$ and $\EE_1$\\
        $\Gamma_1, \Gamma_2, \Rbr{\Delta_3, D\oplus E} F, \Lbr{\Delta_1,  \Delta_2} C $ \` by $\oplus_2$ on $\GG$      
      \end{tabbing}
  
    \item [Right-Commutative Case:]
      \[\EE  = \infer[\query]
        {\Gamma_2,  \Lbr{\Delta_2}A^\perp, \Qbr{\Delta, \query D}E}
        {\EE_1:: \Gamma_2,  \Lbr{\Delta_2}A^\perp, \Qbr{\Delta}E, D}
      \]
      \begin{tabbing}
        $\GG :: \Gamma_1, \Gamma_2, \Qbr{\Delta}E, D, \Lbr{\Delta_1,  \Delta_2} C $ \` by i.h.~\ref{cs11} on $\DD$ and $\EE_1$\\
        $\Gamma_1, \Gamma_2, \Qbr{\Delta, \query D}E, \Lbr{\Delta_1,  \Delta_2} C $ \` by $\query$ on $\GG$      
      \end{tabbing}
    \end{description}
  \item The proof is analogous to (\ref{cs1}),  simply by replacing $\Left$ by $\Right$. 
  \item The proof is analogous to (\ref{cs11}),  simply by replacing $\Left$ by $\Right$. 

  \item By induction on $\DD$
    \begin{description}
    \item [Impossible Cases:]  $\ax$, $\one$, $\bang$
    \item [Key Case:]
      \[\DD  =
        \infer[\query]
        {\Gamma_1, \Qbr{\Delta_1, \query A}C}
        {\DD_1::\Gamma_1, \Qbr{\Delta_1}C, A}
      \]
      \begin{tabbing}
        $\Gamma_1, \Qbr {\Delta_1, \Delta_2}C$ \` by i.h. \ref{cs33} on $\DD_1$ and $\EE$
      \end{tabbing}
    \item [Left-Commutative Case:]
      \[\DD  =
        \infer[\bot]
        {\Gamma_1,\Qbr{\Delta_1,\query A} C, \bot}
        {\DD_1 :: \Gamma_1,\Qbr{\Delta_1,\query A} C, [\ast]} \]
      \begin{tabbing}
        $\GG :: \Gamma_1, [\ast], \Qbr{\Delta_1,  \Delta_2} C $ \` by i.h.~\ref{cs3} on $\DD_1$ and $\EE$\\
        $\Gamma_1,  \bot, \Qbr{\Delta_1,  \Delta_2} C $ \` by $\bot$ on $\GG$
      \end{tabbing}
    \item [Left-Commutative Case:]
      \[\DD  =  \infer[\tensor]
        {\Gamma_1,\Qbr{\Delta_1,\query A} C, [\Delta_3]\Delta_4, D \tensor E}
        {\DD_1::\Delta_3, D \quad \DD_2::\Gamma_1,\Qbr{\Delta_1,\query A} C, \Delta_4, E} \]
      \begin{tabbing}
        $\GG :: \Gamma_1,  \Delta_4, E, \Qbr{\Delta_1,  \Delta_2} C $ \` by i.h.~\ref{cs3} on $\DD_2$ and $\EE$\\
        $ \Gamma_1,  [\Delta_3]\Delta_4, D \tensor E, \Qbr{\Delta_1,  \Delta_2} C $ \` by $\tensor$ on $\DD_1$ and $\GG$
      \end{tabbing}
    \item [Left-Commutative Case:]
      \[\DD  =
       \infer[\parr]
       {\Gamma_1,\Qbr{\Delta_1,\query A} C, D \parr E}
       {\DD_1::\Gamma_1,\Qbr{\Delta_1,\query A} C, [D] E} \]
     \begin{tabbing}
       $\GG :: \Gamma_1,  [D]E, \Qbr{\Delta_1,  \Delta_2} C $ \` by i.h.~\ref{cs3} on $\DD_1$ and $\EE$\\
       $ \Gamma_1,  D \parr E, \Qbr{\Delta_1,  \Delta_2} C $ \` by $\parr$ on $\GG$ 
      \end{tabbing}

    \item [Left-Commutative Case:]
      \[\DD  =
        \infer[\with]
        {\Gamma_1,\Qbr{\Delta_1,\query A} C,\oplus\Delta,D \with E}
        {\DD_1 :: \Gamma_1,\Qbr{\Delta_1,\query A} C,\Lbr{\oplus\Delta}D
          \quad
          \DD_2 ::\Gamma_1,\Qbr{\Delta_1,\query A} C,\Rbr{\oplus\Delta} E}
      \]
      \begin{tabbing}
      $\GG_1 :: \Gamma_1,  \Lbr{\oplus\Delta}D , \Qbr{\Delta_1,  \Delta_2} C $ \` by i.h.~\ref{cs3} on $\DD_1$ and $\EE$\\
      $\GG_2 :: \Gamma_1,  \Rbr{\oplus\Delta} E, \Qbr{\Delta_1,  \Delta_2} C $ \` by i.h.~\ref{cs3} on $\DD_2$ and $\EE$\\
       $ \Gamma_1,  D \with E, \Qbr{\Delta_1,  \Delta_2} C $ \` by $\with$ on $\GG_1$ and $\GG_2$
     \end{tabbing}



    \item [Left-Commutative Case:]
      \[\DD  =
        \infer[\oplus_1]
        {\Gamma_1, \Qbr{\Delta_1,\query A} C, \Rbr{\Delta_3, D\oplus E}F}
        {\DD_1 :: \Gamma_1, \Qbr{\Delta_1,\query A} C, \Rbr{\Delta_3} F, D}
      \]
      \begin{tabbing}
        $\GG :: \Gamma_1,  \Rbr{\Delta_3} F, D, \Qbr{\Delta_1,  \Delta_2} C $ \` by i.h.~\ref{cs3} on $\DD_1$ and $\EE$\\
        $\Gamma_1,  \Rbr{\Delta_3, D\oplus E} F, \Qbr{\Delta_1,  \Delta_2} C $ \` by $\oplus_1$ on $\GG$      
      \end{tabbing}
    \item [Left-Commutative Case:]
      \[\DD  =
        \infer[\oplus_2]
        {\Gamma_1, \Qbr{\Delta_1,\query A} C, \Rbr{\Delta_3, D\oplus E}F}
        {\DD_1 :: \Gamma_1, \Qbr{\Delta_1,\query A} C, \Rbr{\Delta_3} F, E}
      \]
      \begin{tabbing}
        $\GG :: \Gamma_1,  \Rbr{\Delta_3} F, E, \Qbr{\Delta_1,  \Delta_2} C $ \` by i.h.~\ref{cs3} on $\DD_1$ and $\EE$\\
        $\Gamma_1,  \Rbr{\Delta_3, D\oplus E} F, \Qbr{\Delta_1,  \Delta_2} C $ \` by $\oplus_2$ on $\GG$      
      \end{tabbing}
  
    \item [Left-Commutative Case:]
      \[\DD  = \infer[\query]
        {\Gamma_1, \Qbr{\Delta_1,\query A} C, \Qbr{\Delta, \query D}E}
        {\DD_1:: \Gamma_1, \Qbr{\Delta_1,\query A} C, \Qbr{\Delta}E, D}
      \]
      \begin{tabbing}
        $\GG :: \Gamma_1,  \Qbr{\Delta}E, D, \Qbr{\Delta_1,  \Delta_2} C $ \` by i.h.~\ref{cs3} on $\DD_1$ and $\EE$\\
        $\Gamma_1,  \Qbr{\Delta, \query D}E, \Qbr{\Delta_1,  \Delta_2} C $ \` by $\query$ on $\GG$      
      \end{tabbing}

    \end{description}

  \item By induction on $\EE$. 
    \begin{description}
    \item [Impossible Cases:]  $\ax$, $\one$, $\bang$
    \item [Key Case:]
      \[\EE =
        \infer[\query]
        {\Gamma_2, \Qbr{\query B}A^\perp}
        {\EE_1::\Gamma_2, B, A^\perp}
      \]
      \begin{tabbing}
        $\GG:: \Gamma_1, \Gamma_2,  \Qbr {\Delta_1}C, B$ \` by i.h. \ref{cuta} on $\DD$ and $\EE_1$\\
        $\Gamma_1, \Gamma_2,  \Qbr {\Delta_1, \query B}C$ \` by $\query$ on $\GG$
      \end{tabbing}
    \item [Right-Commutative Case:]
      \[\EE  =
        \infer[\bot]
        {\Gamma_2, \Qbr{\Delta_2}A^\perp, \bot}
        {\EE_1 :: \Gamma_2, \Qbr{\Delta_2}A^\perp, [\ast]} \]
      \begin{tabbing}
        $\GG :: \Gamma_1, \Gamma_2,[\ast], \Qbr{\Delta_1,  \Delta_2} C $ \` by i.h.~\ref{cs33} on $\DD$ and $\EE_1$\\
        $\Gamma_1, \Gamma_2, \bot, \Qbr{\Delta_1,  \Delta_2} C $ \` by $\bot$ on $\GG$
      \end{tabbing}
    \item [Right-Commutative Case:]
      \[\EE  =  \infer[\tensor]
        {\Gamma_2, \Qbr{\Delta_2}A^\perp, [\Delta_3]\Delta_4, D \tensor E}
        {\EE_1::\Delta_3, D \quad \EE_2::\Gamma_2, \Qbr{\Delta_2}A^\perp, \Delta_4, E} \]
      \begin{tabbing}
        $\GG :: \Gamma_1, \Gamma_2, \Delta_4, E, \Qbr{\Delta_1,  \Delta_2} C $ \` by i.h.~\ref{cs33} on $\DD$ and $\EE_2$\\
        $ \Gamma_1, \Gamma_2, [\Delta_3]\Delta_4, D \tensor E, \Qbr{\Delta_1,  \Delta_2} C $ \` by $\tensor$ on $\EE_1$ and $\GG$
      \end{tabbing}
    \item [Right-Commutative Case:]
      \[\EE  =
       \infer[\parr]
       {\Gamma_2, \Qbr{\Delta_2}A^\perp, D \parr E}
       {\EE_1::\Gamma_2, \Qbr{\Delta_2}A^\perp, [D] E} \]
     \begin{tabbing}
       $\GG :: \Gamma_1, \Gamma_2, [D]E, \Qbr{\Delta_1,  \Delta_2} C $ \` by i.h.~\ref{cs33} on $\DD$ and $\EE_1$\\
       $ \Gamma_1, \Gamma_2, D \parr E, \Qbr{\Delta_1,  \Delta_2} C $ \` by $\parr$ on $\GG$ 
      \end{tabbing}
    \item [Right-Commutative Case:]
      \[\EE  =
        \infer[\with]
        {\Gamma_2,\Qbr{\Delta_2}A^\perp,\oplus\Delta,D \with E}
        {\EE_1 :: \Gamma_2, \Qbr{\Delta_2}A^\perp,\Lbr{\oplus\Delta}D
          \quad
          \EE_2 ::\Gamma_2, \Qbr{\Delta_2}A^\perp,\Rbr{\oplus\Delta} E}
      \]
      \begin{tabbing}
      $\GG_1 :: \Gamma_1, \Gamma_2, \Lbr{\oplus\Delta}D , \Qbr{\Delta_1,  \Delta_2} C $ \` by i.h.~\ref{cs33} on $\DD$ and $\EE_1$\\
      $\GG_2 :: \Gamma_1, \Gamma_2, \Rbr{\oplus\Delta} E, \Qbr{\Delta_1,  \Delta_2} C $ \` by i.h.~\ref{cs33} on $\DD$ and $\EE_2$\\
       $ \Gamma_1, \Gamma_2, D \with E, \Qbr{\Delta_1,  \Delta_2} C $ \` by $\with$ on $\GG_1$ and $\GG_2$
     \end{tabbing}

%

    \item [Right-Commutative Case:]
      \[\EE  =
        \infer[\oplus_1]
        {\Gamma_2,  \Qbr{\Delta_2}A^\perp, \Rbr{\Delta_3, D\oplus E}F}
        {\EE_1 :: \Gamma_2,  \Qbr{\Delta_2}A^\perp, \Rbr{\Delta_3} F, D}
      \]
      \begin{tabbing}
        $\GG :: \Gamma_1, \Gamma_2, \Rbr{\Delta_3} F, D, \Qbr{\Delta_1,  \Delta_2} C $ \` by i.h.~\ref{cs33} on $\DD$ and $\EE_1$\\
        $\Gamma_1, \Gamma_2, \Rbr{\Delta_3, D\oplus E} F, \Qbr{\Delta_1,  \Delta_2} C $ \` by $\oplus_1$ on $\GG$      
      \end{tabbing}

    \item [Right-Commutative Case:]
      \[\EE  =
        \infer[\oplus_2]
        {\Gamma_2,  \Qbr{\Delta_2}A^\perp, \Rbr{\Delta_3, D\oplus E}F}
        {\EE_1 :: \Gamma_2,  \Qbr{\Delta_2}A^\perp, \Rbr{\Delta_3} F, E}
      \]
      \begin{tabbing}
        $\GG :: \Gamma_1, \Gamma_2, \Rbr{\Delta_3} F, E, \Qbr{\Delta_1,  \Delta_2} C $ \` by i.h.~\ref{cs33} on $\DD$ and $\EE_1$\\
        $\Gamma_1, \Gamma_2, \Rbr{\Delta_3, D\oplus E} F, \Qbr{\Delta_1,  \Delta_2} C $ \` by $\oplus_2$ on $\GG$      
      \end{tabbing}
  
    \item [Right-Commutative Case:]
      \[\EE  = \infer[\query]
        {\Gamma_2,  \Qbr{\Delta_2}A^\perp, \Qbr{\Delta, \query D}E}
        {\EE_1:: \Gamma_2,  \Qbr{\Delta_2}A^\perp, \Qbr{\Delta}E, D}
      \]
      \begin{tabbing}
        $\GG :: \Gamma_1, \Gamma_2, \Qbr{\Delta}E, D, \Qbr{\Delta_1,  \Delta_2} C $ \` by i.h.~\ref{cs33} on $\DD$ and $\EE_1$\\
        $\Gamma_1, \Gamma_2, \Qbr{\Delta, \query D}E, \Qbr{\Delta_1,  \Delta_2} C $ \` by $\query$ on $\GG$      
      \end{tabbing}
  
    \item [Right-Commutative Case:]
      \[\EE  = \infer[\query]
        {\Gamma_2,  \Qbr{\Delta_2,\query D}A^\perp}
        {\EE_1:: \Gamma_2,  \Qbr{\Delta_2}A^\perp, D}
      \]
      \begin{tabbing}
        $\GG :: \Gamma_1, \Gamma_2, D, \Qbr{\Delta_1,  \Delta_2} C $ \` by i.h.~\ref{cs33} on $\DD$ and $\EE_1$\\
        $\Gamma_1, \Gamma_2, \Qbr{\Delta_1,  \Delta_2, \query D} C $ \` by $\query$ on $\GG$      
      \end{tabbing}
    \end{description}
    
  \end{enumerate}
\end{proof}

\renewcommand{\Lbr}{\Select[l]}
\renewcommand{\Rbr}{\Select[r]}
\renewcommand{\Qbr}{{\color{blue}{\mathcal Q}}}

\newpage
\section{Proof of
  Lemma~\ref{lem:opluspermute}}\label{app:opluspermute}
{\bf Lemma~\ref{lem:opluspermute}.}
Let
$\syncseq P{\Gamma, \ltriCtx{\Delta, \phl{x}:\thl{A\oplus
      B}}{z}{C}}$. Then, there exists $\phl{P'}$ such that
\begin{displaymath}
  \infer[\oplus_l]
  {\syncseq {\inl x{P'}} {\Gamma,  \ltriCtx{\Delta, \phl{x}:\thl{A\oplus B}}{z}{C} }}
  {\syncseq {P'} {\Gamma,  \ltriCtx{{\Delta}}{z}{C}, \pairQ{x}{A}}}
\end{displaymath}
(similar result for the right case, with $\oplus_r$).

\bigskip
\noindent Let
$\syncseq P{\Gamma, \triCtx{\Delta, \pairQ x{?A}}{z}{C}}$. Then, there
exists $\phl {P'}$ such that
\begin{smallequation*}
  \infer[?]{
    \syncseq {\client xy{P'} }{\Gamma, \Query[\![\Delta, \phl x:\thl{\query A}]\!]\phl z:\thl{C}} 
  }{
    \syncseq {P'}{\Gamma, \Query[\![\Delta]\!]\phl z:\thl{C}, \phl y:\thl A}
  }
\end{smallequation*}

\begin{proof}
  We look at the case of $\oplus_l$ in detail. The other cases are
  very similar. We show that, besides the rule introducing
  $A\oplus B$, the proof does not change its structure at all. This is
  done by induction on the size of the proof by showing that
  $\oplus_l$ can permute down with any of the other rules.
  \begin{itemize}

  \item $\one$/$Ax$. Not applicable.
    
  \item $\perp$. If the last applied rule is $\perp$, then it must be
    such that:
    \[
      \infer[\perp]
      {
        \syncseq{\wait wP} {\Gamma, \ltriCtx{\Delta,
            \phl{x}:\thl{A\oplus B}}{z}{C},  \pairQ
          w{\bot}}
      }
      {
        \syncseq{P}{\Gamma, \ltriCtx{\Delta,
            \phl{x}:\thl{A\oplus B}}{z}{C}, [\ast]}
      }
    \]
    By induction hypothesis, there exists $\phl{P'}$ such that
    \begin{smallequation*}
      \infer[\perp]
      {
        \syncseq{\wait w{\inl x{P'}}} {\Gamma, \ltriCtx{\Delta,
            \phl{x}:\thl{A\oplus B}}{z}{C}, \pairQ
          w{\bot}}
      }
      {
        \infer[\oplus_l]
        {\syncseq {\inl x{P'}} {\Gamma,  \ltriCtx{\Delta,
              \phl{x}:\thl{A\oplus B}}{z}{C}, [\ast] }}
        {
        \syncseq{P'}{\Gamma, \ltriCtx{\Delta}{z}{C},
          \phl{x}:\thl{A}, [\ast]}
        }
      }
    \end{smallequation*}
    Clearly, we can make the two rules commute, obtaining: 
    \begin{displaymath}
      \infer[\oplus_l]
      {
        \syncseq{\inl x{\wait w{P'}}}
          {\Gamma, \ltriCtx{\Delta, \pairQ x{A\oplus B}}{z}{C},   \pairQ
            w{\perp}}
      }
      {
        \infer[\perp]
        {
          \syncseq{\wait w{P'}}
          {\Gamma, \ltriCtx{\Delta}{z}{C}, \pairQ xA,  \pairQ
            w{\bot}}
        }
        {
          {\syncseq {P'} {\Gamma,  \ltriCtx{{\Delta}}{z}{C}, \pairQ{x}{A},
               [\ast]}}
        }
      }
    \end{displaymath}

  \item $\parr$. If the last applied rule is $\parr$, then it must be
    such that:
    \[
      \infer[\parr]
      {
        \syncseq{\recv wyP} {\Gamma, \ltriCtx{\Delta,
            \phl{x}:\thl{A\oplus B}}{z}{C},  \pairQ
          w{D \parr E}}
      }
      {
        \syncseq{P}{\Gamma, \ltriCtx{\Delta,
            \phl{x}:\thl{A\oplus B}}{z}{C}, \quadQ yDwE}
      }
    \]
    By induction hypothesis, there exists $\phl{P'}$ such that
    \begin{displaymath}
      \infer[\parr]
      {
        \syncseq{\recv wy{\inl x{P'}}} {\Gamma, \ltriCtx{\Delta,
            \phl{x}:\thl{A\oplus B}}{z}{C}, \pairQ
          w{D \parr E}}
      }
      {
        \infer[\oplus_l]
        {\syncseq {\inl x{P'}} {\Gamma,  \ltriCtx{\Delta,
              \phl{x}:\thl{A\oplus B}}{z}{C}, \quadQ yDwE }}
        {
        \syncseq{P'}{\Gamma, \ltriCtx{\Delta}{z}{C},
          \phl{x}:\thl{A}, \quadQ yDwE}
        }
      }
    \end{displaymath}
    Clearly, we can make the two rules commute, obtaining: 
    \begin{displaymath}
      \infer[\oplus_l]
      {
        \syncseq{\inl x{\recv wy{P'}}}
          {\Gamma, \ltriCtx{\Delta, \pairQ x{A\oplus B}}{z}{C},   \pairQ
            w{D \parr E}}
      }
      {
        \infer[\parr]
        {
          \syncseq{\recv wy{P'}}
          {\Gamma, \ltriCtx{\Delta}{z}{C}, \pairQ xA,  \pairQ
            w{D \parr E}}
        }
        {
          {\syncseq {P'} {\Gamma,  \ltriCtx{{\Delta}}{z}{C}, \pairQ{x}{A},
               \quadQ yDwE}}
        }
      }
    \end{displaymath}

  \item $\tensor$. If the last applied rule is $\tensor$, then it must
    be such that:
    \[ \infer[\tensor]{ \syncseq{\send wyPQ} {\Gamma, \ltriCtx{\Delta,
            \phl{x}:\thl{A\oplus B}}{z}{C}, [\Delta_1]\Delta_2, \pairQ
          w{D \tensor E}} } { \syncseq{P}{\Delta_1, \pairQ yD} &
        \syncseq{Q}{\Gamma, \ltriCtx{\Delta,
            \phl{x}:\thl{A\oplus B}}{z}{C}, \Delta_2, \pairQ wE} }
    \]
    By induction hypothesis, there exists $\phl{Q'}$ such that
    \begin{displaymath}
      \infer[\tensor]{ \syncseq{\send wyP{\inl x{Q'}}} {\Gamma, \ltriCtx{\Delta,
            \phl{x}:\thl{A\oplus B}}{z}{C}, [\Delta_1]\Delta_2, \pairQ
          w{D \tensor E}} } { \syncseq{P}{\Delta_1, \pairQ yD} &
        \infer[\oplus_l]
        {\syncseq {\inl x{Q'}} {\Gamma,  \ltriCtx{\Delta,
              \phl{x}:\thl{A\oplus B}}{z}{C}, \Delta_2, \pairQ wE }}
        {\syncseq {Q'} {\Gamma,  \ltriCtx{{\Delta}}{z}{C}, \pairQ{x}{A},
            \Delta_2, \pairQ wE}}
      }
    \end{displaymath}
    Clearly, we can make the two rules commute, obtaining: 
    \begin{displaymath}
      \infer[\oplus_l]
      {
        \syncseq{\inl x{\send wyP{Q'}}}
          {\Gamma, \ltriCtx{\Delta, \pairQ x{A\oplus B}}{z}{C},  [\Delta_1]\Delta_2, \pairQ
            w{D \tensor E}}
      }
      {
        \infer[\tensor]
        {
          \syncseq{\send wyP{Q'}}
          {\Gamma, \ltriCtx{\Delta}{z}{C}, \pairQ xA, [\Delta_1]\Delta_2, \pairQ
            w{D \tensor E}}
        }
        {
          \syncseq{P}{\Delta_1, \pairQ yD} &
          {\syncseq {Q'} {\Gamma,  \ltriCtx{{\Delta}}{z}{C}, \pairQ{x}{A},
              \Delta_2, \pairQ wE}}
        }
      }
    \end{displaymath}

  \item $\oplus_l$/$\oplus_r$. Here we can have three subcases. If the
    last applied rule is $\oplus_l$ and it is indeed working on
    endpoint $x$ then we are done (this is the base case). Otherwise,
    rule $\oplus_l$ may be working either on a formulas inside the box
    $\ltriCtx{{\Delta}}{z}{C}$ (in $\Delta$) or in some other box. In
    both cases, we proceed as usual. Below, we look at the case where
    the formula is in another box.
    \[
      \infer[\oplus_l]
      {
        \syncseq{\inl wP} {\Gamma, \ltriCtx{\Delta,
            \phl{x}:\thl{A\oplus B}}{z}{C},
          \ltriCtx{\Delta', \pairQ w{D\oplus E}}yF
        }
      }
      {
        \syncseq{P}{\Gamma, \ltriCtx{\Delta,
            \phl{x}:\thl{A\oplus B}}{z}{C},
          \ltriCtx{\Delta'}yF, \pairQ w{D}
        }
      }
    \]
    By induction hypothesis, there exists $\phl{P'}$ such that
    \begin{displaymath}
      \infer[\oplus_l]
      {
        \syncseq{\inl w{\inl x{P'}}} {\Gamma, \ltriCtx{\Delta,
            \phl{x}:\thl{A\oplus B}}{z}{C},
          \ltriCtx{\Delta', \pairQ w{D\oplus E}}yF
        }
      }
      {
        \infer[\oplus_l]
        {
          \syncseq{\inl x{P'}}{\Gamma, \ltriCtx{\Delta,
              \phl{x}:\thl{A\oplus B}}{z}{C},
            \ltriCtx{\Delta'}yF, \pairQ w{D}
          }
        }
        {
          \syncseq{P'}{\Gamma, \ltriCtx{\Delta}{z}{C},
            \pairQ xA,
            \ltriCtx{\Delta'}yF, \pairQ w{D}
          }
        }
      }
    \end{displaymath}
    Clearly, we can make the two rules commute, obtaining: 
    \begin{displaymath}
      \infer[\oplus_l]
      {
        \syncseq{\inl w{\inl x{P'}}} {\Gamma, \ltriCtx{\Delta,
            \phl{x}:\thl{A\oplus B}}{z}{C},
          \ltriCtx{\Delta', \pairQ w{D\oplus E}}yF
        }
      }
      {
        \infer[\oplus_l]
        {
          \syncseq{\inl x{P'}}{\Gamma, \ltriCtx{\Delta}{z}{C},
            \pairQ xA,
            \ltriCtx{\Delta',  \pairQ w{D\oplus E}}yF
          }
        }
        {
          \syncseq{P'}{\Gamma, \ltriCtx{\Delta}{z}{C},
            \pairQ xA,
            \ltriCtx{\Delta'}yF, \pairQ w{D}
          }
        }
      }
    \end{displaymath}

  \item $\with$. In this case, we need to apply the induction
    hypothesis to both branches of the rule $\with$. If that is the
    last applied rule, then it must have the following format:
    \begin{displaymath}
      \infer[\with]
      {\syncseq{\Case wPQ} {\Gamma,\ltriCtx{\Delta, \pairQ x{A\oplus
              B}}{z}{C}, \oplus\Delta,\pairQ w{C \with D}}
      }{
    	\syncseq P {\Gamma, \ltriCtx{\Delta, \pairQ x{A\oplus
              B}}{z}{C}, \ltriCtx {\oplus\Delta}w{C}}
    	&
    	\syncseq Q {\Gamma, \ltriCtx{\Delta, \pairQ x{A\oplus B}}{z}{C}, \rtriCtx {\oplus\Delta}w{D}}
      }
    \end{displaymath}
    By induction hypothesis, there exist $\phl{P'}$ and $\phl{Q'}$
    such that: 
    \begin{displaymath}
      \infer[\with]
      {\syncseq{\Case w{\inl x{P'}}{\inl x{Q'}}} {\Gamma,\ltriCtx{\Delta, \pairQ x{A\oplus
              B}}{z}{C}, \oplus\Delta,\pairQ w{C \with D}}
      }
      {
        \infer[\oplus_l]
        {
          \syncseq {\inl x{P'}} {\Gamma, \ltriCtx{\Delta, \pairQ x{A\oplus
                B}}{z}{C}, \ltriCtx {\oplus\Delta}w{C}}
        }
        {
          \syncseq {P'} {\Gamma, \ltriCtx{\Delta}{z}{C}, \pairQ xA,
            \ltriCtx {\oplus\Delta}w{C}}
        }
    	&
    	\infer[\oplus_l]
        {
          \syncseq {\inl x{Q'}} {\Gamma, \ltriCtx{\Delta, \pairQ x{A\oplus
                B}}{z}{C}, \rtriCtx {\oplus\Delta}w{D}}
        }
        {
          \syncseq {Q'} {\Gamma, \ltriCtx{\Delta}{z}{C}, \pairQ xA,
            \rtriCtx {\oplus\Delta}w{D}}
        }
      }
    \end{displaymath}
    As in the previous cases, we can now commute the two rules:
    \begin{displaymath}
      \infer[\oplus_l]
      {
          \syncseq{\inl x{\Case w{P'}{Q'}}} {\Gamma,\ltriCtx{\Delta, \pairQ x{A\oplus
                B}}{z}{C}, \oplus\Delta,\pairQ w{C \with D}}
      }
      {
        \infer[\with]
        {
          \syncseq{\Case w{P'}{Q'}} {\Gamma,\ltriCtx{\Delta}{z}{C},
            \pair xA, \oplus\Delta,\pairQ w{C \with D}}
        }
        {
          \syncseq {P'} {\Gamma, \ltriCtx{\Delta}{z}{C}, \pairQ xA,
            \ltriCtx {\oplus\Delta}w{C}}
          &
          \syncseq {Q'} {\Gamma, \ltriCtx{\Delta}{z}{C}, \pairQ xA,
            \rtriCtx {\oplus\Delta}w{D}}
        }
      }
    \end{displaymath}

  \item $\query$. In the case of $\query$, it must be the case that:
    \begin{displaymath}
      \infer[\query]
      {
        \syncseq
        {\client xyP }
        {\Gamma, \ltriCtx{\Delta,
            \phl{x}:\thl{A\oplus B}}{z}{C}, \Query[\![\Delta, \phl
          x:\thl{\query A}]\!]\phl z:\thl{C}}
      }
      {
        \syncseq{P}
        {\Gamma, \ltriCtx{\Delta,
            \phl{x}:\thl{A\oplus B}}{z}{C}, \Query[\![\Delta]\!]\phl
          z:\thl{C}, \pairQ xA
        }
      }
    \end{displaymath}
    By induction hypothesis, there exists $\phl{P'}$ such that
    \begin{displaymath}
      \infer[\query]
      {
        \syncseq
        {\client wy{\inl x{P'}} }
        {\Gamma, \ltriCtx{\Delta,
            \phl{x}:\thl{A\oplus B}}{z}{C}, \Query[\![\Delta, \phl
          w:\thl{\query A}]\!]\phl t:\thl{C}}
      }
      {
        \infer[\oplus_l]
        {
          \syncseq{\inl x{P'}}
          {\Gamma, \ltriCtx{\Delta,
              \phl{x}:\thl{A\oplus B}}{z}{C}, \Query[\![\Delta]\!]\phl
            t:\thl{C}, \pairQ yA
          }
        }
        {
          \syncseq{P'}
          {\Gamma, \ltriCtx{\Delta}zC, \pair xA, 
            \Query[\![\Delta]\!]\phl
            t:\thl{C}, \pairQ yA
          }
        }
      }
    \end{displaymath}
    Finally, we can swap the two rules and obtain: 
    \begin{displaymath}
      \infer[\oplus_l]
      {
        \syncseq
        {\inl x{\client wy{P'}}}
        {\Gamma, \ltriCtx{\Delta,
            \phl{x}:\thl{A\oplus B}}{z}{C},
          \Query[\![\Delta, \phl
          w:\thl{\query A}]\!]\phl t:\thl{C}}
      }
      {
        \infer[\query]
        {
          \syncseq{\client wy{P'}}
          {\Gamma, \ltriCtx{\Delta}zC, \pair xA,
            \Query[\![\Delta, \phl
          w:\thl{\query A}]\!]\phl t:\thl{C}
          }
        }
        {
          \syncseq{P'}
          {\Gamma, \ltriCtx{\Delta}zC, \pair xA, 
            \Query[\![\Delta]\!]\phl
            t:\thl{C}, \pairQ yA
          }
        }
      }
    \end{displaymath}

  \item $\bang$. Not applicable.

  \end{itemize}
\end{proof}
  
\newpage
\section{Proof of
  Lemma~\ref{lem:parradmissibility} (part 1)}\label{app:parradmissibility}
{\bf Lemma~\ref{lem:parradmissibility} (part 1).}  Let
$\mathcal D::\syncseqtris P{\Gamma, \quadQ yAxB}$ such
that $\mathcal D$ is $\parr$-free. Then, there exists a $\parr$-free
proof $\mathcal E$ and $\phl Q$ such that
$\mathcal E::\syncseqtris {Q}{\Gamma, \pairQ x{A\parr
    B}}$.

\begin{proof}
  We proceed by induction on the size of the proof, and look at the
  last applied rule. Below, we only report the key case. The full
  proof is in Appendix~\ref{app:parradmissibility}

  \begin{itemize}
  \item $\one/\Did{Ax}/\parr/\bang$. Not applicable.
  \item $\perp$. In this case, we have:
    \begin{displaymath}
      \infer[\bot]
      {\syncseq {\wait z{P'}} {\Gamma, \quadQ yAxB, \pairQ{z}{\bot}}}
      {\syncseq {P'} {\Gamma, \quadQ yAxB, [\ast]}}
    \end{displaymath}
    By induction hypothesis, since $\mathcal D$ in
    $\mathcal D:: \phl {P'}$ is $\parr$-free, by
    applying $\bot$ again, we obtain:
    \begin{displaymath}
      \infer[\bot]
      {\syncseq {\wait z{Q'}} {\Gamma, \pairQ x{A\parr B}, \pairQ{z}{\bot}}}
      {\syncseq {Q'} {\Gamma, \pairQ x{A\parr B}, [\ast]}}
    \end{displaymath}
    
  \item $\tensor\parr^\perp_p$. If the last applied rule is
    $\tensor\parr^\perp_p$, we have two cases:
    \begin{itemize}
    \item $\quadQ yAxB$ is not principle for $\tensor\parr^\perp_p$:
      \begin{displaymath}
        \infer[\tensor\parr^\perp_p]
        {
          \syncseqtris \cdot
          {
            \Gamma, \quadQ yAxB, [\Delta_1]\Delta_2, \{\phl{x_i} : \thl{A_i \parr B_i}\}_i,
            \pairQ y{C
              \tensor D}
          }
        }
        {
          \syncseqtris {\cdot} {\Delta_1, \{\pairQ{x_i}{A_i}\}_i,  \pairQ yC}
          & \syncseqtris {\cdot} {\thl {\Gamma}, \quadQ yAxB, \Delta_2, \{\phl{x_i}:\thl{B_i}\}_i,
            \phl
            y:\thl D}
        }
      \end{displaymath}
      By induction hypothesis, since our proof is $\parr$-free, by
      applying $\tensor\parr^\perp_p$ again, we obtain:
      \begin{displaymath}
        \infer[\tensor\parr^\perp_p]
        {
          \syncseqtris \cdot
          {
            \Gamma, \pairQ x{A\parr B}, [\Delta_1]\Delta_2, \{\phl{x_i} : \thl{A_i \parr B_i}\}_i,
            \pairQ y{C
              \tensor D}
          }
        }
        {
          \syncseqtris {\cdot} {\Delta_1, \{\pairQ{x_i}{A_i}\}_i,  \pairQ yC}
          & \syncseqtris {\cdot} {\thl {\Gamma}, \pairQ x{A\parr B}, \Delta_2, \{\phl{x_i}:\thl{B_i}\}_i,
            \phl
            y:\thl D}
        }
      \end{displaymath}

    \item If $\quadQ yAxB$ is indeed modified by by
      $\tensor\parr^\perp_p$, then we have a base case:
      \begin{displaymath}
        \infer[\tensor\parr^\perp_p]
        {
          \syncseqbis \cdot
          {
            \Gamma, \quadQ yAxB, [\Delta_1]\Delta_2, \{\phl{x_i} : \thl{A_i \parr B_i}\}_i,
            \pairQ y{C
              \tensor D}
          }
        }
        {
          \syncseqbis {\cdot} {\Delta_1, A, \{\pairQ{x_i}{A_i}\}_i,  \pairQ yC}
          & \syncseqbis {\cdot} {\thl {\Gamma}, B, \Delta_2, \{\phl{x_i}:\thl{B_i}\}_i,
            \phl
            y:\thl D}
        }
      \end{displaymath}
      Clearly, we can change the rule $\tensor\parr^\perp_p$
      obtaining: 
      \begin{displaymath}
        \infer[\tensor\parr^\perp_p]
        {
          \syncseqbis \cdot
          {
            \Gamma, \pairQ x{A\parr B}, [\Delta_1]\Delta_2, \{\phl{x_i} : \thl{A_i \parr B_i}\}_i,
            \pairQ y{C
              \tensor D}
          }
        }
        {
          \syncseqbis {\cdot} {\Delta_1, A, \{\pairQ{x_i}{A_i}\}_i,  \pairQ yC}
          & \syncseqbis {\cdot} {\thl {\Gamma}, B, \Delta_2, \{\phl{x_i}:\thl{B_i}\}_i,
            \phl
            y:\thl D}
        }
      \end{displaymath}
    \end{itemize}
    
  \item $\with$. In this case, we have:
    \begin{displaymath}
      \infer[\with]
      {\syncseq{\Case x{P'}{Q'}} {\Gamma, \quadQ yAxB, \oplus\Delta,\pairQ x{A \with B}}
      }{
    	\syncseq {P'} {\Gamma, \quadQ yAxB, \ltriCtx {\oplus\Delta}x{A}}
    	&
    	\syncseq {Q'} {\Gamma, \quadQ yAxB, \rtriCtx {\oplus\Delta}x{B}}
      }
    \end{displaymath}
    By induction hypothesis, since both proofs in the premise of
    $\with$ are $\parr$-free, by applying $\with$ again, we obtain:
    \begin{displaymath}
      \infer[\with]
      {\syncseq{\Case x{P'}{Q'}} {\Gamma, \pairQ x{A\parr B}, \oplus\Delta,\pairQ x{A \with B}}
      }{
    	\syncseq {P'} {\Gamma, \pairQ x{A\parr B}, \ltriCtx {\oplus\Delta}x{A}}
    	&
    	\syncseq {Q'} {\Gamma, \pairQ x{A\parr B}, \rtriCtx {\oplus\Delta}x{B}}
      }
    \end{displaymath}

  \item $\oplus_l$ (similar for $\oplus_r$). In this case, we have:
    \begin{displaymath}
      \infer[\oplus_l]
      {\syncseq {\inl w{P'}} {\Gamma,  \quadQ yAxB, \ltriCtx{\Delta, \phl{w}:\thl{A\oplus B}}{z}{C} }}
      {\syncseq {P'} {\Gamma,  \quadQ yAxB, \ltriCtx{{\Delta}}{z}{C}, \pairQ{w}{A}}}
    \end{displaymath}
    By induction hypothesis, since $\mathcal D$ in
    $\mathcal D:: \phl {P'}$ is $\parr$-free, by
    applying $\oplus_l$ again, we obtain:
    \begin{displaymath}
      \infer[\oplus_l]
      {\syncseq {\inl w{Q'}} {\Gamma,  \pairQ x{A\parr B}, \ltriCtx{\Delta, \phl{w}:\thl{A\oplus B}}{z}{C} }}
      {\syncseq {Q'} {\Gamma,  \pairQ x{A\parr B}, \ltriCtx{{\Delta}}{z}{C}, \pairQ{w}{A}}}
    \end{displaymath}

  \item $\query$. In this case, we have:
    \begin{displaymath}
      \infer[?]{
    	\syncseq {\client xyP }{\Gamma, \quadQ yAxB, \Query[\![\Delta, \phl x:\thl{\query A}]\!]\phl z:\thl{C}} 
      }{
    	\syncseq {P}{\Gamma, \quadQ yAxB, \Query[\![\Delta]\!]\phl z:\thl{C}, \phl y:\thl A}
      }
    \end{displaymath}
    By induction hypothesis, since $\mathcal D$ in
    $\mathcal D:: \phl {P'}$ is $\parr$-free, by
    applying $\query$ again, we obtain:
    \begin{displaymath}
      \infer[?]{
    	\syncseq {\client xyP }{\Gamma, \pairQ x{A\parr B}, \Query[\![\Delta, \phl x:\thl{\query A}]\!]\phl z:\thl{C}} 
      }{
    	\syncseq {P}{\Gamma, \pairQ x{A\parr B}, \Query[\![\Delta]\!]\phl z:\thl{C}, \phl y:\thl A}
      }
    \end{displaymath}

  \end{itemize}
\end{proof}

\newpage
\section{Proof of
  Lemma~\ref{lem:perpadmissibility} (part 2)}\label{app:perpadmissibility}
{\bf Lemma~\ref{lem:perpadmissibility} (part2).}  Let
$\mathcal D::\syncseqtris P{\Gamma, \Star}$ such that
$\mathcal D$ is $\perp$-free. Then, there exists a $\perp$-free proof
$\mathcal E$ and $\phl Q$ such that
$\mathcal E::\syncseqtris {Q}{\Gamma, \pairQ x\perp}$.

\begin{proof}
  Similar to that of previous Lemma. 

  \begin{itemize}
  \item $\Did{Ax}/\perp/\bang$. Not applicable.

  \item $\one\bot^\perp$. This is the base case. If $\Gamma$ contains
    $\Star$ then we do nothing. Otherwise,
    \begin{displaymath}
      \infer[\one\bot^\perp]
      {
        \syncseqbis \cdot
        {
          \{\phl {x_i}:\thl{\perp}\}_i, 
          \phl{y}:\thl{\one}, \Star
        }
      }
      {
      }
    \end{displaymath}
    which can be replaced by
    \begin{displaymath}
      \infer[\one\bot^\perp]
      {
        \syncseqbis \cdot
        {
          \{\phl {x_i}:\thl{\perp}\}_i,  \phl{y}:\thl{\one},  \pairQ x\perp
        }
      }
      {
      }
    \end{displaymath}
    
  \item $\tensor\parr^\perp_p$. In this case, we have:
      \begin{displaymath}
        \infer[\tensor\parr^\perp_p]
        {
          \syncseqbis \cdot
          {
            \Gamma, \quadQ yAxB, [\Delta_1]\Delta_2, \{\phl{x_i} : \thl{A_i \parr B_i}\}_i,
            \pairQ y{C
              \tensor D}
          }
        }
        {
          \syncseqbis {\cdot} {\Delta_1, \{\pairQ{x_i}{A_i}\}_i,  \pairQ yC}
          & \syncseqbis {\cdot} {\thl {\Gamma}, [\ast], \Delta_2, \{\phl{x_i}:\thl{B_i}\}_i,
            \phl
            y:\thl D}
        }
      \end{displaymath}
      By induction hypothesis, since our proof is $\perp$-free, by
      applying $\tensor\parr^\perp_p$ again, we obtain:
      \begin{displaymath}
        \infer[\tensor\parr^\perp_p]
        {
          \syncseqbis \cdot
          {
            \Gamma, \pairQ x{\perp}, [\Delta_1]\Delta_2, \{\phl{x_i} : \thl{A_i \parr B_i}\}_i,
            \pairQ y{C
              \tensor D}
          }
        }
        {
          \syncseqbis {\cdot} {\Delta_1, \{\pairQ{x_i}{A_i}\}_i,  \pairQ yC}
          & \syncseqbis {\cdot} {\thl {\Gamma}, \pairQ x{\perp}, \Delta_2, \{\phl{x_i}:\thl{B_i}\}_i,
            \phl
            y:\thl D}
        }
      \end{displaymath}

  \item $\parr$. In this case, we have:
    \begin{displaymath}
      \infer[\parr]{
        \syncseq{\recv xyP} {\Gamma, [\ast], \pairQ x {A \parr B}}
      }{
        \syncseq P {\Gamma, [\ast], \quadQ yAxB}
      }     
    \end{displaymath}
    And, by induction hypothesis, 
    \begin{displaymath}
      \infer[\parr]{
        \syncseq{\recv xyP} {\Gamma, \pairQ x\perp, \pairQ x {A \parr B}}
      }{
        \syncseq P {\Gamma, \pairQ x\perp, \quadQ yAxB}
      }     
    \end{displaymath}
    
  \item $\with$. In this case, we have:
    \begin{displaymath}
      \infer[\with]
      {\syncseq{\Case x{P'}{Q'}} {\Gamma, \quadQ yAxB, \oplus\Delta,\pairQ x{A \with B}}
      }{
    	\syncseq {P'} {\Gamma, [\ast], \ltriCtx {\oplus\Delta}x{A}}
    	&
    	\syncseq {Q'} {\Gamma, [\ast], \rtriCtx {\oplus\Delta}x{B}}
      }
    \end{displaymath}
    By induction hypothesis, since both proofs in the premise of
    $\with$ are $\perp$-free, by applying $\with$ again, we obtain:
    \begin{displaymath}
      \infer[\with]
      {\syncseq{\Case x{P'}{Q'}} {\Gamma, \pairQ x\perp, \oplus\Delta,\pairQ x{A \with B}}
      }{
    	\syncseq {P'} {\Gamma, \pairQ x\perp, \ltriCtx {\oplus\Delta}x{A}}
    	&
    	\syncseq {Q'} {\Gamma, \pairQ x\perp, \rtriCtx {\oplus\Delta}x{B}}
      }
    \end{displaymath}

  \item $\oplus_l$ (similar for $\oplus_r$). In this case, we have:
    \begin{displaymath}
      \infer[\oplus_l]
      {\syncseq {\inl w{P'}} {\Gamma,  [\ast], \ltriCtx{\Delta, \phl{w}:\thl{A\oplus B}}{z}{C} }}
      {\syncseq {P'} {\Gamma,  [\ast], \ltriCtx{{\Delta}}{z}{C}, \pairQ{w}{A}}}
    \end{displaymath}
    By induction hypothesis, since $\mathcal D$ in
    $\mathcal D:: \phl {P'}$ is $\perp$-free, by
    applying $\oplus_l$ again, we obtain:
    \begin{displaymath}
      \infer[\oplus_l]
      {\syncseq {\inl w{Q'}} {\Gamma,  \pairQ x{\perp}, \ltriCtx{\Delta, \phl{w}:\thl{A\oplus B}}{z}{C} }}
      {\syncseq {Q'} {\Gamma,  \pairQ x{\perp}, \ltriCtx{{\Delta}}{z}{C}, \pairQ{w}{A}}}
    \end{displaymath}

  \item $\query$. In this case, we have:
    \begin{displaymath}
      \infer[?]{
    	\syncseq {\client xyP }{\Gamma, [\ast], \Query[\![\Delta, \phl x:\thl{\query A}]\!]\phl z:\thl{C}} 
      }{
    	\syncseq {P}{\Gamma, [\ast], \Query[\![\Delta]\!]\phl z:\thl{C}, \phl y:\thl A}
      }
    \end{displaymath}
    By induction hypothesis, since $\mathcal D$ in
    $\mathcal D:: \phl {P'}$ is $\perp$-free, by
    applying $\query$ again, we obtain:
    \begin{displaymath}
      \infer[?]{
    	\syncseq {\client xyP }{\Gamma, \pairQ x{\perp},
          \Query[\![\Delta, \phl x:\thl{\query A}]\!]\phl z:\thl{C}}
      }{
    	\syncseq {P}{\Gamma, \pairQ x{\perp}, \Query[\![\Delta]\!]\phl z:\thl{C}, \phl y:\thl A}
      }
    \end{displaymath}

  \end{itemize}
\end{proof}

%


\end{document}